\pdfoutput=1
\documentclass{tlp}

\usepackage{stmaryrd, amssymb, extarrows}
\usepackage{verbatim} % for \begin{comment} ... \end{comment}
\usepackage{fancyvrb}
\usepackage{hyperref}

\hypersetup{% PDF INFO
  pdftitle = {A Transformation-based Implementation for CLP with Qualification and Proximity},
  pdfkeywords = {Constraint Logic Programming, Program Transformation, Qualification Domains and Values, Similarity and Proximity Relations, Flexible Information Retrieval},
  pdfauthor = {R. Caballero, M. Rodríguez-Artalejo and C. A. Romero-Díaz}
}

\newcommand{\nc}{\newcommand}

\nc{\prox}{$\sim$}
\nc{\SL}{\mathcal{T}} % Lattice T used by Subrahmanian
\nc{\GL}{\mathcal{G}} % Multi-adoint Lattice L corresponding to [0,1] with fuzzy & and => in Goedel's sense

\nc{\trans}[2]{S_{#1}(#2)} % Translation of a  $\sqlp{\simrel}{\qdom}$ $\Prog$ into a $\qlp{\simrel}{\qdom}$ program $\trans{\simrel}{\Prog}$

\nc{\extended}[2]{H_{#1}(#2)} % Extended Program, e.g. $\extended{\lambda}{\Prog}$
\nc{\abstracted}[2]{{#1}_{#2}} % Abstracted Program, e.g. $\abstracted{\Prog}{\lambda}$

\nc{\MAProg}{\mathcal{P}_{E, \simrel}}

\nc{\linear}[1]{#1_\ell}

\nc{\Set}{\mathcal{S}}
\nc{\transform}[1]{#1^{\mathcal{T}}}

\nc{\transQ}[1]{\transform{#1}}   % transformed of a QCLP program. Eg. $\transQ{\Prog}$
\nc{\transRen}[1]{\widehat{#1}}       % copy of the program with renamed predicates. Eg. $\transRen{\Prog}$
\nc{\transSim}[1]{{#1}_{S}}       % Part of the program containing the clauses obtained by similarities. Eg. $\transSim{\Prog}$
\nc{\transUni}[1]{{#1}_{\sim}}    % Program clauses for unification. Eg. $\transUni{\Prog}$
\nc{\transPay}[1]{{#1}_{\mathit{pay}}}    % Program clauses for paying for using similarities. Eg. $\transPay{\Prog}$

\nc{\elimS}[1]{\mathrm{elim}_{\simrel}(#1)} % First transformation: SQCLP to QCLP
\nc{\elimD}[1]{\mathrm{elim}_{\qdom}(#1)} % Second transformation: QCLP to CLP

\nc{\Con}[1]{\mbox{Con}_{#1}}  % Examples: $\Con{\cdom}$, $\Con{\rdom}$

\nc{\pc}[3]{\mathsf{#1}_{#2}(#3)} % #1(#2)
\nc{\qval}[1]{\pc{qVal}{}{#1}} % $\qval(#1)$
\nc{\qbound}[1]{\pc{qBound}{}{#1}} % $\qbound(#1)$
\nc{\qvali}[2]{\pc{qVal}{#1}{#2}} % $\qval(#1)$
\nc{\qboundi}[2]{\pc{qBound}{#1}{#2}} % $\qbound(#1)$

\nc{\encode}[1]{\ulcorner#1\urcorner} % e.g. $\encode{W \dleq \alpha \circ W'}$.

\nc{\spair}[2]{\llparenthesis\, #1, #2 \,\rrparenthesis}

\newcommand{\schemenp}[1]{\mbox{#1}} % Scheme #1 (without parameters)
\newcommand{\scheme}[2]{\schemenp{#1}(#2)} % Scheme(#1) over #2

\newcommand{\clp}[1]{\scheme{CLP}{#1}} % CLP over C; e.g. $\clp{\cdom}$
\newcommand{\qclp}[2]{\scheme{QCLP}{#1,#2}} % QCFLP over D and C; e.g. $\qclp{\qdom}{\cdom}$
\newcommand{\sqclp}[3]{\scheme{SQCLP}{#1,#2,#3}} % SQCLP; e.g. $\sqclp{\simrel}{\qdom}{\cdom}$
\newcommand{\cdom}{\mathcal{C}} % Default C domain
\newcommand{\rdom}{\mathcal{R}} % Constraint domain for real arithmetic constraints
\newcommand{\qdom}{\mathcal{D}} % Default D domain
\newcommand{\aqdomd}[1]{D_{#1} \setminus \{\bt\}} % D_{#1}\{bot}
\newcommand{\bqdomd}[1]{(D_{#1} \setminus \{\bt\}) \uplus \{?\}} % (D_{#1}\{bot}) u {?}
\newcommand{\aqdom}{\aqdomd{}} % D\{bot}
\newcommand{\bqdom}{\bqdomd{}} % (D\{bot}) u {?}

\newcommand{\bt}{\text{\bf{b}}} % bottom element
\newcommand{\tp}{\text{\bf{t}}} % top element
\newcommand{\dleq}{\trianglelefteqslant} % lower or equal
\newcommand{\dlt}{\vartriangleleft} % lower than
\newcommand{\dgeq}{\trianglerighteqslant} % greater or equal
\newcommand{\B}{\mathcal{B}} % Domain B
\newcommand{\U}{\mathcal{U}} % Domain U
\newcommand{\W}{\mathcal{W}} % Domain W
\newcommand{\simrel}{\mathcal{S}} % Default simrel
\newcommand{\sid}{\simrel_{\mathrm{id}}} % identity relation
\newcommand{\Prog}{\mathcal{P}} % Program P
\newcommand{\Var}{\mathcal V\!ar} % Variables
\newcommand{\War}{\mathcal W\!ar} % Qualification Variables
\newcommand{\set}[2]{\mathrm{#1}(#2)} % Generic set, e.g. '\set{dom}{\theta}' for 'dom(\theta)
\newcommand{\domset}[1]{\set{dom}{#1}} % dom(#1)
\newcommand{\varset}[1]{\set{var}{#1}} % var(#1)
\newcommand{\warset}[1]{\set{war}{#1}} % war(#1)
\nc{\Exp}{\mbox{Exp}_{\bot}(\Sigma,B,\Var)} % Set of all expressions
\nc{\TExp}{\mbox{Exp}(\Sigma,B,\Var)} % set of all total expressions
\nc{\GExp}{\mbox{Exp}_{\bot}(\Sigma,B)} % Set of all ground expressions
\nc{\TGExp}{\mbox{Exp}(\Sigma,B)} % Set of all total ground expressions
\nc{\Term}{\mbox{Term}_{\bot}(\Sigma,B,\Var)} % Set of all terms
\nc{\TTerm}{\mbox{Term}(\Sigma,B,\Var)} % Set of all total terms
\nc{\GTerm}{\mbox{Term}_{\bot}(\Sigma,B)} % Set of all ground terms
\nc{\TGTerm}{\mbox{Term}(\Sigma,B)} % Set of all total ground terms

\nc{\At}{\mbox{At}(\Sigma,B,\Var)} % Set of all atoms
\nc{\GAt}{\mbox{GAt}(\Sigma,B)} % Set of all ground atoms
\nc{\PAt}{\mbox{PAt}(\Sigma,B,\Var)} % Set of all primitive atoms
\nc{\GPAt}{\mbox{GPAt}(\Sigma,B)} % Set of all ground primitive atoms

\nc{\Atz}{\mbox{At}_{\Sigma}} % Open Herbrand Universe
\nc{\QAtz}{\mbox{At}_{\Sigma}(\qdom)} % Qualified Open Herbrand Unvierse

\nc{\sust}{\mbox{Subst}_\Sigma} % Set of all substitutions
\nc{\Sust}{\mbox{Subst}(\Sigma,B,\Var)} % Set of all substitutions, showing $\Sigma$, $B$ and $\Var$
\nc{\GSust}{\mbox{GSubst}(\Sigma,B)} % Set of all ground substitutions, showing $\Sigma$, $B$ and $\Var$

\nc{\Soln}[3]{\mbox{Sol}_{#1}^{#2}(#3)} % Set of all Solutions obtained in #2 steps
\nc{\Sol}[2]{\Soln{#1}{}{#2}} % Set of all Solutions
\nc{\GSol}[2]{\mbox{GSol}_{#1}(#2)}
\nc{\Solc}[1]{\Sol{\cdom}{#1}}
\nc{\CAns}[2]{\mbox{C\!Ans}_{#1}(#2)} % Set of all Computed Answers

\newcommand{\qat}[2]{#1 \sharp #2} % q-atom; e.g. $\qat{A}{d}$
\newcommand{\cat}[2]{#1 \Leftarrow #2} % c-atom; e.g. $\cat{A}{\Pi}$
\newcommand{\cqat}[3]{\qat{#1}{#2} \Leftarrow #3} % cq-atom; e.g. $\cqat{A}{d}{\Pi}$

\newcommand{\qgets}[1]{\xleftarrow{#1}} % Weigthed arrow for clauses

\nc{\closure}[1]{\mbox{cl}_{#1}} % Closure of a set: cl_#1(#2)

\newcommand{\model}[1]{~{\models_{#1}}~} % models, e.g. C-model: \model{\cdom}

\newcommand{\M}[1]{\mathcal{M}_{#1}} % Least model of P, parameterized version; e.g. $\M{\Prog}$.
\newcommand{\Mp}{\M{\Prog}} % Least model of P

\newcommand{\eqdef}{~{=_{\mathrm{def}}}~} % =_def
\newcommand{\infi}{\bigsqcap} % Infimum
\newcommand{\entail}[1]{~{\succcurlyeq_{#1}}~} % entails, e.g. C-entail: \entail{\cdom}

\newcommand{\sep}{\talloblong} % Separator

\newcommand{\NAT}{\mathbb{N}}
\newcommand{\REAL}{\mathbb{R}}

\newcommand{\tup}[1]{\overline{#1}}   % e.g. $\tup{X}$, $\tup{t}$
\newcommand{\ntup}[2]{\tup{#1}_{#2}}  % e.g. $\ntup{X}{m}$, $\ntup{t}{n}$
\newcommand{\infx}[2]{\ {\vdash}_{\!#1}^{\!#2}\ } % Inference step symbol for programs
\newcommand{\CHL}{\mbox{CHL}} % Name w/o params. Params within parentheses; e.g. \CHL(...)
\newcommand{\chln}[2]{\infx{#1}{#2}} % CHL over #1 in #2 steps; e.g. \chl{\cdom}{n}
\newcommand{\chlc}{\chln{\cdom}{}} % CHL over C
\newcommand{\chlcn}[1]{\chln{\cdom}{#1}} % CHL over C in #1 steps; e.g. \chlcn{n}

\newcommand{\QCHL}{\mbox{QCHL}} % Name w/o params. Params wihtin patentheses; e.g. \QCHL(...)
\newcommand{\qchln}[3]{\infx{#1,#2}{#3}} % QCHL over #1,#2 in #3 steps; e.g. \qchl{\qdom}{\cdom}{n}
\newcommand{\qchldc}{\qchln{\qdom}{\cdom}{}} % QCHL over D,C
\newcommand{\qchldcn}[1]{\qchln{\qdom}{\cdom}{#1}} % QCHL over D,C in #1 steps; e.g. \qchldcn{n}

\newcommand{\SQCHL}{\mbox{SQCHL}} % Name w/o params. Params wihtin patentheses; e.g. \SQCHL(...)
\newcommand{\sqchln}[4]{\infx{#1,#2,#3}{#4}} % SQCHL over #1,#2,#3 in #4 steps; e.g. \sqchl{\simrel}{\qdom}{\cdom}{n}
\newcommand{\sqchlrdc}{\sqchln{\simrel}{\qdom}{\cdom}{}} % QCHL inference over D,C
\newcommand{\sqchlrdcn}[1]{\sqchln{\simrel}{\qdom}{\cdom}{#1}} % e.g. \qchldcn{n}

\submitted{26 July 2010}
\revised{19 March 2011, 7 January 2012}
\accepted{12 January 2012}

\title[A Transformation-based Implementation for SQCLP]
    {A Transformation-based Implementation\\ for CLP with Qualification and Proximity \thanks{This work has been partially supported by the Spanish projects STAMP (TIN2008-06622-C03-01), PROMETIDOS--CM (S2009TIC-1465) and GPD--UCM (UCM--BSCH--GR58/08-910502).}}

  \author[R. Caballero, M. Rodr\'iguez-Artalejo and C. A. Romero-D\'iaz]
    {R. CABALLERO, M. RODR\'IGUEZ-ARTALEJO and C. A. ROMERO-D\'IAZ\\
    Departamento de Sistemas Inform\'aticos y Computaci\'on, Universidad Complutense\\
    Facultad de Inform\'atica, 28040 Madrid, Spain\\
    \email{\{rafa,mario\}@sip.ucm.es, cromdia@fdi.ucm.es}}

\begin{document}

\maketitle

\begin{abstract}
\begin{center}
To appear
in Theory and Practice of Logic Programming (TPLP) \vspace*{0.4cm}
\end{center}

Uncertainty in logic programming has been widely investigated in the last decades,
leading to multiple extensions of the classical LP paradigm.
However, few of these are designed as
extensions of the well-established and powerful CLP scheme for Constraint Logic Programming.
In a previous work we have proposed the SQCLP ({\em proximity-based qualified constraint logic programming}) scheme as a
quite expressive extension of CLP with support for qualification values and proximity relations
as generalizations of uncertainty values and similarity relations, respectively.
In this paper we provide a transformation technique for transforming SQCLP programs and goals into semantically equivalent CLP programs and goals, and a practical Prolog-based implementation of some particularly useful instances of the SQCLP scheme.
We also illustrate, by showing some simple---and working---examples, how the prototype can be effectively used as a tool for solving problems where qualification values and proximity relations play a key role.
Intended use of SQCLP includes flexible information retrieval applications.

\end{abstract}


% DEFAULT ENVIRONMENTS
%\begin{defn}Definition.\qed\end{defn}
\newtheorem{defn}{Definition}[section]
%\begin{thm}Theorem.\qed\end{thm}
\newtheorem{thm}{Theorem}[section]
%\begin{lem}Lemma.\qed\end{lem}
\newtheorem{lem}{Lemma}[section]
%\begin{prop}Proposition.\qed\end{prop}
\newtheorem{prop}{Proposition}[section]
%\begin{pf}Proof.\end{pf}
%\begin{exmp}Example.\qed\end{exmp}
\newtheorem{exmp}{Example}[section]
%\begin{cor}Corolary.\end{cor}
\newtheorem{cor}{Corollary}[section]

% Section 0: Provisional
%\input{J0} % \ref{sec:provisional}

% Section 1: Introduction
% ----
% Section 1: Introduction
% ----
\section{Introduction}
\label{sec:introduction}

% Introduction and (quick overview of) related CLP extensions:

Many extensions of LP ({\em logic programming}) to deal with uncertain knowledge and uncertainty have been proposed in the last decades.
These extensions have been proposed from different and somewhat unrelated perspectives, leading to multiple approaches in the way of using uncertain knowledge and understanding uncertainty.

% About this paper:

A recent work by us \cite{RR10} focuses on the declarative semantics of a new proposal for an extension of the CLP scheme supporting qualification values and proximity relations.
More specifically, this work defines a new generic scheme SQCLP ({\em proximity-based qualified constraint logic programming}) whose instances $\sqclp{\simrel}{\qdom}{\cdom}$ are parameterized by a proximity relation $\simrel$, a qualification domain $\qdom$ and a constraint domain $\cdom$.
The current paper is intended as a continuation of \cite{RR10} with the aim of providing a semantically correct program transformation technique that allows us to implement a sound and complete implementation of some useful instances of SQCLP on top of existing CLP systems like {\em SICStus Prolog} \cite{sicstus} or {\em SWI-Prolog} \cite{swipl}.
In the introductory section of \cite{RR10} we have already summarized some related approaches of SQCLP with a special emphasis on their declarative semantics and their main semantic differences with SQCLP.
In the next paragraphs we present a similar overview but, this time, putting the emphasis on the goal resolution procedures and system implementation techniques, when available.

% Annotated LP

Within the extensions of LP using annotations in program clauses we can find the seminal proposal of {\em quantitative logic programming} by \cite{VE86} that inspired later works such as the GAP ({\em generalized annotated programs}) framework by \cite{KS92} and our former scheme QLP ({\em qualified logic programming}).
In the proposal of van Emden, one can find a primitive goal solving procedure based on and/or trees (these are similar to the alpha-beta trees used in game theory), used to prune the search space when proving some specific ground atom for some certainty value in the real interval $[0,1]$.
In the case of GAP, the goal solving procedure uses constrained SLD resolution in conjunction with a---costly---computation of so-called {\em reductants} between variants of program clauses.
In contrast, QLP goal solving uses a more efficient resolution procedure called SLD($\qdom$) resolution, implemented by means of real domain constraints, used to compute the qualification value of the head atom based on the attenuation factor of the program clause and the previously computed qualification values of the body atoms.
Admittedly, the gain in efficiency of SLD($\qdom$) w.r.t. GAP's goal solving procedure is possible because QLP focuses on a more specialized class of annotated programs.
While in all these three approaches there are some results of soundness and completeness, the results for the QLP scheme are the stronger ones (again, thanks to its also more focused scope w.r.t. GAP).

% Fuzzy LP vs. Proximity relations

From a different viewpoint, extensions of LP supporting uncertainty can be roughly classified into two major lines: approaches based on fuzzy logic \cite{Zad65,Haj98,Ger01} and approaches based on similarity relations.
Historically, Fuzzy LP languages were motivated by expert knowledge representation applications.
Early Fuzzy LP languages implementing the resolution principle introduced in \cite{Lee72} include Prolog-Elf \cite{IK85}, Fril Prolog \cite{BMP95} and F-Prolog \cite{LL90}.  More recent approaches such as the Fuzzy LP languages in \cite{Voj01,GMV04} and Multi-Adjoint LP (MALP for short) in the sense of \cite{MOV01a} use clause annotations and a fuzzy interpretation of the connectives
and aggregation operators occurring in program clauses and goals.
The Fuzzy Prolog system proposed in \cite{GMV04} is implemented by means of real constrains on top of a CLP($\rdom$) system, using a syntactic expansion of the source code during the Prolog compilation.
A complete procedural semantics for MALP using reductants has been presented in \cite{MOV01b}.
A method for translating a MALP like program into standard Prolog has been described in \cite{JMP09}.

% Similarity-based LP

The second line of research mentioned in the previous paragraph
was motivated by applications in the field of flexible query answering.
Classical LP is extended to Similar\-i\-ty-based LP (SLP for short),
leading to languages which keep the classical syntax of LP clauses but use a similarity relation over a set of symbols $S$ to allow ``flexible'' unification of syntactically different symbols with a certain approximation degree.
Similarity relations over a given set $S$ have been defined in \cite{Zad71,Ses02} and related literature as fuzzy relations represented by mappings $\simrel : S \times S \to [0,1]$ which satisfy reflexivity, symmetry and transitivity axioms analogous to those required for classical equivalence relations.
Resolution with flexible unification can be used as a sound and complete goal solving procedure for SLP languages as shown e.g. in \cite{AF02,Ses02}.
SLP languages include {\em Likelog} \cite{AF99,Arc02} and more recently {\em SiLog} \cite{LSS04},
which has been implemented by means of an extended Prolog interpreter and proposed as a useful tool for web knowledge discovery.

% Extensions of SLP: SQLP, Proximity-based LP

In the last years, the SLP approach has been extended in various ways. The SQLP ({\em similarity-based qualified logic programming}) scheme proposed in \cite{CRR08} extended SLP by allowing program clause annotations in QLP style and generalizing similarity relations to mappings $\simrel : S \times S \to D$ taking values in a qualification domain not necessarily identical to the real interval $[0,1]$. As implementation technique for SQLP,
\cite{CRR08} proposed a semantically correct program transformation into QLP, whose goal solving procedure has been described above.
Other related works on transformation-based implementations of SLP languages include \cite{Ses01,MOV04}.
More recently, the SLP approach has been generalized to work with {\em proximity relations} in the sense of \cite{DP80} represented by mappings $\simrel : S \times S \to [0,1]$ which satisfy reflexivity and symmetry axioms but do not always satisfy transitivity.
SLP like languages using proximity relations include {\sf Bousi$\sim$Prolog} \cite{JR09} and the SQCLP scheme \cite{RR10}.
Two prototype implementations of {\sf Bousi$\sim$Prolog} are available:
a low-level implementation \cite{JR09b} based on an adaptation of the classical WAM (called {\em Similarity WAM}) implemented in {\sc Java} and able to execute a Prolog program in the context of a similarity relation defined on the first order alphabet induced by that program;
and a high-level implementation \cite{JRG09} done on top of {\em SWI-Prolog} by means of a program transformation from {\sf Bousi$\sim$Prolog} programs into a so-called {\em Translated BPL code} than can be executed according to the weak SLD resolution principle by a meta-interpreter.

% Summary of related approaches (in CLP):

Let us now refer to approaches related to constraint solving and CLP.
An analogy of proximity relations in the context of partial constraint satisfaction can be found in \cite{FW92},
where several metrics are proposed to measure the proximity between the solution sets of two different constraint satisfaction problems.
Moreover, some extensions of LP supporting uncertain reasoning use constraint solving as implementation technique,
as discussed in the previous paragraphs.
However, we are only aware of three approaches which have been conceived as extensions of the classical CLP scheme proposed for the first time in \cite{JL87}.
These three approaches are:
\cite{Rie98phd} that extends the formulation of CLP by \cite{HS88} with quantitative LP in the sense of \cite{VE86} and adapts van Emden's idea of and/or trees to obtain a goal resolution procedure;
\cite{BMR01} that proposes a semiring-based approach to CLP, where constraints are solved in a soft way with levels of consistency represented by values of the semiring, and is implemented with {\tt clp(FD,S)} for a particular class of semirings which enable to use local consistency algorithms,
as described  in \cite{GC98};
and the SQCLP scheme proposed in our previous work \cite{RR10}, which was designed as a common extension of SQLP and CLP.

% Summary of goals and results of this paper:

As we have already said at the beginning of this introduction, this paper deals with transformation-based implementations of the SQCLP scheme.
Our main results include: a) a transformation technique for transforming SQCLP programs into semantically equivalent CLP programs via two specific program transformations named elim$_\simrel$ and elim$_\qdom$; and b) and a practical Prolog-based implementation which relies on the aforementioned program transformations and supports several useful SQCLP instances.
As far as we know, no previous work has dealt with the implementation of extended LP languages for uncertain reasoning which are able to support clause annotations, proximity relations and CLP style programming.
In particular, our previous paper \cite{CRR08} only presented a transformation analogous to elim$_\simrel$ for a programming scheme less expressive than SQCLP,  which supported neither non-transitive proximity relations nor CLP programming.
Moreover, the transformation-based implementation reported in \cite{CRR08} was not implemented in a system.

% Note about recommended knowledge and organization of the paper

The reader is assumed to be familiar with the semantic foundations of LP  \cite{Llo87,Apt90} and CLP \cite{JL87,JMM+98}.
The rest of the paper is structured as follows:
Section \ref{sec:sqclp} gives an abridged presentation of the  SQCLP scheme and its declarative semantics,
followed by an abstract discussion of goal solving intended to serve as a theoretical guideline for practical implementations.
Section \ref{sec:qclpclp} briefly discusses two specializations of SQCLP, namely QCLP and CLP, which are used as the targets
of the program transformations elim$_\simrel$ and elim$_\qdom$, respectively.
Section \ref{sec:implemen} presents these two program transformations along with mathematical results which prove their semantic correctness, relying on the declarative semantics of the SQCLP, QCLP and CLP schemes.
Section \ref{sec:practical} presents a {\tt Prolog}-based prototype system that relies on  the transformations proposed in the previous section and implements several useful SQCLP instances.
Finally, Section \ref{sec:conclusions} summarizes conclusions and points to some lines of planned future research.

 %\ref{sec:introduction}

% Section 2: The Scheme SQCLP and its Declarative Semantics
% ----
% Section 2: The Scheme SQCLP and its Declarative Semantics
% ----
\section{The Scheme SQCLP and its Declarative Semantics}
\label{sec:sqclp}

In this section we first recall the essentials of the SQCLP scheme and its declarative semantics,
which were developed in detail in previous works \cite{RR10,RR10TR}.
Next we present an abstract discussion of goal solving intended to serve as a theoretical guideline for practical implementations of SQCLP instances.

% Constraint Domains
\subsection{Constraint Domains}
\label{sec:sqclp:cdom}

As in the CLP scheme, we will work with constraint domains related to signatures.
We assume an {\em universal programming signature} $\Gamma = \langle DC, DP
\rangle$ where $DC = \bigcup_{n \in \NAT}  DC^n$ and $DP = \bigcup_{n \in \NAT}  DP^n$ are
countably infinite and mutually disjoint sets of free function symbols (called {\em data constructors} in
the sequel) and {\em defined predicate} symbols, respectively, ranked by arities.
We will use {\em domain specific signatures}
$\Sigma = \langle DC, DP, PP \rangle$ extending $\Gamma$ with a disjoint set $PP = \bigcup_{n \in
\NAT}  PP^n$ of {\em primitive predicate} symbols, also ranked by arities. The idea is that
primitive predicates come along with constraint domains, while defined predicates are specified in
user programs. Each $PP^n$ may be any countable set of $n$-ary predicate symbols.

{\em Constraint domains} $\cdom$, sets of constraints $\Pi$ and their solutions, as well as terms, atoms and substitutions over a given $\cdom$ are well known notions underlying  the CLP scheme. In this paper we assume a relational formalization of constraint domains
as mathematical structures $\cdom$ providing a carrier set $C_\cdom$ (consisting of ground terms built from data constructors
and a given set $B_\cdom$ of $\cdom$-specific basic values) and an interpretation of various $\cdom$-specific primitive predicates.
For the examples in this paper we will use a constraint domain $\rdom$ which allows to work with arithmetic constraints over the real numbers,
and is defined to include:
\begin{itemize}
\item
The set of basic values $B_\rdom = \REAL$.
Note that $C_\rdom$ includes ground terms built from real values and data constructors, in addition to real numbers.
\item
Primitive predicates for encoding the usual arithmetic operations over $\REAL$.
For instance, the addition operation $+$ over $\REAL$  is encoded by a ternary primitive predicate $op_+$
such that, for any $t_1, t_2 \in C_\rdom$, $op_{+}(t_1,t_2,t)$  is true in $\rdom$ iff $t_1, t_2, t \in \REAL$ and $t_1 + t_2 = t$.
In particular, $op_{+}(t_1,t_2,t)$ is false in $\rdom$ if either $t_1$ or $t_2$ includes data constructors.
The primitive predicates encoding other arithmetic operations such as $\times$ and $-$ are defined analogously.
% In practice, atoms $op_{+}(t_1,t_2,t)$ will be written as $t_1 + t_2 = t$ in concrete examples,
% and similarly for other arithmetic operations such as $\times$, $-$, etc.
\item
Primitive predicates for encoding the usual inequality relations over $\REAL$.
For instance, the ordering $\leq$ over $\REAL$ is encoded by a binary primitive predicate $cp_{\leq}$
such that, for any $t_1, t_2 \in C_\rdom$, $cp_{\leq}(t_1,t_2)$ is true in $\rdom$ iff $t_1, t_2 \in \REAL$ and $t_1 \leq t_2$.
In particular,  $cp_{\leq}(t_1,t_2)$ is false in $\rdom$ if either $t_1$ or $t_2$ includes data constructors.
The primitive predicates encoding the other inequality relations, namely $>$, $\geq$ and $>$, are defined analogously.
% In practice, atoms $cp_{\leq}(t_1,t_2)$ will be written as $t_1 \leq t_2$ in concrete examples,
% and similarly for the other inequality relations, namely $>$, $\geq$ and $>$.
\end{itemize}
We assume the following classification of atomic $\cdom$-constraints: defined atomic constraints $p(\ntup{t}{n})$, where $p$ is a program-defined predicate symbol; primitive constraints $r(\ntup{t}{n})$ where $r$ is a $\cdom$-specific primitive predicate symbol; and equations $t == s$.

We use $\Con{\cdom}$ as a notation for the set of all $\cdom$-constraints and
$\kappa$ as a notation for an atomic primitive constraint.
Constraints are interpreted by means of {\em $\cdom$-valuations} $\eta \in \mbox{Val}_{\cdom}$,
which are ground substitutions.
The set $\Solc{\Pi}$ of solutions of $\Pi \subseteq \Con{\cdom}$ includes all the valuations $\eta$ such that $\Pi\eta$
is true when interpreted in $\cdom$.
$\Pi \subseteq \Con{\cdom}$ is called {\em satisfiable} if $\Solc{\Pi} \neq \emptyset$ and {\em unsatisfiable} otherwise.
$\pi \in \Con{\cdom}$ {\em is entailed} by $\Pi \subseteq \Con{\cdom}$
(noted  $\Pi \model{\cdom} \pi$) iff $\Solc{\Pi} \subseteq \Solc{\pi}$.

% Qualification Domains
\subsection{Qualification Domains}
\label{sec:sqclp:qdom}

\label{qd:explanations}
{\em Qualification domains} were inspired by \cite{VE86} and firstly introduced in \cite{RR08} with the aim of providing elements, called qualification values,  which can be attached to computed answers. They are  defined as structures $\qdom = \langle D, \dleq, \bt, \tp, \circ \rangle$ verifying the following requirements:
\begin{enumerate}
\item
$\langle D, \dleq, \bt, \tp \rangle$ is a lattice with extreme points $\bt$ (called {\em infimum} or {\em bottom} element) and $\tp$ (called {\em maximum} or {\em top} element) w.r.t. the partial ordering $\dleq$ (called {\em qualification ordering}). For given elements  $d, e \in D$, we  write $d \sqcap e$ for the {\em greatest lower bound} ($glb$) of $d$ and $e$, and $d \sqcup e$ for the {\em least upper bound} ($lub$) of $d$ and $e$. We also write $d \dlt e$ as abbreviation for $d \dleq e \land d \neq e$.
    \item $\circ : D \times D \rightarrow D$, called {\em attenuation operation}, verifies the following axioms:
        \begin{enumerate}
            \item $\circ$ is associative, commutative and monotonic w.r.t. $\dleq$.
            \item $\forall d \in D : d \circ \tp = d$ and $d \circ \bt = \bt$.
            \item $\forall d, e \in D  : d \circ e \dleq e$.
            % and even $\bt \neq d \circ e \dleq e$ if $d,e \in \aqdom$.
            \item $\forall d, e_1, e_2 \in D : d \circ (e_1 \sqcap e_2) = (d \circ e_1) \sqcap (d \circ e_2)$.
        \end{enumerate}
\end{enumerate}
For any $S = \{e_1, e_2, \ldots, e_n\} \subseteq D$, the $glb$ (also called {\em infimum} of $S$)
exists and can be computed as $\infi S =e_1 \sqcap e_2 \sqcap \cdots \sqcap e_n$
(which reduces to $\tp$ in the case $n = 0$).
The dual claim concerning $lub$s is also true.
As an easy consequence of the axioms, one gets the identity $d \circ \infi S =  \infi \{d \circ e \mid e \in S\}$.

% Qualification domains in relation to Gerla's Lower bound constraint frame and other related structures
%
%Observe that this notion of qualification domain is similar to that of {\em lower bound constraint frame} proposed by G. Gerla in \cite{Ger01}. However, this proposal of G. Gerla is based on
%a notion of closure system defined over complete lattices, while qualification domains
%are based on not necessarily complete lattices and fulfill some additional properties related to the attenuation operator.

\label{review2:response01}
Some of the axioms postulated for the attenuation operator---associativity, commutativity and monotonicity---are also required for t-norms in fuzzy logic, usually defined as binary operations over the real number interval $[0,1]$.
More generally, there are formal relationships between qualification domains and some other existing proposals of
lattice-based structures for uncertain reasoning,
such as the lower bound constraint frames proposed in \cite{Ger01},
the multi-adjoint lattices for fuzzy LP languages proposed in \cite{MOV01a,MOV01b} and the semirings for soft constraint solving proposed in \cite{BMR01,GC98}.
However, qualification domains are a class of mathematical structures that differs from all these approaches.
Their base lattices do not need to be complete and the axioms concerning the attenuation operator
require additional properties w.r.t. t-norms.
Some differences w.r.t. multi-adjoint algebras and the semirings from \cite{BMR01} have
been discussed in more detail in  \cite{CRR08} and \cite{RR10}, respectively.

% Stable qualification domains. Construction of qualification domains.

Many useful qualification domains are such that $\forall d, e \in D \setminus \{\bt\} : d \circ e \neq \bt$.
In the sequel, any qualification domain $\qdom$ that verifies this property will be called {\em stable}.
More technical details, explanations and examples concerning qualification domains can be found in  \cite{RR10TR}.
Examples include three basic qualification domains which are stable, namely:
the qualification domain $\B$ of classical boolean values,
the qualification domain $\U$ of uncertainty values,
the qualification domain $\W$ of weight values.
Moreover,  Theorem 2.1 of  \cite{RR10TR} shows that
the ordinary cartesian product  $\qdom_1 \!\times \qdom_2$ of two qualification domains is again a qualification domain,
while the strict cartesian product $\qdom_1 \!\otimes \qdom_2$ of two stable qualification domains is a stable qualification domain.

% Expressing a qualification domain in a constraint domain.
\subsection{Expressing a Qualification Domain in a Constraint Domain}
\label{sec:sqclp:qdomexpress}

The SQCLP scheme depends crucially on the ability to encode qualification domains into constraint domains, in the sense defined below:
\begin{defn} [Expressing $\qdom$ in $\cdom$]
\label{dfn:expressible}
A qualification domain $\qdom$ is expressible in a constraint domain $\cdom$
if there is an injective mapping $\imath : \aqdom \to C$ (thought as an embedding of $ \aqdom$ into $C$)
and moreover:
  \begin{enumerate}
     \item
     There is a $\cdom$-constraint $\qval{X}$ with free variable $X$ such that
     $\Solc{\qval{X}}$ is the set of all $\eta \in \mbox{Val}_\cdom$  verifying  $\eta(X) \in ran(\imath)$. \\
     \emph{Informal explanation:} For each qualification value $x \in \aqdom$ we think of $\imath(x) \in C$ as
     the representation of $x$ in $\cdom$.
     Therefore, $ran(\imath)$ is the set of those elements of $C$ which can be used
     to represent qualification values, and $\qval{X}$ constraints the value of $X$ to
     be some of these representations.
     \item
    There is a $\cdom$-constraint $\qbound{X,Y,Z}$ with free variables $X$, $Y$ and $Z$ encoding ``$x \dleq y \circ z$'' in the following sense:
    any $\eta \in \mbox{Val}_\cdom$ such that $\eta(X) = \iota(x)$, $\eta(Y) = \iota(y)$ and $\eta(Z) = \iota(z)$
    verifies $\eta \in \Solc{\qbound{X,Y,Z}}$ iff $x \dleq y \circ z$. \\
    \emph{Informal explanation:} $\qbound{X,Y,Z}$ constraints the values of $X, Y, Z$ to be the representations of
    three qualification values $x, y, z \in \aqdom$ such that $x \dleq y \circ z$.
  \end{enumerate}
In addition, if $\qval{X}$ and $\qbound{X,Y,Z}$ can be chosen as existential constraints
of the form $\exists X_1 \ldots \exists X_n(B_1 \land \ldots \land B_m)$---where $B_j ~ (1 \leq j \leq m)$ are atomic---we say that $\qdom$ is {\em existentially expressible} in $\cdom$. \mathproofbox
\end{defn}

It can be proved that $\B$, $\U$, $\W$ and and any qualification domain built from these with the help of the strict cartesian product $\otimes$ are existentially expressible in any constraint domain $\cdom$ that includes the basic values and computational features of $\rdom$. The example below
illustrates the existential representation of three typical qualification domains in $\rdom$:

% Example: representing qualification domains in $\rdom$

\begin{exmp}
\label{exmp:qdom-representation}
\begin{enumerate}
\item
$\U$ can be existentially expressed in $\rdom$ as follows:
$D_\U \setminus \{\bt\} = D_\U \setminus \{0\} = (0,1] \subseteq \REAL \subseteq C_\rdom$;
therefore $\imath$ can be taken as the identity embedding mapping from $(0,1]$ into $\REAL$.
Moreover,  $\qval{X}$ can be built as the existential $\rdom$-constraint $cp_<(0,X) \land cp_\leq(X,1)$
and $\qbound{X,Y,Z}$ can be built as the existential $\rdom$-constraint $\exists X' (op_\times(Y,Z,X') \land cp_\leq(X,X'))$.
\item
$\W$ can be existentially expressed in $\rdom$ as follows:
{$D_\W \setminus \{\bt\} = D_\W \setminus \{\infty\} = [0,\infty) \subseteq \REAL \subseteq C_\rdom$;
therefore $\imath$ can be taken as the identity embedding mapping from $[0,\infty)$ into $\REAL$.
Moreover,  $\qval{X}$ can be built as the existential $\rdom$-constraint $cp_\geq(X,0)$
and $\qbound{X,Y,Z}$ can be built as the existential $\rdom$-constraint $\exists X' (op_+(Y,Z,X') \land cp_\geq(X,X'))$.}
\item
$\U{\otimes}\W$ can be existentially expressed in $\rdom$ as follows:
$D_{\U{\otimes}\W} \setminus \{\bt\} =  (0,1] \times [0,\infty) \subseteq \REAL \times \REAL$;
therefore $\imath : D_{\U{\otimes}\W} \setminus \{\bt\} \to D_\rdom$ can bee defined as
$\imath(x,y)$ = {\sf pair}$(x,y)$, using a binary constructor {\sf pair} $\in DC^2$ to represent the ordered pair $(x,y)$
as an element of $D_\rdom$.
Moreover, taking into account the two previous items of the example:
\begin{itemize}
\item
$\qval{X}$ can be built as
$\exists X_1 \exists X_2 (X ==$ {\sf pair}$(X_1,X_2) \land  cp_<(0,X_1) \land cp_\leq(X_1,1) \land cp_\geq(X_2,0))$.
\item $\qbound{X,Y,Z}$ can be built as
$\exists X_1 \exists X'_1 \exists X_2 \exists X'_2 \exists Y_1 \exists Y_2 \exists Z_1 \exists Z_2
(X ==$ {\sf pair}$(X_1,X_2) \land Y ==$ {\sf pair}$(Y_1,Y_2) \land Z ==$ {\sf pair}$(Z_1,Z_2) \land
op_\times(Y_1,Z_1,X'_1) \land cp_\leq(X_1,X'_1) \land op_+(Y_2,Z_2,X'_2) \land cp_\geq(X_2,X'_2))$.
 \mathproofbox
\end{itemize}
\end{enumerate}
\end{exmp}

% Programs and Declarative Semantics.
\subsection{Programs and Declarative Semantics}
\label{sec:sqclp:programs}

% Amissible triples.

Instances $\sqclp{\simrel}{\qdom}{\cdom}$ of the SQCLP scheme are parameterized by so-called
{\em admissible triples} $\langle \simrel, \qdom, \cdom \rangle$ consisting of a constraint domain $\cdom$, a qualification domain $\qdom$ and a proximity relation $\simrel : S \times S \to D$---where $D$ is the carrier set of $\qdom$ and $S$ is the set of all variables, basic values  and signature symbols available in $\cdom$---satisfying the following properties:
\begin{itemize}
\item $\forall x\in S : \simrel(x,x) = \tp$ (reflexivity).
\item $\forall x,y\in S : \simrel(x,y) = \simrel(y,x)$ (symmetry).
%\item Some additional technical conditions explained in \cite{RR10TR}. % Reemplazado por:
\item $\simrel$ restricted to $\Var$ behaves as the identity --- i.e. $\simrel(X,X) = \tp$ for all $X \in \Var$ and $\simrel(X,Y) = \bt$ for all $X,Y \in \Var$ such that $X \neq Y$.
\item For any $x,y \in S$, $\simrel(x,y) \neq \bt$ can happen only if:
\begin{itemize}
\item $x = y$ are identical.
\item $x$ and $y$ are both: basic values; data constructor symbols with the same arity; or defined predicate symbols with the same arity.
\end{itemize}
In particular, $\simrel(p,p') \neq \bt$ cannot happen if $p$ and $p'$ are syntactically different primitive predicate symbols.
\end{itemize}
A proximity relation $\simrel$ is called {\em similarity} iff it satisfies the additional property $\forall x,y,z\in S : \simrel(x,z) \dgeq \simrel(x,y) \sqcap \simrel(y,z)$ (transitivity).
A given proximity relation $\simrel$ can be extended to work over terms, atoms and other syntactic objects in an obvious way.
%The definition for the case of terms and atoms can be found e.g. in  \cite{RR10TR},  pg. 10. % Reemplazado por:
The definition for the case of terms is as follows:
\begin{enumerate}
\item For any term $t$, $\simrel(t,t) = \tp$.
\item For $X\in\Var$ and for any term $t$ different from $X$, $\simrel(X,t) = \simrel(t,X) = \bt$.
\item For any two data constructor symbols $c$ and $c'$ with different arities, $\simrel(c(\ntup{t}{n}),$ $c'(\ntup{t'}{m})) = \bt$.
\item For any two data constructor symbols $c$ and $c'$ with the same arity, $\simrel(c(\ntup{t}{n}),$ $c'(\ntup{t'}{n})) = \simrel(c,c') \sqcap \simrel(t_1,t'_1) \sqcap \cdots \sqcap \simrel(t_n,t'_n)$.
\end{enumerate}
For the case of finite substitutions $\sigma$ and $\theta$ whose domain is a subset of a finite set of variables $\{X_1, \ldots, X_m\}$,
$\simrel(\sigma,\theta)$ can be naturally defined as
$\simrel(X_1\sigma,X_1\theta) \sqcap \ldots \sqcap \simrel(X_m\sigma,X_m\theta)$.

% Programs, clauses and their declarative meaning

A $\sqclp{\simrel}{\qdom}{\cdom}$-program is a set $\Prog$ of  \emph{qualified
program rules} (also called \emph{qualified clauses})
$C : A \qgets{\alpha} \qat{B_1}{w_1}, \ldots, \qat{B_m}{w_m}$, where $A$ is a defined atom,
$\alpha \in \aqdom$ is  called the {\em attenuation factor} of the clause and
each $\qat{B_j}{w_j} ~ (1 \le j \le m)$ is an atom $B_j$
annotated with a so-called {\em threshold value} $w_j \in \bqdom$.
The intended meaning of $C$ is as follows:
if for all $1 \leq j \leq m$ one has $\qat{B_j}{e_j}$ (meaning that $B_j$ holds with qualification value $e_j$)
for some $e_j \dgeq^? w_j$,
then $\qat{A}{d}$ (meaning that $A$ holds with qualification value $d$)
can be inferred for any $d \in \aqdom$ such that $d \dleq \alpha \circ \infi_{j = 1}^m e_j$.
By convention, $e_j \dgeq^? w_j$ means $e_j \dgeq w_j$ if $w_j ~{\neq}~?$ and is identically true otherwise.
In practice threshold values equal to `?' and attenuation values equal to $\tp$ can be omitted.

% BEGIN library figure
\begin{figure}[ht]
\figrule
\footnotesize\it
\renewcommand{\arraystretch}{1.4}
\begin{tabular}{rl}
& \% Book representation: book( ID, Title, Author, Lang, Genre, VocLvl, Pages ). \\
\tiny 1 & library([ book(1, `Tintin', `Herg\'e', french, comic, easy, 65), \\
\tiny 2 & $\quad$ book(2, `Dune', `F.P.~Herbert', english,\ sciFi,\ medium,\ 345), \\
\tiny 3 & $\quad$ book(3, `Kritik der reinen Vernunft', `I.~Kant', german, philosophy, difficult, 1011), \\
\tiny 4 & $\quad$ book(4, `Beim Hauten der Zwiebel', `G.~Grass', german, biography, medium, 432) ]) \\[3mm]
& \% Auxiliary predicate for computing list membership: \\
\tiny 5 & member(B, [B$\mid$\_]) \\
\tiny 6 & member(B, [\_$\mid$T]) $\gets$ member(B, T) \\[3mm]
& \% Predicates for getting the explicit attributes of a given book: \\
\tiny 7 & getId(book(ID, \_Title, \_Author, \_Lang, \_Genre, \_VocLvl, \_Pages), ID) \\
\tiny 8 & getTitle(book(\_ID, Title, \_Author, \_Lang, \_Genre, \_VocLvl, \_Pages), Title) \\
\tiny 9 & getAuthor(book(\_ID, \_Title, Author, \_Lang, \_Genre, \_VocLvl, \_Pages), Author) \\
\tiny 10 & getLanguage(book(\_ID, \_Title, \_Author, Lang, \_Genre, \_VocLvl, \_Pages), Lang) \\
\tiny 11 & getGenre(book(\_ID, \_Title, \_Author, \_Lang, Genre, \_VocLvl, \_Pages), Genre) \\
\tiny 12 & getVocLvl(book(\_ID, \_Title, \_Author, \_Lang, \_Genre, VocLvl, \_Pages), VocLvl) \\
\tiny 13 & getPages(book(\_ID, \_Title, \_Author, \_Lang, \_Genre, \_VocLvl, Pages), Pages) \\[3mm]
& \% Function for guessing the reader level of a given book: \\
\tiny 14 & guessRdrLvl(B, basic) $\gets$ getVocLvl(B, easy), getPages(B, N), N $<$ 50 \\
\tiny 15 & guessRdrLvl(B, intermediate) $\qgets{0.8}$ getVocLvl(B, easy), getPages(B, N), N $\ge$ 50 \\
\tiny 16 & guessRdrLvl(B, basic) $\qgets{0.9}$ getGenre(B, children) \\
\tiny 17 & guessRdrLvl(B, proficiency) $\qgets{0.9}$ getVocLvl(B, difficult), getPages(B, N), N $\ge$ 200 \\
\tiny 18 & guessRdrLvl(B, upper) $\qgets{0.8}$ getVocLvl(B, difficult), getPages(B, N), N $<$ 200 \\
\tiny 19 & guessRdrLvl(B, intermediate) $\qgets{0.8}$ getVocLvl(B, medium) \\
\tiny 20 & guessRdrLvl(B, upper) $\qgets{0.7}$ getVocLvl(B, medium) \\[3mm]
& \% Function for answering a particular kind of user queries: \\
\tiny 21 & search(Lang, Genre, Level, Id)  $\gets$ library(L)\#1.0, member(B, L)\#1.0, \\
\tiny 22 & $\quad$ getLanguage(B, Lang), getGenre(B, Genre), \\
\tiny 23 & $\quad$ guessRdrLvl(B, Level), getId(B, Id)\#1.0 \\[3mm]
& \% Proximity relation $\simrel_s$: \\
\tiny 24 & $\simrel_s$(sciFi, fantasy) = $\simrel_s$(fantasy, sciFi) = 0.9 \\
\tiny 25 & $\simrel_s$(adventure, fantasy) = $\simrel_s$(fantasy, adventure) = 0.7 \\
\tiny 26 & $\simrel_s$(essay, philosophy) = $\simrel_s$(philosophy, essay) = 0.8 \\
\tiny 27 & $\simrel_s$(essay, biography) = $\simrel_s$(biography, essay) = 0.7 \\
\end{tabular}
\normalfont
\caption{$\sqclp{\simrel_s}{\,\U}{\rdom}$-program $\Prog_{\!s}$ ({\em Library with books in different languages})}
\label{fig:library}
\figrule
\vspace*{-4mm}
\end{figure}
% END library figure

Figure \ref{fig:library} shows a simple $\sqclp{\simrel_s}{\,\U}{\rdom}$-program $\Prog_{\!s}$ which illustrates the expressivity of the SQCLP scheme to deal with problems involving flexible information retrieval. Predicate \mbox{\it search} can be used to answer queries asking for books in the library matching some desired language, genre and reader level.
Predicate \mbox{\it guessRdrLvl} takes advantage of attenuation factors to encode heuristic rules to compute reader levels on the basis of vocabulary level and other book features. The other predicates compute book features in the natural way,
and the proximity relation $\simrel_s$ allows flexibility in any unification (i.e. solving of equality constraints) arising during the invocation of the program predicates.

% qc-atoms and declarative semantics

The declarative semantics of a given $\sqclp{\simrel}{\qdom}{\cdom}$-program $\Prog$ relies on {\em qualified constrained atoms} (briefly {\em qc-atoms}) of the form $\cqat{A}{d}{\Pi}$, intended to assert that the validity of atom  $A$ with qualification degree $d \in D$ is entailed by the constraint set $\Pi$.
A qc-atom is called {\em defined}, {\em primitive} or {\em equational} according to the syntactic form of $A$; and it is called {\em observable} iff $d \in \aqdom$ and $\Pi$ is satisfiable.

Program interpretations are defined as sets of observable qc-atoms which obey a natural closure condition.
The results proved in \cite{RR10} show two equivalent ways to characterize declarative semantics: using a fix-point approach and a proof-theoretical approach.
For the purposes of the present paper it suffices to consider the proof-theoretical approach that relies on a formal inference system called {\em Proximity-based Qualified Constrained Horn Logic}---in symbols, $\SQCHL(\simrel,\qdom,\cdom)$---intended to infer observable qc-atoms from $\Prog$ and consisting of the three inference rules displayed in Figure \ref{fig:sqchl}. Rule {\bf SQEA} depends on a relation $\approx_{d,\Pi}$ between terms that is defined in the following way: $t \approx_{d,\Pi} s$ iff there exist two terms $\hat{t}$ and $\hat{s}$ such that $\Pi \model{\cdom} t == \hat{t}$, $\Pi \model{\cdom} s == \hat{s}$ and $\bt \neq d \dleq \simrel(\hat{t},\hat{s})$.
Recall that the notation $\Pi \model{\cdom} \pi$ makes sense for any $\cdom$-constraint $\pi$ and is a shorthand for $\Solc{\Pi} \subseteq \Solc{\pi}$, as explained in Subsection \ref{sec:sqclp:cdom}.
The relation $\approx_{d,\Pi}$ allows to deduce equations from $\Pi$ in a flexible way, i.e. taking the proximity relation $\simrel$ into account.
In the sequel we will use $t \approx_{d} s$ as a shorthand for $t \approx_{d,\emptyset} s$,
which holds iff $\bt \neq d \dleq \simrel(t,s)$.
% An equivalent proof-theoretic semantics

% Figure: SQCHL(R,D,C)
\begin{figure}[ht]
  \figrule
  \centering
  \vspace{2mm}
  \begin{tabular}{l}
  % SQDA : Proximity-based Qualified Defined Atom
  \textbf{SQDA} ~ $\displaystyle\frac
    {~ (~ \cqat{(t'_i == t_i\theta)}{d_i}{\Pi} ~)_{i=1 \ldots n} \quad (~ \cqat{B_j\theta}{e_j}{\Pi} ~)_{j=1 \ldots m} ~}
    {\cqat{p'(\ntup{t'}{n})}{d}{\Pi}}$ \\ \\
    \hspace{13mm} if $(p(\ntup{t}{n}) \qgets{\alpha} \qat{B_1}{w_1}, \ldots, \qat{B_m}{w_m}) \in \Prog$\!,~
        $\theta$ subst.,~ $\simrel(p',p) = d_0 \neq \bt$, \\
    \hspace{13mm} $e_j \dgeq^? w_j ~ (1 \le j \le m)$ and
        $d \dleq \bigsqcap_{i = 0}^{n}d_i \sqcap \alpha \circ \bigsqcap_{j = 1}^m e_j$. \\
  \\
  % SQEA : Proximity-based Qualified Equality Atom
  \textbf{SQEA} ~ $\displaystyle\frac
    {}
    {\quad \cqat{(t == s)}{d}{\Pi} \quad}$ ~
  if $t \approx_{d, \Pi} s$.
  % SQPA : Proximity-based Qualified Primitive Atom
  ~~ \textbf{SQPA} ~ $\displaystyle\frac
    {}
    {\quad \cqat{\kappa}{d}{\Pi} \quad}$ ~
  if $\Pi \model{\cdom} \kappa$. \\
  \end{tabular}
  \vspace{2mm}
  \caption{Proximity-based Qualified Constrained Horn Logic}
  \label{fig:sqchl}
  \figrule
  \vspace{-3mm}
\end{figure}
% End of figure.

We write $\Prog \sqchlrdc \varphi$ to indicate that $\varphi$ can be deduced from $\Prog$ in
$\SQCHL(\simrel,$ $\qdom,\cdom)$, and $\Prog \sqchlrdcn{k} \varphi$ in the case that the deduction can be performed with exactly $k$ {\bf SQDA} inference steps.
As usual in formal inference systems, $\SQCHL(\simrel,\qdom,\cdom)$ proofs can be represented as {\em proof trees} whose nodes correspond to qc-atoms, each node being inferred from its children by means of some $\SQCHL(\simrel,\qdom,\cdom)$ inference step.

The following theorem, proved in \cite{RR10TR}, characterizes least program models in the scheme SQCLP.
This result allows to use $\SQCHL(\simrel,\qdom,\cdom)$-derivability as a logical criterion for proving the semantic correctness of program transformations, as we will do in Section \ref{sec:implemen}.

% Theorem: SQCHL Least Models
\begin{thm}[Logical characterization of least program models in SQCHL]
\label{thm:SQCHL-leastmodel} For any $\sqclp{\simrel}{\qdom}{\cdom}$-program $\Prog$, its least
model can be characterized as: $$\Mp = \{\varphi \mid \varphi \mbox{ is an observable defined
qc-atom and }\Prog \sqchlrdc \varphi\} \mathproofbox$$
\end{thm}

%As an easy consequence of the previous theorem we can prove:
%
%% Cor: Correctness
%\begin{cor}[$\SQCHL(\simrel,\qdom,\cdom)$ is sound and complete]
%\label{cor:correctness} For any $\sqclp{\simrel}{\qdom}{\cdom}$-program $\Prog$ and any observable
%qc-atom $\varphi$ one has:
%\begin{enumerate}
%  \item $\Prog \sqchlrdc \varphi \Longleftrightarrow \Prog \model{\simrel,\qdom,\cdom} \varphi \Longleftrightarrow \Mp \isqchlrdc \varphi$.
%  \item $\Prog \sqchlrdc \varphi \Longrightarrow \Prog \model{\simrel,\qdom,\cdom} \varphi$ ({\em soundness}).
%  \item $\Prog \model{\simrel,\qdom,\cdom} \varphi \Longrightarrow \Prog \sqchlrdc \varphi$ ({\em completeness}). \mathproofbox
%\end{enumerate}
%\end{cor}

% Goals and Goal Solving.
\subsection{Goals and Goal Solving}
\label{sec:sqclp:goals}

Goals for a given $\sqclp{\simrel}{\qdom}{\cdom}$-program $\Prog$ have the form
$$G ~:~ \qat{A_1}{W_1},~ \ldots,~ \qat{A_m}{W_m} \sep W_1 \dgeq^? \!\beta_1,~ \ldots,~ W_m \dgeq^? \!\beta_m$$
abbreviated as $(\qat{A_i}{W_i},~ W_i \dgeq^? \!\beta_i)_{i = 1 \ldots m}$.
The $\qat{A_i}{W_i}$ are called {\em annotated atoms}.
If all atoms $A_i,\, i = 1 \ldots m,$ are equations $t_i == s_i$, the goal $G$ is called a {\em unification problem}.
The pairwise different variables $W_i \in \War$ are called qualification variables;
they are taken from a  set $\War$ assumed to be disjoint from the set $\Var$ of data variables used in terms.
The conditions $W_i \dgeq^? \!\beta_i$ (with $\beta_i \in \bqdom$)
are called {\em threshold conditions} and their intended meaning (relying on the notations `?' and `$\dgeq^?$') is as already explained when introducing program clauses in Subsection~\ref{sec:sqclp:programs}.
In the sequel, $\warset{o}$ will denote the set of all qualification variables occurring in the syntactic object $o$. In particular, for a goal $G$ as displayed above, $\warset{G}$ denotes the set $\{W_i \mid 1 \leq i \leq m\}$.
In the case $m = 1$ the goal is called  {\em atomic}.
The following definition relies on  $\SQCHL(\simrel,\qdom,\cdom)$-derivability to provide a natural
declarative notion of goal solution:

%the set of solutions of a goal $G$ w.r.t. program $\Prog$ is noted $\Sol{\Prog}{G}$ and consists of all triples $\langle \sigma, \mu, \Pi \rangle$ such that $\sigma$ is a $\cdom$-substitution (not required to be ground), $\mu : \{W_1, \ldots, W_m\} \to \aqdom$, $\Pi$ is a satisfiable and finite set of atomic $\cdom$-constraints and the following two conditions hold for all $i = 1 \ldots m$: $W_i\mu = d_i \dgeq^? \!\beta_i$ and $\Prog \sqchlrdc \cqat{A_i\sigma}{W_i\mu}{\Pi}$.
%Although operational semantics is not investigated in this paper, {\em computed answers} obtained by means of a correct goal solving system for $\sqclp{\simrel}{\qdom}{\cdom}$ are expected to be valid solutions in this sense.

\begin{defn}[Possible Answers and Goal Solutions]
\label{dfn:goalsol}
Assume a given $\sqclp{\simrel}{\qdom}{\cdom}$-program $\Prog$ and a goal $G$ for $\Prog$ with the syntax displayed above.
Then:
\begin{enumerate}
\item
A {\em possible answer} for $G$ is any triple $ans = \langle \sigma, \mu, \Pi \rangle$ such that $\sigma$ is a $\cdom$-substitution, $W\!\mu \in \aqdom$ for all $W \in \domset{\mu}$,  and $\Pi$ is a satisfiable and finite set of atomic $\cdom$-constraints.
The qualification value $\lambda_{ans} = \bigsqcap_{i = 1}^m W_i\mu$ is called the {\em qualification level} of  $ans$.
 \item
 A possible answer $\langle \sigma, \mu, \Pi \rangle$ is called a  {\em solution} for $G$ iff the conditions
$W_i\mu = d_i \dgeq^? \!\beta_i$ and $\Prog \sqchlrdc \cqat{A_i\sigma}{W_i\mu}{\Pi}$ hold for all $i = 1 \ldots m$.
Note that $\Prog \sqchlrdc \cqat{A_i\sigma}{W_i\mu}{\Pi}$ amounts to $t_i\sigma \approx_{W_i\mu, \Pi} s_i\sigma$  in the case that $A_i$  is an equation $t_i == s_i$.
The set of all solutions for $G$ w.r.t. $\Prog$ is noted $\Sol{\Prog}{G}$.
\item
A solution $\langle \eta, \rho, \Pi \rangle$ for $G$ is called {\em ground} iff $\Pi = \emptyset$ and $\eta \in \mbox{Val}_\cdom$ is a variable valuation such that $A_i\eta$ is a ground atom for all $i = 1 \ldots m$.
The set of all ground solutions for $G$ w.r.t. $\Prog$ is noted $\GSol{\Prog}{G} \subseteq \Sol{\Prog}{G}$.
\item
A ground solution $gsol = \langle \eta, \rho, \emptyset \rangle \in \GSol{\Prog}{G}$
is {\em subsumed} by a possible answer $ans = \langle \sigma, \mu, \Pi \rangle$
 iff $W_i\mu \dgeq W_i\rho$ for $i = 1 \ldots m$
 (which implies $\lambda_{ans} \dgeq \lambda_{gsol}$)
 and there is some $\nu \in \Solc{\Pi}$ such that
  $X\eta = X\sigma\nu$ holds for each variable $X \in \mathrm{var}(G)$.
\item
A ground solution $gsol = \langle \eta, \rho, \emptyset \rangle \in \GSol{\Prog}{G}$
is {\em subsumed} by a possible answer $ans = \langle \sigma, \mu, \Pi \rangle$ {\em in the flexible sense}
 iff  $\lambda_{ans} \dgeq \lambda_{gsol}$
 and there is some $\nu \in \Solc{\Pi}$ such that
 $\simrel(X\eta,X\sigma\nu) \dgeq \lambda_{gsol}$
 holds for each variable $X \in \mathrm{var}(G)$. \mathproofbox
\end{enumerate}
\end{defn}

%Implicitly, item (1) in the previous definition requires $\cqat{A_i\sigma}{W_i\mu}{\Pi}$ to be observable qc-atoms in the sense of Definition \ref{defn:atoms-entail}, which is trivially true because $W_i\mu = d_i \in \aqdom$ and $\Pi$ is satisfiable.
%In fact, Definition \ref{defn:atoms-entail} was designed with the aim of using observable qc-atoms as observations of valid solutions for atomic goals.

A possible goal $G_s$ for the library program displayed in Figure \ref{fig:library} is
\begin{center}
$G_s:$ \it search(german, essay, intermediate, ID)\#W\/ $\sep$ W $\ge$ 0.65
\end{center}
and one solution for $G_s$ is
$\langle \{\textit{ID} \mapsto 4\}, \{\textit{W} \mapsto 0.7\}, \emptyset \rangle$.
In this simple case, the constraint set $\Pi$ within the solution is empty.

The following example will be used to discuss some implementation issues
in Subsection \ref{sec:practical:SQCLP:eqsimrel}.
%Other examples of goal solutions can be found in \cite{RR10TR}.
% and Sections \ref{sec:implemen} and \ref{sec:practical} below.
\begin{exmp}
\label{exmp:eqs-implementation}
Assume the admissible triple $\langle \simrel, \U, \rdom \rangle$ where the proximity relation $\simrel$ is such that:
%\begin{itemize}
 $\simrel(a,b) = \simrel(b,a) = 0.9$,
 $\simrel(a,c) = \simrel(c,a) = 0.9$, and
 $\simrel(b,c) = \simrel(c,b) = 0.4$.
%\end{itemize}
Let $\Prog$ be the empty program.
Then, the goal $G$: \[\qat{(X==Y)}{W_1},\ \qat{(X==b)}{W_2},\ \qat{(Y==c)}{W_3} \sep W_1 \ge 0.8,\ W_2 \ge 0.8,\ W_3 \ge 0.8\] is a unification problem.
Its valid solutions in the sense of Definition \ref{dfn:goalsol}
include $\mbox{sol}_i = \langle \sigma_i, \mu_i, \emptyset\rangle$ ($i = 1,2,3$), where:
\[
\begin{array}{l@{\hspace{7.5mm}}l}
\sigma_1 = \{ X\mapsto a,\ Y\mapsto a \} & \mu_1 = \{ W_1\mapsto 1,\ W_2\mapsto 0.9,\ W_3\mapsto 0.9 \} \\
\sigma_2 = \{ X\mapsto b,\ Y\mapsto a \} & \mu_2 = \{ W_1\mapsto 0.9,\ W_2\mapsto 1,\ W_3\mapsto 0.9 \} \\
\sigma_3 = \{ X\mapsto a,\ Y\mapsto c \} & \mu_3 = \{ W_1\mapsto 0.9,\ W_2\mapsto 0.9,\ W_3\mapsto 1 \} \\
\end{array}
\]
as well as some less interesting  solutions assigning lower qualification values to the variables $W_i$ ($i = 1,2,3$).
In this simple example, all the solutions are ground, but this is not always the case in general.
Note that sol$_2$ is subsumed by sol$_1$ in the flexible sense because:
\begin{itemize}
\item $\nu = \varepsilon \in \Solc{\emptyset}$ satisfies $\simrel(X\sigma_2,X\sigma_1\varepsilon) = \simrel(b, a) = 0.9 \dgeq 0.9$ and also $\simrel(Y\sigma_2, Y\sigma_1\varepsilon) = \simrel(a, a) = 1 \dgeq 0.9$.
\item The qualification level of both sol$_2$ and sol$_1$ is $0.9$, thus trivially, $0.9 \dgeq 0.9$.
\end{itemize}
Moreover,  sol$_3$ is also subsumed by sol$_1$ in the flexible sense, because:
\begin{itemize}
\item $\nu = \varepsilon \in \Solc{\emptyset}$ satisfies $\simrel(X\sigma_3,X\sigma_1\varepsilon) = \simrel(a, a) = 1 \dgeq 0.9$ and also $\simrel(Y\sigma_3, Y\sigma_1\varepsilon) = \simrel(c, a) = 0.9 \dgeq 0.9$.
\item The qualification level of both sol$_3$ and sol$_1$ is $0.9$, thus trivially, $0.9 \dgeq 0.9$.
\end{itemize}
In fact, it is easy to check that any of the three ground solutions sol$_1$, sol$_2$ and sol$_3$ subsumes the other two in the flexible sense.
\mathproofbox
\end{exmp}

%% Goal:
%search(german,essay,intermediate,Id)#W :: W >= 0.65
%% Ejecutando...
%| ?- search(german,essay,intermediate,Id)#W :: W >= 0.65.
%Id = 4,
%{W=<0.7},
%{W>=0.65} ? ; (x20 veces)
%no

In practice, users of SQCLP languages will rely on  some available {\em goal solving system} for computing goal solutions.
The following definition provides an abstract specification of semantically correct goal solving systems
which will serve as a theoretical guideline for the implementation presented in Section \ref{sec:practical}:
\begin{defn}[Correct Abstract Goal Solving Systems for SQCLP]
\label{dfn:goalsolsys}
An {\em abstract goal solving system} for $\sqclp{\simrel}{\qdom}{\cdom}$ is any device $\mathcal{CA}$ that takes a program $\Prog$ and a goal $G$ as input and yields a set $\mathcal{CA}_{\Prog}(G)$ of possible answers $\langle \sigma, \mu, \Pi \rangle$ (called {\em computed answers}) as output.  Moreover:
\begin{enumerate}
\item
$\mathcal{CA}$ is called {\em sound} iff every computed answer is a solution,
i.e. $\mathcal{CA}_{\Prog}(G) \subseteq \mbox{Sol}_\Prog(G)$.
\item
$\mathcal{CA}$ is called {\em weakly complete} iff for every ground solution
$gsol \in \mbox{GSol}_\Prog(G)$
there is some computed answer $ans \in \mathcal{CA}_{\Prog}(G)$ such that
$ans$ subsumes $gsol$.
\item
$\mathcal{CA}$ is called  {\em weakly complete in the flexible sense} iff for every ground solution
$gsol \in \mbox{GSol}_\Prog(G)$
there is some computed answer $ans \in \mathcal{CA}_{\Prog}(G)$ such that
$ans$ subsumes $gsol$ in the flexible sense.
\item
$\mathcal{CA}$ is called {\em correct} iff it is both sound and weakly complete.
\item
$\mathcal{CA}$ is called {\em correct in the flexible sense} iff
 it is both sound and weakly complete in the flexible sense. \mathproofbox
\end{enumerate}
\end{defn}

%The following example illustrates subsumption in the flexible sense:

%\begin{exmp}
%\label{exmp:eqs-flexsubsum}
%Recall the unification problem $G$ and its solutions sol$_i$ $(i = 1,2,3)$
%from the previous Example \ref{exmp:eqs-implementation} above.
%Note that sol$_2$ is subsumed by sol$_1$ in the flexible sense because:
%\begin{itemize}
%\item $\nu = \varepsilon \in \Solc{\emptyset}$ satisfies $\simrel(X\sigma_2,X\sigma_1\varepsilon) = \simrel(b, a) = 0.9 \dgeq 0.9$ and also $\simrel(Y\sigma_2, Y\sigma_1\varepsilon) = \simrel(a, a) = 1 \dgeq 0.9$.
%\item The qualification level of both sol$_2$ and sol$_1$ is $0.9$, thus trivially, $0.9 \dgeq 0.9$.
%\end{itemize}
%Moreover,  sol$_3$ is also subsumed by sol$_1$ in the flexible sense, because:
%\begin{itemize}
%\item $\nu = \varepsilon \in \Solc{\emptyset}$ satisfies $\simrel(X\sigma_3,X\sigma_1\varepsilon) = \simrel(a, a) = 1 \dgeq 0.9$ and also $\simrel(Y\sigma_3, Y\sigma_1\varepsilon) = \simrel(c, a) = 0.9 \dgeq 0.9$.
%\item The qualification level of both sol$_3$ and sol$_1$ is $0.9$, thus trivially, $0.9 \dgeq 0.9$.
%\end{itemize}
%In fact, it is easy to check that any of the three ground solutions sol$_1$, sol$_2$ and sol$_3$ subsumes the other two in the flexible sense. Therefore, any goal solving system that computes at least one of these three solutions will satisfy weak completeness in the flexible sense concerning goal $G$.  \mathproofbox
%\end{exmp}

% Extended SLD resolution for uncertain LP languages.
Extensions of the well-known SLD-resolution procedure \cite{Llo87,Apt90} can be used as a basis to obtain correct goal solving systems for extended LP languages. In particular,  constraint SLD-resolution provides a correct goal solving system for instances of the CLP scheme, as proved e.g. in \cite{JMM+98}%
\footnote{In fact, constraint SLD-resolution is complete in a stronger sense than weak completeness. As proved in \cite{JMM+98},  every solution - even if it is not ground - is subsumed in a suitable sense by a finite set of computed solutions.}.
Several extensions of the SLD-resolution, tailored to different LP languages supporting uncertain reasoning, have already been mentioned in Section \ref{sec:introduction}.

% Our approach to goal solving in this paper.

Rather than developing an  extension of SLD resolution tailored to the SQCLP scheme,
our aim in this paper is to to investigate  goal solving systems based on a semantically correct program transformation from SQCLP into CLP.
Sections \ref{sec:implemen} and \ref{sec:practical} present the transformation technique and its implementation on top of a CLP {\tt Prolog} system, respectively.
As we will explain in Subsection \ref{sec:practical:SQCLP},
weak completeness as specified in Definition \ref{dfn:goalsolsys}(2) is very hard to achieve in a practical implementation,
while flexible weak completeness in the sense of Definition \ref{dfn:goalsolsys}(3)
is a satisfactory  notion for extended LP languages which use  proximity relations.
For instance, similarity-based SLD resolution as presented in \cite{Ses02} is complete in a flexible sense.
Therefore, the {\tt Prolog}-based prototype system presented in Section \ref{sec:practical}
aims at soundness and weak completeness in the flexible sense,
as specified in Definition \ref{dfn:goalsolsys}(3).
The definition and lemma below can be used as an abstract guideline for converting a correct goal solving system
$\mathcal{CA}$ into another goal solving system $\mathcal{FCA}$ which is correct in the flexible sense
and may be easier to implement, because it yields smaller sets of computed answers.

% Flexible Restrictions of an Abstract Goal Solving System
\begin{defn}[Flexible Restrictions of an Abstract Goal Solving System]
\label{dfn:flexrestr}
Let  $\mathcal{CA}$ and $\mathcal{FCA}$ be two abstract goal solving systems for $\sqclp{\simrel}{\qdom}{\cdom}$.
We say that $\mathcal{FCA}$ is a {\em flexible restriction} of $\mathcal{CA}$  iff the two following conditions hold
for any choice of a program $\Prog$ and a goal $G: (\qat{A_i}{W_i},~ W_i \dgeq^? \!\beta_i)_{i = 1 \ldots m}$:
\begin{enumerate}
\item
$\mathcal{FCA}_{\Prog}(G) \subseteq \mathcal{CA}_{\Prog}(G)$.
Informally, $\mathcal{FCA}$ is restricted to compute some of the answers computed by $\mathcal{CA}$.
\item
For each $ans = \langle \sigma, \mu, \Pi \rangle \in \mathcal{CA}_{\Prog}(G)$ there is
some $\widehat{ans} = \langle \hat{\sigma}, \hat{\mu}, \Pi \rangle \in \mathcal{FCA}_{\Prog}(G)$
such that
$\lambda_{\widehat{ans}} \dgeq \lambda_{ans}$
and  $\simrel(X\sigma,X\hat{\sigma}) \dgeq \lambda_{ans}$
holds for each variable $X \in \mathrm{var}(G)$.
Informally, each answer computed by $\mathcal{CA}$ is close (w.r.t. $\simrel$) to some
of the answers computed by $\mathcal{FCA}$.
 \mathproofbox
\end{enumerate}
\end{defn}

% Correctness of Flexible Restrictions.
\begin{lem}[Flexible Correctness of Flexible Restrictions]
\label{lema:flexrestr}
Let $\mathcal{CA}$ be a correct abstract goal solving system for $\sqclp{\simrel}{\qdom}{\cdom}$.
Then any flexible restriction $\mathcal{FCA}$ of  $\mathcal{CA}$ is correct in the flexible sense.
\end{lem}
\begin{proof*}
By assumption, $\mathcal{CA}$ is sound and weakly complete.
We must prove soundness and weak completeness in the flexible sense for $\mathcal{FCA}$.
Let a $\sqclp{\simrel}{\qdom}{\cdom}$-program $\Prog$ and a goal
$G: (\qat{A_i}{W_i},~ W_i \dgeq^? \!\beta_i)_{i = 1 \ldots m}$ for $\Prog$ be given.

\smallskip\noindent
--- \emph{Soundness.}
$\mathcal{FCA}_{\Prog}(G) \subseteq \mbox{Sol}_\Prog(G)$ trivially follows from $\mathcal{FCA}_{\Prog}(G) \subseteq \mathcal{CA}_{\Prog}(G)$ (true because $\mathcal{FCA}$ refines $\mathcal{CA}$) and $\mathcal{CA}_{\Prog}(G) \subseteq \mbox{Sol}_\Prog(G)$ (true because $\mathcal{CA}$ is sound).

\smallskip\noindent
--- \emph{Weak completeness in the flexible sense.}
In order to  check the conditions stated in Definition \ref{dfn:goalsolsys}(3), let $gsol =  \langle \eta, \rho, \emptyset \rangle \in \mbox{GSol}_\Prog(G)$ be given.
Since $\mathcal{CA}$ is weakly complete, there is some $ans = \langle \sigma, \mu, \Pi \rangle \in \mathcal{CA}_{\Prog}(G)$ that subsumes $gsol$ and hence:
\begin{enumerate}[$(f)$]
\item[$(a)$] $W_i\mu \dgeq W_i\rho$ for $i = 1 \ldots m$, which implies $\lambda_{ans} \dgeq \lambda_{gsol}$.
\item[$(b)$] There  is some $\nu \in \Solc{\Pi}$ such that  $X\eta = X\sigma\nu$ holds for all $X \in \mathrm{var}(G)$.
\end{enumerate}
Since $\mathcal{FCA}$ is a flexible refinement of $\mathcal{CA}$, there is some $\widehat{ans} = \langle \hat{\sigma}, \hat{\mu}, \Pi \rangle \in \mathcal{FCA}_{\Prog}(G)$ that is close to $ans$ and thus verifies:
\begin{enumerate}[$(f)$]
\item[$(c)$] $\lambda_{\widehat{ans}} \dgeq \lambda_{ans}$.
\item[$(d)$] $\simrel(X\sigma,X\hat{\sigma}) \dgeq \lambda_{ans}$ holds for all $X \in \mathrm{var}(G)$.
\end{enumerate}
Now we can claim:
\begin{enumerate}[$(f)$]
\item[$(e)$] $\lambda_{\widehat{ans}} \dgeq \lambda_{gsol}$ --- follows from $(c)$ and $(a)$.
\item[$(f)$] $\simrel(X\sigma\nu,X\hat{\sigma}\nu) \dgeq \lambda_{gsol}$ holds for all $X \in \mathrm{var}(G)$ --- follows from $(d)$ and $(a)$.
\item[$(g)$] $\simrel(X\eta,X\hat{\sigma}\nu) \dgeq \lambda_{gsol}$ holds for all $X \in \mathrm{var}(G)$ --- follows from $(f)$ and $(b)$.
\end{enumerate}
%
%\noindent (a)~
%$W_i\mu \dgeq W_i\rho$ for $i = 1 \ldots m$, which implies $\lambda_{ans} \dgeq \lambda_{gsol}$.\\
%\noindent (b)~
%There  is some $\nu \in \Solc{\Pi}$ such that  $X\eta = X\sigma\nu$ holds for all $X \in \mathrm{var}(G)$.
%
%\noindent
%Since $\mathcal{FCA}$ is a flexible refinement of $\mathcal{CA}$,
%there is some $\hat{ans} = \langle \hat{\sigma}, \hat{\mu}, \Pi \rangle %\in \mathcal{FCA}_{\Prog}(G)$
%that is close to $ans$ and thus verifies:
%
%\noindent (c)~
%$\lambda_{\hat{ans}} \dgeq \lambda_{ans}$. \\
%\noindent (d)~
%$\simrel(X\sigma,X\hat{\sigma}) \dgeq \lambda_{ans}$ holds for all $X \in \mathrm{var}(G)$.
%
%\noindent
%Now we can claim:
%
%\noindent (d)~
%$\lambda_{\hat{ans}} \dgeq \lambda_{gsol}$
%~(follows from (c) and (a)). \\
%\noindent (e)~
%$\simrel(X\sigma\nu,X\hat{\sigma}\nu) \dgeq \lambda_{gsol}$ holds for all $X \in \mathrm{var}(G)$
%~(follows from (d) and (a)).\\
%\noindent (f)~
%$\simrel(X\eta,X\sigma\nu) \dgeq \lambda_{gsol}$ holds for all $X \in \mathrm{var}(G)$
%~(follows from (e) and (b)). \\
%
Since $\nu \in \Solc{\Pi}$, (e) and (g) guarantee that $\widehat{ans}$ subsumes $gsol$ in the flexible sense. This finishes the proof. \mathproofbox
\end{proof*}

% Incompleteness of Sessa's flexible unification algorithm in presence of a non-transitive proximity relation

Let us finish this section with a remark concerning unification.
Both our implementation and SLD-based goal solving systems for SLP languages---we view \cite{AF02,Ses02} as representative proposals of this kind; others have been cited in Section \ref{sec:introduction}---must share the ability to solve unification problems modulo a  given proximity relation $\simrel : S \times S \to [0,1]$ over signature symbols, that is assumed to be transitive in \cite{Ses02} and some other related works,  but not in  {\sf Bousi$\sim$Prolog} \cite{JR09,JR09b} and our own approach.
The lack of transitivity makes a crucial difference. The unification algorithms modulo $\simrel$ known for the case that $\simrel$ is a similarity relation fail to be complete in the flexible sense if $\simrel$ is a non-transitive proximity relation. More details on this issue are given in Subsection \ref{sec:practical:SQCLP} when discussing the implementation of unification modulo $\simrel$ in our prototype system for SQCLP programming.

\section{The Schemes QCLP \& CLP as Specializations of SQCLP}
\label{sec:qclpclp}

As discussed in the concluding section of \cite{RR10},
several specializations of the SQCLP scheme can be obtained by partial instantiation of its parameters.
In particular, QCLP and CLP can be defined as schemes with instances:
$$
\begin{array}{r@{\hspace{1mm}}c@{\hspace{1mm}}l}
\qclp{\qdom}{\cdom} &\eqdef& \sqclp{\sid}{\qdom}{\cdom}\\
\clp{\cdom} &\eqdef& \sqclp{\sid}{\B}{\cdom} = \qclp{\B}{\cdom}\\
\end{array}
$$
with $\sid$  the {\em identity} proximity relation
and $\B$ the qualification domain including just the two classical boolean values.
As explained in the introduction, QCLP and CLP are the targets of two program transformations to
be developed in Section \ref{sec:implemen}.
In this  brief section we provide an explicit description of the syntax and semantics of these two schemes,
derived from their behaviour as specializations of SQCLP.

% Subsections
% ----
% Subsection 3.1: Presentation of the QCLP Scheme
% ----
\subsection{Presentation of the QCLP Scheme}
\label{sec:cases:qclp}

% QCLP instances

As already explained, the instances of QCLP can be defined by the equation QCLP($\qdom$,$\cdom$) = SQCLP($\sid$,$\qdom$,$\cdom$).
Due to the admissibility of the parameter triple $\langle \sid, \qdom, \cdom \rangle$, the qualification domain $\qdom$ must be (existentially) expressible in the constraint domain $\cdom$.
Technically, the QCLP scheme can be seen as a common extension of the classical CLP scheme for Constraint Logic Programming \cite{JL87,JMM+98} and the QLP scheme for Qualified Logic Programming originally introduced in \cite{RR08}.
Intuitively,  QCLP programming behaves like SQCLP programming, except that proximity information other than the identity is not available for proving equalities.

% QCLP programs

Program clauses and observable qc-atoms in QCLP are defined in the same way as in  SQCLP.
The library program $\Prog_{\!s}$ in Figure \ref{fig:library} becomes a $\qclp{\,\U}{\rdom}$-program $\Prog'_{\!s}$ just by replacing $\sid$ for $\simrel$.
Of course, $\Prog'_{\!s}$ does not support flexible unification as it was the case with $\Prog_{\!s}$.

% QCLP declarative semantics

As explained in Subsection \ref{sec:sqclp:programs}, the proof system consisting of the three displayed in Figure \ref{fig:sqchl} characterizes the declarative semantics of
a given $\sqclp{\simrel}{\qdom}{\cdom}$-program $\Prog$. In the particular case $\simrel = \sid$, the inference rules specialize to those  displayed in Figure \ref{fig:qchl}, yielding a formal proof system called \emph{Qualified Constrained Horn Logic}---in symbols, $\QCHL(\qdom,\cdom)$---which characterizes the declarative semantics of a given $\qclp{\qdom}{\cdom}$-program $\Prog$.
Note that rule {\bf SQEA} depends on a relation $\approx_{\Pi}$ between terms that is defined to behave the same as the specialization of $\approx_{d,\Pi}$ to the case $\simrel = \sid$.
It is easily checked that $t \approx_{\Pi} s$ does not depend on $d$ and holds iff $\Pi \model{\cdom} t == s$.
Both $\approx_{d,\Pi}$ and $\approx_{\Pi}$ allow to use the constraints within $\Pi$ when deducing equations. However, $c(\ntup{t}{n}) \approx_{\Pi} c'(\ntup{s}{n})$ never holds in the case that $c$ and $c'$ are not syntactically identical.

 \begin{figure}[ht]
  \figrule
  \centering
  \begin{tabular}{l}\\
  % QDA : Qualified Defined Atom
  \textbf{QDA} ~ $\displaystyle\frac
    {~ (~ \cqat{(t'_i == t_i\theta)}{d_i}{\Pi} ~)_{i = 1 \ldots n} \quad (~ \cqat{B_j\theta}{e_j}{\Pi} ~)_{j = 1 \ldots m} ~}
    {\cqat{p(\ntup{t'}{n})}{d}{\Pi}}$ \\ \\
  \hspace{11mm} if $(p(\ntup{t}{n}) \qgets{\alpha} \qat{B_1}{w_1}, \ldots, \qat{B_m}{w_m}) \in \Prog$, $\theta$ subst., \\
  \hspace{11mm} $e_j \dgeq^? w_j ~ (1 \le j \le m)$ and $d \dleq \bigsqcap_{i = 1}^{n}d_i \sqcap \alpha \circ \bigsqcap_{j = 1}^m e_j$.\\ 
  \\
  % QEA : Qualified Equality Atom
  \textbf{QEA} ~ $\displaystyle\frac
    {}
    {\quad \cqat{(t == s)}{d}{\Pi} \quad}$ ~
  if $t \approx_{\Pi} s$.
  % QPA : Qualified Primitive Atom
  \hspace{8mm} \textbf{QPA} ~ $\displaystyle\frac
    {}
    {\quad \cqat{\kappa}{d}{\Pi} \quad}$ ~
  if $\Pi \model{\cdom} \kappa$. \\
  \\
  \end{tabular}
  \caption{Qualified Constrained Horn Logic}
  \label{fig:qchl}
  \figrule
\end{figure}

$\SQCHL(\simrel,\qdom,\cdom)$ proof trees and the notations related to them can be naturally specialized to $\QCHL(\qdom,\cdom)$.
In particular, we will use the notation $\Prog \qchldc \varphi$ (resp. $\Prog \qchldcn{k} \varphi$)
to indicate that the qc-atom $\varphi$ can be inferred in $\QCHL(\qdom,\cdom)$ from the program $\Prog$
(resp. it can be inferred by using exactly $k$ \textbf{QDA} inference steps).
%The main properties of the logic $\SQCHL(\simrel,\qdom,\cdom)$
%(namely, Lemma \ref{lem:ep}, Theorem \ref{thm:SQCHL-leastmodel} and Corollary \ref{cor:correctness})
%can be also specialized to $\QCHL(\qdom,\cdom)$.
Clearly, Theorem \ref{thm:SQCHL-leastmodel} specializes to QCHL yielding the following result that is stated here for convenience:

% Theorem: QCHL Least Models
\begin{thm}[Logical characterization of least program models in QCHL]
\label{thm:QCHL-leastmodel} For any $\qclp{\qdom}{\cdom}$-program $\Prog$, its least model can be characterized as: $$\Mp = \{\varphi \mid \varphi \mbox{ is an observable defined qc-atom and }\Prog \qchldc \varphi\} \mathproofbox$$
\end{thm}

% QCLP goals and their solutions

Concerning goals and their solutions, their specialization to the particular case $\simrel = \sid$ leaves the syntax of goals $G$ unaffected and leads to the following definition, almost identical to Definition \ref{dfn:goalsol}:

\begin{defn}[Possible Answers and Goal Solutions in QCLP]
\label{dfn:qclp-goalsol} Assume a given $\qclp{\simrel}{\qdom}{\cdom}$-program $\Prog$ and a goal
$G: (~ \qat{A_i}{W_i}, W_i \dgeq^? \!\beta_i ~)_{i = 1 \ldots m}$. Then:
\begin{enumerate}
\item
{\em Possible answers} $ans = \langle \sigma, \mu, \Pi \rangle$ for $G$ and their qualification levels 
are defined as in SQCLP (see Definition \ref{dfn:goalsol}(1)).
 \item
 A {\em solution} for $G$ is any possible answer $\langle \sigma, \mu, \Pi \rangle$ 
that verifies the conditions in Definition \ref{dfn:goalsol}(2), 
 except that the requirement $\Prog \sqchlrdc \cqat{A_i\sigma}{W_i\mu}{\Pi}$
 used in Definition \ref{dfn:goalsol} for SQCLP becomes now
 $\Prog  \qchldc \cqat{A_i\sigma}{W_i\mu}{\Pi}$ for QCLP.
%  such that $\sigma$ is a $\cdom$-substitution, $W\mu \in \aqdom$  for all $W \in \domset{\mu}$,
% $\Pi$ is a satisfiable and finite set of atomic $\cdom$-constraints,
% and the following two conditions hold for all $i = 1 \ldots m$:
% $W_i\mu = d_i \dgeq^? \!\beta_i$ and $\Prog  \qchldc \cqat{A_i\sigma}{W_i\mu}{\Pi}$.
The set of all solutions for $G$ is noted $\Sol{\Prog}{G}$.
\item
The subset $\GSol{\Prog}{G} \subseteq \Sol{\Prog}{G}$ of all {\em ground} solutions is
defined exactly as in Definition \ref{dfn:goalsol}(3).
%A solution $\langle \eta, \rho, \Pi \rangle$ for $G$ is called {\em ground} iff $\Pi = \emptyset$ and
%$\eta \in \mbox{Val}_\cdom$ is a variable valuation such that $A_i\eta$ is a ground atom for all $i = 1 \ldots m$.
%The set of all ground solutions for $G$ is noted $\GSol{\Prog}{G} \subseteq \Sol{\Prog}{G}$.
%
\item
The subsumption relation between a ground solution $\langle \eta, \rho, \emptyset \rangle \in \GSol{\Prog}{G}$
and an arbitrary solution $\langle \sigma, \mu, \Pi \rangle$  is
defined exactly as in Definition \ref{dfn:goalsol}(4).
Subsumption in the flexible sense cannot be considered in QCLP due to the absence of a proximity relation.
%A ground solution $\langle \eta, \rho, \emptyset \rangle \in \GSol{\Prog}{G}$ is {\em subsumed} by
%$\langle \sigma, \mu, \Pi \rangle$ iff
%there is some $\nu \in \Solc{\Pi}$ s.t. $\eta =_{\varset{G}} \sigma\nu$ and $W_i\rho \dleq W_i\mu$ for $i = 1 \ldots m$.
\mathproofbox
\end{enumerate}
\end{defn}

% Correct abstract goal solving systems for QCLP

Finally, the notion  of correct abstract goal solving system for $\mbox{SQCLP}$ given in Definition
\ref{dfn:goalsolsys} specializes to $\mbox{QCLP}$ with only one minor modification: 
weak completeness in the flexible sense cannot be considered here,  due to the absence of a proximity relation. 
Therefore, we state the following definition: 

\begin{defn}[Correct Abstract Goal Solving Systems for QCLP]
\label{dfn:qclp-goalsolsys}
An {\em abstract goal solving system} for $\qclp{\qdom}{\cdom}$ is any device $\mathcal{CA}$ that takes a program $\Prog$ and a goal $G$ as input and yields a set $\mathcal{CA}_{\Prog}(G)$ of possible answers $\langle \sigma, \mu, \Pi \rangle$ 
(called {\em computed answers}) as output.  Moreover:
\begin{enumerate}
\item
$\mathcal{CA}$ is called {\em sound} iff every computed answer is a solution, 
i.e. $\mathcal{CA}_{\Prog}(G) \subseteq \mbox{Sol}_\Prog(G)$.
\item
$\mathcal{CA}$ is called {\em weakly complete} iff for every ground solution 
$gsol \in \mbox{GSol}_\Prog(G)$ 
there is some computed answer $ans \in \mathcal{CA}_{\Prog}(G)$ such that 
$ans$ subsumes $gsol$.
\item
$\mathcal{CA}$ is called {\em correct} iff it is both sound and weakly complete. \mathproofbox
\end{enumerate}
\end{defn}

 %\ref{sec:qclpclp:qclp}
% ----
% Subsection 3.2: Presentation of the CLP Scheme
% ----
\subsection{Presentation of the CLP Scheme}
\label{sec:cases:clp}

% CLP instances and programs

As already explained, the instances of CLP can be defined by the equation CLP($\cdom$) = SQCLP($\sid,\B,\cdom$), 
or equivalently, CLP($\cdom$) = QCLP($\B,\cdom$).
Due to the fixed  choice $\qdom = \B$, the only qualification value $d \in \aqdom$ available for use as attenuation factor or threshold value is $d = \tp$.
Therefore, CLP can only include threshold values equal to `?' and attenuation values equal to the top element $\tp = true$ of $\B$.
As explained in Section \ref{sec:sqclp}, such trivial threshold and attenuation values can be omitted,
and CLP clauses can  be written with the simplified syntax  $A \gets B_1, \ldots, B_m$.

% QCLP declarative semantics

Since $\tp = true$ is the only non-trivial qualification value available in CLP, qc-atoms $\cqat{A}{d}{\Pi}$ are always of the form $\cqat{A}{true}{\Pi}$ and can be written as $\cat{A}{\Pi}$.
Moreover, all the side conditions for the inference rule {\bf QDA} in Figure \ref{fig:qchl} become trivial when specialized to the case $\qdom = \B$. 
Therefore, the specialization of $\QCHL(\qdom,\cdom)$ to the case $\qdom = \B$ leads to the formal proof system called \emph{Constrained Horn Logic}---in symbols, $\CHL(\cdom)$---consisting of the three inference rules displayed in Figure \ref{fig:chl}, which characterizes the declarative semantics of a given $\clp{\cdom}$-program $\Prog$.

\begin{figure}[ht]
  \figrule
  \centering
  \begin{tabular}{l}\\
  % DA : Defined Atom
  \textbf{DA} ~ $\displaystyle\frac
    {~ (~ \cat{(t'_i == t_i\theta)}{\Pi} ~)_{i = 1 \ldots n} \quad (~ \cat{B_j\theta}{\Pi} ~)_{j = 1 \ldots m} ~}
    {\cat{p(\ntup{t'}{n})}{\Pi}}$ \\ \\
  \hspace{9mm} if $(p(\ntup{t}{n}) \gets B_1, \ldots, B_m) \in \Prog$ and $\theta$ subst. \\
   \\
  % EA : Equality Atom
  \textbf{EA} ~ $\displaystyle\frac
    {}
    {\quad \cat{(t == s)}{\Pi} \quad}$ ~
  if $t \approx_{\Pi} s$.
  % PA : Primitive Atom
  \hspace{10mm} \textbf{PA} ~ $\displaystyle\frac
    {}
    {\quad \cat{\kappa}{\Pi} \quad}$ ~
  if $\Pi \model{\cdom} \kappa$. \\
  \\
  \end{tabular}
  \caption{Constrained Horn Logic}
  \label{fig:chl}
  \figrule
\end{figure}

$\QCHL(\qdom,\cdom)$ proof trees and the notations related to them can be naturally specialized to $\CHL(\cdom)$.
In particular, we will use the notation $\Prog \chlc \varphi$ (resp. $\Prog \chlcn{k} \varphi$)
to indicate that the qc-atom $\varphi$ can be inferred in $\CHL(\cdom)$ from the program $\Prog$
(resp. it can be inferred by using exactly $k$ \textbf{DA} inference steps).
% The main properties of the logic $\SQCHL(\simrel,\qdom,\cdom)$
% (namely, Lemma \ref{lem:ep}, Theorem \ref{thm:SQCHL-leastmodel} and Corollary \ref{cor:correctness})
% can be also specialized to $\QCHL(\qdom,\cdom)$.
Clearly, Theorem \ref{thm:QCHL-leastmodel} specializes to CHL yielding the following result that is stated here for convenience:

% Theorem: CHL Least Models
\begin{thm}[Logical characterization of least program models in CHL]
\label{thm:CHL-leastmodel} For any $\clp{\cdom}$-program $\Prog$, its least model can be characterized as: $$\Mp = \{\varphi \mid \varphi \mbox{ is an observable defined qc-atom and }\Prog \chlc \varphi\} \mathproofbox$$
\end{thm}

%$\SQCHL(\simrel,\qdom,\cdom)$ proof trees and the notations related to them can be naturally specialized to $\CHL(\cdom)$.
%In particular, we will use the notation $\Prog \chlc \varphi$ (resp. $\Prog \chln{\cdom}{k} \varphi$) to indicate that the c-atom $\varphi$ can be inferred in $\CHL(\cdom)$ from the program $\Prog$ (resp. it can be inferred by using exactly $k$ \textbf{QDA} inference steps).
%The main properties of the logic $\SQCHL(\simrel,\qdom,\cdom)$ (namely, Lemma \ref{lem:ep}, Theorem \ref{thm:SQCHL-leastmodel} and Corollary \ref{cor:correctness}) can be also specialized to $\CHL(\cdom)$.
%In particular, the specialization of Theorem \ref{thm:SQCHL-leastmodel} ensures that any $\clp{\cdom}$-program $\Prog$ has a least model which can be characterized as
%$$\Mp = \{\varphi \mid \varphi \mbox{ is a defined observable c-atom and }\Prog \chlc \varphi\} \enspace .$$

% CLP goals and their solutions

Concerning goals and their solutions, their specialization to the scheme CLP leads to the following definition:

\begin{defn}[Goals and their Solutions in CLP]
\label{dfn:clp-goalsol}
Assume a given $\clp{\cdom}$-program $\Prog$.  Then:
\begin{enumerate}
\item
Goals for $\Prog$ have the form $G:\, A_1, \ldots, A_m$,
abbreviated as $(A_i)_{i = 1 \ldots m}$, where $A_i ~ (1 \leq i \leq m)$ are atoms.
\item
A {\em possible answer} for a goal $G$ is any pair $ans = \langle \sigma, \Pi \rangle$ such that 
$\sigma$ is a $\cdom$-substitution and $\Pi$ is a satisfiable and finite set of atomic $\cdom$-constraints.
 \item
 A possible answer $\langle \sigma, \Pi \rangle$ is called a {\em solution} for $G$ iff
  $\Prog \chlc \cat{A_i\sigma}{\Pi}$ holds for all $i = 1 \ldots m$.
The set of all solutions for $G$ is noted $\Sol{\Prog}{G}$.
\item
A solution $\langle \eta, \Pi \rangle$ for $G$ is called {\em ground} iff $\Pi = \emptyset$ and
$\eta \in \mbox{Val}_\cdom$ is a variable valuation such that $A_i\eta$ is a ground atom for all $i = 1 \ldots m$.
The set of all ground solutions for $G$ is noted $\GSol{\Prog}{G}$.
Obviously, $\GSol{\Prog}{G} \subseteq \Sol{\Prog}{G}$.
\item
A ground solution $\langle \eta, \emptyset \rangle \in \GSol{\Prog}{G}$ is {\em subsumed} by
$\langle \sigma, \Pi \rangle$ iff there is some $\nu \in \Solc{\Pi}$ s.t.
$\eta =_{\varset{G}} \sigma\nu$. \mathproofbox
\end{enumerate}
\end{defn}

The notion  of correct abstract goal solving system for SQCLP given in Definition \ref{dfn:qclp-goalsolsys}
specializes to CLP with a minor  change, namely: computed answers are pairs $\langle \sigma, \Pi \rangle$.
Formally, the definition for CLP is as follows:

\begin{defn}[Correct Abstract Goal Solving Systems for CLP]
\label{dfn:clp-goalsolsys}
A {\em goal solving system} for $\clp{\cdom}$ is any device $\mathcal{CA}$ that takes a program $\Prog$ and a goal $G$ as input and yields a set $\mathcal{CA}_{\Prog}(G)$ of possible answers $\langle \sigma, \Pi \rangle$ (called {\em computed answers}) as output.  Moreover, soundness, weak completeness and weak correctness of $\mathcal{CA}$ are defined exactly as in 
Definition \ref{dfn:qclp-goalsolsys}. \mathproofbox 
\end{defn}

% Technical lemma on existential constraints versus program defined predicates

We close this Subsection with a technical lemma that will be useful for proving some results in Subsection \ref{sec:implemen:QCLP2CLP}:

% Lemma
\begin{lem}
\label{lema:dec}
Assume an existential $\cdom$-constraint $\pi(\ntup{X}{n}) = \exists Y_1\ldots\exists Y_k(B_1 \land \ldots \land B_m)$
with free variables $\ntup{X}{n}$ and a given $\clp{\cdom}$-program $\Prog$ including the clause
$C:\, p(\ntup{X}{n}) \gets B_1, \ldots, B_m$, where $p \in DP^n$ does not occur at the head of any other clause of $\Prog$.
Then, for any n-tuple $\ntup{t}{n}$ of $\cdom$-terms and any finite and satisfiable $\Pi \subseteq \Con{\cdom}$,
one has: 
\begin{enumerate}
\item
$\Prog \chlc (\cat{p(\ntup{t}{n})}{\Pi}) \Longrightarrow \Pi \models_{\cdom} \pi(\ntup{t}{n})$,
where $\pi(\ntup{t}{n})$ stands for the result of applying the substitution $\{\ntup{X}{n} \mapsto \ntup{t}{n}\}$
to $\pi(\ntup{X}{n})$.
\item
The opposite implication
$\Pi \models_{\cdom} \pi(\ntup{t}{n}) \Longrightarrow  \Prog \chlc (\cat{p(\ntup{t}{n})}{\Pi})$
holds if $\ntup{t}{n}$ is a ground term tuple.
Note that for ground $\ntup{t}{n}$ the constraint entailment $\Pi \models_{\cdom} \pi(\ntup{t}{n})$
simply means that $\pi(\ntup{t}{n})$ is true in $\cdom$.
%
%\item 
%$\Pi \models_{\cdom} \pi(\ntup{t}{n}) \Longrightarrow  \Prog \chlc (\cat{p(\ntup{t}{n})}{\Pi})$
%may fail if $\ntup{t}{n}$ is not a ground term tuple.
%
\end{enumerate}
\end{lem}
\begin{proof*}
We prove each item separately:
\begin{enumerate}
% Item 1.
\item
Assume $\Prog \chlc (\cat{p(\ntup{t}{n})}{\Pi})$. 
Note that $C$ is the only clause for $p$ in $\Prog$ and that each atom $B_j$ in $C$'s body is an atomic constraint.
Therefore, the $\CHL(\cdom)$ proof must use a \textbf{DA} step based on an instance
$C\theta$ of clause $C$ such that $\Pi \models_{\cdom} t_i == X_i\theta$ holds for all $1 \leq i \leq n$ and 
$\Pi \models B_j\theta$ holds for all $1 \leq j \leq m$.
These conditions and the syntactic form of $\pi(\ntup{X}{n})$ obviously imply $\Pi \models_{\cdom} \pi(\ntup{t}{n})$.
% Item 2.
\item
Assume now $\Pi \models_{\cdom} \pi(\ntup{t}{n})$ and $\ntup{t}{n}$ ground. 
Then $\pi(\ntup{t}{n})$ is true in $\cdom$, and 
due to the syntactic form of $\pi(\ntup{X}{n})$, there must be some substitution $\theta$ such that
$X_i\theta = t_i$ (syntactic identity) for all $1 \leq i \leq n$ and
$B_j\theta$ is ground and true in $\cdom$ for all $1 \leq j \leq m$.
Trivially, $\Pi \models_{\cdom} t_i == X_i\theta$ holds for all  $1 \leq i \leq n$ and
$\Pi \models_{\cdom} B_j\theta$ also holds for all $1 \leq j \leq m$.
Then, it is obvious that $\Prog \chlc (\cat{p(\ntup{t}{n})}{\Pi})$ can be proved 
by using a  \textbf{DA} step based on the instance $C\theta$ of clause $C$. \mathproofbox
% Item 3.
%\item
%We prove that $\Pi \models_{\cdom} \pi(\ntup{t}{n}) \Longrightarrow  \Prog \chlc (\cat{p(\ntup{t}{n})}{\Pi})$
%can fail if $\ntup{t}{n}$ is not ground by presenting a counterexample based on the constraint domain $\rdom$, using the syntax for $\rdom$-constraints explained in \cite{RR10TR}.
%Consider the existential $\rdom$-constraint $\pi(X) = \exists Y(op_{+}(Y,Y,X))$,
%and a $\clp{\rdom}$-program $\Prog$ including the clause $C:\, p(X) \gets op_{+}(Y,Y,X)$ and no other 
%occurrence of the defined predicate symbol $p$. Consider also $\Pi = \{cp_{\geq}(X,0.0)\}$ and $t = X$.
%Then  $\Pi \models_{\rdom} \pi(X)$ is obviously true, because any real number $x \geq 0.0$ 
%satisfies $\exists Y(op_{+}(Y,Y,x))$ in $\rdom$.
%However, there is no $\rdom$-term $s$ such that $\Pi \models_{\rdom} op_{+}(s,s,X)$,
%and therefore there is no instance $C\theta$ of clause $C$ that can be used to prove $\Prog \chlc (\cat{p(X)}{\Pi})$
%by  applying a  \textbf{DA} step. \mathproofbox
\end{enumerate}
\end{proof*}

We remark that the second item of the previous lemma can fail if $\ntup{t}{n}$ is not ground.
This can be checked by presenting a counterexample based on the constraint domain $\rdom$, using the syntax for $\rdom$-constraints explained in \cite{RR10TR}.
Consider the existential $\rdom$-constraint $\pi(X) = \exists Y(op_{+}(Y,Y,X))$,
and a $\clp{\rdom}$-program $\Prog$ including the clause $C:\, p(X) \gets op_{+}(Y,Y,X)$ and no other 
occurrence of the defined predicate symbol $p$. Consider also $\Pi = \{cp_{\geq}(X,0.0)\}$ and $t = X$.
Then  $\Pi \models_{\rdom} \pi(X)$ is obviously true, because any real number $x \geq 0.0$ 
satisfies $\exists Y(op_{+}(Y,Y,x))$ in $\rdom$.
However, there is no $\rdom$-term $s$ such that $\Pi \models_{\rdom} op_{+}(s,s,X)$,
and therefore there is no instance $C\theta$ of clause $C$ that can be used to prove $\Prog \chlc (\cat{p(X)}{\Pi})$
by  applying a  \textbf{DA} step.

  %\ref{sec:qclpclp:clp}  %\ref{sec:qclpclp}

% Section 4: Implementation by Program Transformation
% ----
% Section 4: Implementation by Program Transformation
% ----
\section{Implementation by Program Transformation}
\label{sec:implemen}

The purpose of this section is to introduce a program transformation that transforms $\sqclp{\simrel}{\qdom}{\cdom}$ programs and goals into semantically equivalent $\clp{\cdom}$ programs and goals.
This transformation is performed as the composition of the two following specific transformations:
\begin{enumerate}
\item
elim$_\simrel$ --- Eliminates the proximity relation $\simrel$ of arbitrary $\sqclp{\simrel}{\qdom}{\cdom}$ programs and goals, producing equivalent $\qclp{\qdom}{\cdom}$ programs and goals.
\item
elim$_\qdom$ --- Eliminates the qualification domain $\qdom$ of arbitrary $\qclp{\qdom}{\cdom}$ programs and goals, producing equivalent $\clp{\cdom}$ programs and goals.
\end{enumerate}

Thus, given a $\sqclp{\simrel}{\qdom}{\cdom}$-program $\Prog$---resp. $\sqclp{\simrel}{\qdom}{\cdom}$-goal $G$---, the composition of the two transformations will produce an equivalent $\clp{\cdom}$-program $\elimD{\elimS{\Prog}}$---resp. $\clp{\cdom}$-goal $\elimD{\elimS{G}}$---.

\begin{exmp}[Running example: $\sqclp{\simrel_r}{\,\U{\otimes}\W}{\rdom}$-program $\Prog_r$]
\label{exmp:pr}
As a running example for this section, consider the $\sqclp{\simrel_r}{\,\U{\otimes}\W}{\rdom}$-program $\Prog_r$ as follows:
\begin{center}
\footnotesize\it
\renewcommand{\arraystretch}{1.4}
\begin{tabular}{rl}
\tiny $R_1$ & famous(sha) $\qgets{(0.9,1)}$ \\
\tiny $R_2$ & wrote(sha, kle) $\qgets{(1,1)}$ \\
\tiny $R_3$ & wrote(sha, hamlet) $\qgets{(1,1)}$ \\
\tiny $R_4$ & good\_work(G) $\qgets{(0.75,3)}$ famous(A)\#(0.5,100), authored(A, G) \\[3mm]
\tiny $S_1$ & $\simrel_r$(wrote, authored) = $\simrel_r$(authored, wrote) = (0.9,0)\\
\tiny $S_2$ & $\simrel_r$(kle, kli) = $\simrel_r$(kli, kle) = (0.8,2)\\
\end{tabular}
\end{center}
where the constants $shakespeare$, $king\_lear$ and $king\_liar$ have been respectively replaced, for clarity purposes in the subsequent examples, by $sha$, $kle$ and $kli$.

In addition, consider the $\sqclp{\simrel_r}{\,\U{\otimes}\W}{\rdom}$-goal $G_r$ as follows:
\begin{center}
\it good\_work(X)\#W $\sep$ \!W $\dgeq^?$ \!\!(0.5,10)
\end{center}

We will illustrate the two transformation by showing, in subsequent examples, the program clauses of $\elimS{\Prog_r}$ and $\elimD{\elimS{\Prog_r}}$ and the goals $\elimS{G_r}$ and $\elimD{\elimS{G_r}}$. \mathproofbox
\end{exmp}

In the following subsections we explain both transformations in detail and we show that they can be used to specify abstract goal solving systems for SQCLP.

% Subsections
% ----
% Subsection 4.1: Transforming SQCLP into QCLP
% ----

\subsection{Transforming SQCLP into QCLP}
\label{sec:implemen:SQCLP2QCLP}

In this subsection we assume that the triple $\langle \simrel,\qdom,\cdom \rangle$ is admissible.
In the sequel we say that a defined predicate symbol $p \in DP^n$ is {\em affected} by a $\sqclp{\simrel}{\qdom}{\cdom}$-program $\Prog$ iff $\simrel(p,p') \neq \bt$ for some $p'\!$ occurring in $\Prog$.
We also say that an atom $A$ is {\em relevant} for $\Prog$ iff some of the three following cases hold: a) $A$ is an equation $t == s$; b) $A$ is a primitive atom $\kappa$; or c) $A$ is a defined atom $p(\ntup{t}{n})$ such that $p$ is affected by $\Prog$.

As a first step towards the definition of the first program transformation elim$_\simrel$, we define a set $EQ_\simrel$ of $\qclp{\qdom}{\cdom}$ program clauses that emulates the behaviour of equations in $\sqclp{\simrel}{\qdom}{\cdom}$.
The following definition assumes that the binary predicate symbol $\sim\ \in DP^2$ (used in infix notation) and the nullary predicate symbols $\mbox{pay}_\lambda \in DP^0$ are not affected by $\Prog$\!.\smallskip

\begin{defn}\label{def:EQ}
We define $EQ_\simrel$ as the following $\qclp{\qdom}{\cdom}$-program:
$$
\begin{array}{l@{\hspace{0mm}}c@{\hspace{0mm}}l}
EQ_\simrel & \eqdef & \{~ X \sim Y \qgets{\tp} \qat{(X == Y)}{?} ~\} \\
&& ~\bigcup~ \{~ u \sim u' \qgets{\tp} \qat{\mbox{pay}_{\lambda}}{?} \mid u,u' \in B_\cdom \mbox{ and } \simrel(u,u') = \lambda \neq \bt ~\} \\
&& ~\bigcup~ \{~ c(\ntup{X}{n}) \sim c'(\ntup{Y}{n}) \qgets{\tp} \qat{\mbox{pay}_{\lambda}}{?},\ (\qat{(X_i \sim Y_i)}{?})_{i = 1 \ldots n} \mid c, c' \in DC^{n} \\
&& \qquad \mbox{and } \simrel(c,c') = \lambda \neq \bt ~\}\\
&& ~\bigcup~ \{~ \mbox{pay}_\lambda \qgets{\lambda} \ \mid \exists
%\mbox{ there exist some }
x,y \in S \mbox{ such that } \simrel(x,y)=\lambda\neq \bt ~\}. \mathproofbox \\
\end{array}
$$
\end{defn}

The following lemma shows the relation between the semantics of equations in $\SQCHL(\simrel,\qdom,\cdom)$
and the behaviour of the binary predicate symbol `$\sim$' defined by $EQ_\simrel$ in $\QCHL(\qdom,\cdom)$.

\begin{lem} \label{lema:equiv}
Consider any two arbitrary terms $t$ and $s$; $EQ_\simrel$ defined as in Definition \ref{def:EQ}; and a satisfiable finite set $\Pi$ of $\cdom$-constraints.
Then, for every $d \in \aqdom$:
$$t \approx_{d,\Pi} s \Longleftrightarrow EQ_\simrel \qchldc \cqat{(t \sim s)}{d}{\Pi}$$
\end{lem}
\begin{proof*}
We separately prove each implication.

% [==>]
\smallskip\noindent [$\Longrightarrow$]
Assume $t \approx_{d,\Pi} s$. Then, there are two terms $\hat{t}$, $\hat{s}$ such that:
$$
(1)~ t \approx_{\Pi} \hat{t} \qquad
(2)~ s \approx_{\Pi} \hat{s} \qquad
(3)~ \hat{t} \approx_{d} \hat{s}
$$
%\begin{enumerate}
%\item[(1)]  $t \approx_{\Pi} \hat{t}$, i.e. $\Pi \model{\cdom} t == \hat{t}$;
%\item[(2)]  $s \approx_{\Pi} \hat{s}$, i.e. $\Pi \model{\cdom} s == \hat{s}$; and
%\item[(3)]  $\hat{t} \approx_{d} \hat{s}$, i.e. $d \dleq \simrel(\hat{t},\hat{s})$.
%\end{enumerate}
We use structural induction on the form of the term $\hat{t}$.
\begin{itemize}
\item
$\hat{t} = Z$, $Z \in \Var$.
From (3) we have $\hat{s} = Z$. Then (1) and (2) become $t \approx_\Pi Z$ and $s \approx_\Pi Z$, therefore $t \approx_\Pi s$.
Now $EQ_\simrel \qchldc \cqat{(t \sim s)}{d}{\Pi}$ can be proved with a proof tree rooted by a {\bf QDA} step of the form:
$$
\displaystyle\frac
    {~ \cqat{(t == X\theta)}{\tp}{\Pi}\qquad \cqat{(s == Y\theta)}{\tp}{\Pi} \qquad \cqat{(X==Y)\theta}{\tp}{\Pi} ~}
    {\cqat{(t \sim s)}{d}{\Pi}}
$$
using the clause $X \sim Y \qgets{\tp} \qat{(X == Y)}{?} \in EQ_\simrel$ instantiated by the substitution $\theta = \{X \!\mapsto t,\ Y \!\mapsto s \}$. Therefore the three premises can be derived from $EQ_\simrel$ with {\bf QEA} steps since $t \approx_\Pi t$, $s \approx_\Pi s$ and $t \approx_\Pi s$, respectively. Checking the side conditions of all inference steps is straightforward.

\item $\hat{t} = u$, $u \in B_{\cdom}$.
From (3) we have $\hat{s} = u'$ for some $u' \in B_\cdom$ such that $d \dleq \lambda = \simrel(u,u')$.
Then (1) and (2) become  $t \approx_\Pi u$ and $s \approx_\Pi u'$, which allow to build a proof of $EQ_\simrel \qchldc \cqat{(t \sim s)}{d}{\Pi}$ by means of a {\bf QDA} step using the clause $u \sim u' \qgets{\tp} \qat{\mbox{pay}_{\lambda}}{?}$.

\item $\hat{t} = c$, $c \in DC^0$\!.
From (3) we have $\hat{s} = c'$ for some $c' \in DC^0$\! such that $d \dleq \lambda = \simrel(c,c')$.
Then (1) and (2) become  $t \approx_\Pi c$ and $s \approx_\Pi c'$\!, which allow us to build a proof of $EQ_\simrel \qchldc \cqat{(t \sim s)}{d}{\Pi}$ by means of a {\bf QDA} step using the clause $c \sim c' \qgets{\tp} \qat{\mbox{pay}_{\lambda}}{?}$.

\item $\hat{t} = c(\ntup{t}{n})$, $c \in DC^n$ with $n > 0$.
In this case, and because of (3), we can assume $\hat{s} = c'(\ntup{s}{n})$ for some $c' \in DC^n$ satisfying $d \dleq d_0 \eqdef \simrel(c,c')$ and $d \dleq d_i \eqdef \simrel(t_i,s_i)$ for $i=1 \dots n$.
Then $EQ_\simrel \qchldc \cqat{(t \sim s)}{d}{\Pi}$ with a proof tree rooted by a {\bf QDA} step of the form:
$$
\displaystyle\frac
  {~
    \begin{array}{l@{\hspace{1cm}}l}
      \cqat{(t == c(\ntup{t}{n}))}{\tp}{\Pi} & \cqat{\mbox{pay}_{d_0}}{d_0}{\Pi} \\
      \cqat{(s == c'(\ntup{s}{n}))}{\tp}{\Pi} & (~ \cqat{(t_i \sim s_i)}{d_i}{\Pi} ~)_{i = 1 \ldots n} \\
    \end{array}
   ~}
  {\cqat{(t \sim s)}{d}{\Pi}}
$$
using the $EQ_\simrel$ clause $C : c(\ntup{X}{n}) \sim c'(\ntup{Y}{n}) \qgets{\tp} \qat{\mbox{pay}_{d_0}}{?}, (\qat{(X_i \sim Y_i)}{?})_{i = 1 \ldots n}$ instantiated by the substitution $\theta = \{ X_1 \mapsto t_1,\ Y_1 \mapsto s_1,\ \dots,\ X_n \mapsto t_n,\ Y_n \mapsto s_n \}$.
Note that $C$ has attenuation factor $\tp$ and threshold values $?$ at the body.
Therefore, the side conditions of the {\bf QDA} step boil down to $d \dleq d_i ~ (1 \leq i \leq n)$ which are true by assumption.
It remains to prove that each premise of the {\bf QDA} step can be derived from $EQ_\simrel$ in QCHL($\qdom,\cdom$):
\begin{itemize}
\item $EQ_\simrel \qchldc \cqat{(t == c(\ntup{t}{n}))}{\tp}{\Pi}$ and $EQ_\simrel \qchldc \cqat{(s == c'(\ntup{s}{n}))}{\tp}{\Pi}$ are trivial consequences of $t \approx_\Pi c(\ntup{t}{n})$ and $s \approx_\Pi c'(\ntup{s}{n})$, respectively.
In both cases, the QCHL($\qdom$,$\cdom$) proofs consist of one single {\bf QEA} step.
\item $EQ_\simrel \qchldc  \cqat{\mbox{pay}_{d_0}}{d_0}{\Pi}$ can be proved using the clause $\mbox{pay}_{d_0} \!\qgets{d_0}\ \in EQ_\simrel$ in one single {\bf QDA} step.
\item $EQ_\simrel \qchldc \cqat{(t_i \sim s_i)}{d_i}{\Pi}$ for ${i = 1 \ldots n}$.
For each $i$, we observe that $t_i \approx_{d_i,\Pi} s_i$ holds because of $\hat{t}_i = t_i$, $\hat{s}_i = s_i$ which satisfy $t_i \approx_\Pi \hat{t}_i$, $s_i \approx_\Pi \hat{s}_i$ and $\hat{t}_i \approx_{d_i} \hat{s}_i$.
Since $\hat{t}_i = t_i$ is a subterm of $\hat{t} = c(\ntup{t}{n})$, the inductive hypothesis can be applied.
\end{itemize}
\end{itemize}

% [<==]
\smallskip\noindent [$\Longleftarrow$]
Let $T$ be a $\QCHL(\qdom,\cdom)$-proof tree witnessing $EQ_\simrel \qchldc \cqat{(t \sim s)}{d}{\Pi}$.
We prove $t \approx_{d,\Pi} s$ reasoning  by induction on the number $n = \Vert T \Vert$ of nodes in $T$ that represent conclusions of  \textbf{QDA} inference steps.
Note that all the program clauses belonging to $EQ_\simrel$ define either the binary predicate symbol `$\sim$' or the nullary predicates $\mbox{pay}_\lambda$.

\begin{description}
% Basis
\item[\bf Basis ($n = 1$).] \hfill

In this case we have for the {\bf QDA} inference step that there can be used three possible $EQ_\simrel$ clauses:
\begin{enumerate}
\item
The program clause is  $X \sim Y \qgets{\tp} \qat{(X == Y)}{?}$.
Then the {\bf QDA} inference step must be of the form:
$$
\displaystyle\frac
  {~ \cqat{(t == t')}{d_1}{\Pi} \quad  \cqat{(s == s')}{d_2}{\Pi} \quad \cqat{(t' == s')}{e_1}{\Pi} ~}
  {\cqat{(t \sim s)}{d}{\Pi}}
$$
with $d \dleq  d_1 \sqcap d_2 \sqcap e_1$.
The proof of the three premises must use the {\bf QEA} inference rule.
Because of the conditions of this inference rule we have $t \approx_\Pi t'$,  $s \approx_\Pi s'$ and $ t' \approx_\Pi s'$. Therefore $t \approx_\Pi s$ is clear. Then $t \approx_{d,\Pi} s$ holds by taking $\hat{t} = \hat{s} = t$ because, trivially, $t \approx_\Pi \hat{t}$, $s \approx_\Pi \hat{s}$ and $\hat{t} \approx_d \hat{s}$.
\item
The program clause is $u \sim u' \qgets{\tp} \qat{\mbox{pay}_{\lambda}}{?}$ with $u, u' \in B_\cdom$ such that $\simrel(u,u') = \lambda \neq \bt$.
The {\bf QDA} inference step must be of the form:
$$
\displaystyle\frac
    {~ \cqat{(t == u)}{d_1}{\Pi} \quad  \cqat{(s == u')}{d_2}{\Pi} \quad \cqat{\mbox{pay}_{\lambda}}{e_1}{\Pi}~ }
    {\cqat{(t \sim s)}{d}{\Pi}}
$$
with $d \dleq d_1 \sqcap d_2 \sqcap  e_1$.
Due to the forms of the {\bf QEA} inference rule and the $EQ_\simrel$ clause $\mbox{pay}_{\lambda} \qgets{\lambda}$, we can assume without loss of generality that $d_1 = d_2 = \tp$ and $e_1 = \lambda$.
Therefore $d \dleq \lambda$.
Moreover, the QCHL($\qdom$,$\cdom$) proofs of the first two premises must use {\bf QEA} inferences.
Consequently we have $t \approx_\Pi u$ and $s \approx_\Pi u'$.
These facts and $u \approx_d u'$ imply $t \approx_{d,\Pi} s$.
\item
The program clause is $c \sim c' \qgets{\tp} \qat{\mbox{pay}_{\lambda}}{?}$ with $c, c' \in DC^0$ such that $\simrel(c,c') = \lambda \neq \bt$.
The {\bf QDA} inference step must be of the form:
$$
\displaystyle\frac
    {~ \cqat{(t == c)}{d_1}{\Pi} \quad  \cqat{(s == c')}{d_2}{\Pi} \quad \cqat{\mbox{pay}_{\lambda}}{e_1}{\Pi} ~}
    {\cqat{(t \sim s)}{d}{\Pi}}
$$
with $d \dleq d_1 \sqcap d_2 \sqcap  e_1$. Due to the forms of the {\bf QEA} inference rule and the $EQ_\simrel$ clause $\mbox{pay}_{\lambda} \qgets{\lambda}$, we can assume without loss of generality that $d_1 = d_2 = \tp$ and $e_1 = \lambda$.
Therefore $d \dleq \lambda$.
Moreover, the QCHL($\qdom$,$\cdom$) proofs of the first two premises must use {\bf QEA} inferences.
Consequently we have $t \approx_\Pi c$ and $s \approx_\Pi c'$.
These facts and $c \approx_d c'$ imply $t \approx_{d,\Pi} s$.
\end{enumerate}

% Inductive case:
\item[\bf Inductive step ($n > 1$).] \hfill

In this case $t$ and $s$ must be of the form $t = c(\ntup{t}{n})$ and $s=c'(\ntup{s}{n})$.
The $EQ_\simrel$ clause used in the {\bf QDA} inference step at the root must be of the form:
$$c(\ntup{X}{n}) \sim c'(\ntup{Y}{n}) \qgets{\tp} \qat{\mbox{pay}_{d_0}}{?},\ (\qat{(X_i \sim Y_i)}{?})_{i = 1 \ldots n}$$
with $\simrel(c,c') = d_0 \neq \bt$. The inference step at the root will be:
$$
\displaystyle\frac
  {~
    \begin{array}{l@{\hspace{1cm}}l}
      \cqat{(t == c(\ntup{t}{n}))}{d_1}{\Pi} & \cqat{pay_{d_0}}{e_0}{\Pi} \\
      \cqat{(s == c'(\ntup{s}{n}))}{d_2}{\Pi} & (~ \cqat{(t_i \sim s_i)}{e_i}{\Pi} ~)_{i = 1 \ldots n} \\
    \end{array}
  ~}
  {\cqat{(t \sim s)}{d}{\Pi}}
$$
with $d \dleq d_1 \sqcap d_2 \sqcap  \bigsqcap_{i = 0}^n e_i$.
Due to the forms of the $EQ_\simrel$ clause $\mbox{pay}_{d_0} \!\qgets{d_0}$ and the {\bf QEA} inference rule there is no loss of generality in assuming $d_1 = d_2 = \tp$ and $e_0 = d_0$, therefore we have $d \dleq d_0 \sqcap \bigsqcap_{i = 1}^n e_i$.
By the inductive hypothesis $t_i  \approx_{e_i, \Pi} s_i ~ (1 \le i \le n)$, i.e. there are constructor terms $\hat{t}_i$, $\hat{s}_i$ such that $t_i \approx_\Pi \hat{t_i}$,  $s_i \approx_\Pi \hat{s}_i$ and $\hat{t}_i \approx_{e_i} \hat{s}_i$ for $i = 1 \ldots n$.
Thus, we can build $\hat{t} = c(\hat{t}_1, \ldots, \hat{t}_n)$ and $\hat{s} = c'(\hat{s}_1, \ldots, \hat{s}_n)$ having $t \approx_{d,\Pi} s$ because:
\begin{itemize}
\item
$t \approx_\Pi \hat{t}$, i.e. $c(\ntup{t}{n}) \approx_\Pi c(\ntup{\hat{t}}{n})$, by decomposition since $t_i \approx_\Pi \hat{t}_i$.
\item
$s \approx_\Pi \hat{s}$, i.e. $c'(\ntup{s}{n}) \approx_\Pi c'(\ntup{\hat{s}}{n})$, again by decomposition since $s_i \approx_\Pi \hat{s}_i$.
\item
$\hat{t} \approx_d \hat{s}$, since
$d \dleq  d_0 \sqcap \bigsqcap_{i = 1}^n e_i \dleq \simrel(c,c') \sqcap \bigsqcap_{i = 1}^n  \simrel(\hat{t}_i,\hat{s}_i) = \simrel(\hat{t},\hat{s}) \enspace .\mathproofbox
$
\end{itemize}
\end{description}
\end{proof*}

We are now ready to define elim$_\simrel$ acting over programs and goals.

\begin{defn}\label{def:sqclptransform}
Assume a $\sqclp{\simrel}{\qdom}{\cdom}$-program $\Prog$ and a $\sqclp{\simrel}{\qdom}{\cdom}$-goal $G$ for $\Prog$ whose atoms are all relevant for $\Prog$.
Then we define:
\begin{enumerate}
\item
For each atom $A$, let $A_\sim$ be $t \sim s$ if $A : t == s$; otherwise let $A_\sim$ be  $A$.
\item
For each clause $C : (p(\ntup{t}{n}) \qgets{\alpha} \tup{B}) \in \Prog$
let $\hat{\mathcal{C}}_\simrel$ be the set of $\qclp{\qdom}{\cdom}$ clauses consisting of:
\begin{itemize}
\item[---]
The clause $\hat{C} : (\transRen{p}_{C}(\ntup{t}{n}) \qgets{\alpha} \ntup{B}{\sim})$, where $\transRen{p}_{C} \in DP^n$ is not affected by $\Prog$ (chosen in a different way for each $C$) and $\ntup{B}{\sim}$ is obtained from $\tup{B}$ by replacing each atom $A$ occurring in $\tup{B}$ by $A_\sim$.
\item[---]
A clause $p'(\ntup{X}{n}) \qgets{\tp} \qat{\mbox{pay}_\lambda}{?},\ (\qat{(X_i \sim t_i)}{?})_{i = 1 \ldots n},\ \qat{ \transRen{p}_C(\ntup{t}{n})}{?}$ for each $p' \in DP^n$ such that $\simrel(p,p') = \lambda \neq \bt$.
Here, $\ntup{X}{n}$ must be chosen as $n$ pairwise different variables not occurring in the clause $C$.
\end{itemize}
\item
$\elimS{\Prog}$ is the $\qclp{\qdom}{\cdom}$-program $EQ_\simrel \cup \hat{\Prog}_\simrel$ where $\hat{\Prog}_\simrel \eqdef \bigcup_{C \in \Prog} \hat{\mathcal{C}}_\simrel$.
\item $\elimS{G}$  is the $\qclp{\qdom}{\cdom}$-goal $G_\sim$ obtained from $G$ by replacing each atom $A$ occurring in $G$ by $A_\sim$.
\mathproofbox
\end{enumerate}
\end{defn}

The following example illustrates the transformation elim$_\simrel$.

\begin{exmp}[Running example: $\qclp{\U\!\otimes\!\W}{\,\rdom}$-program $\elimS{\Prog_r}$]
\label{exmp:elims-pr}
Consider the $\sqclp{\simrel_r}{\,\U{\otimes}\W}{\rdom}$-program $\Prog_r$ and the goal $G_r$ for $\Prog_r$ as presented in Example \ref{exmp:pr}.
The transformed $\qclp{\U{\otimes}\W}{\rdom}$-program $\elimS{\Prog_r}$ is as follows:
\begin{center}
\footnotesize\it
\renewcommand{\arraystretch}{1.6}
\begin{tabular}{rl}
\tiny $\hat{R}_1$ & \^{f}amous$_{R_1}$(sha) $\qgets{(0.9,1)}$ \\
\tiny $R_{1.1}$ & famous(X) $\gets$ pay$_\tp$, X$\sim$sha, \^{f}amous$_{R_1}$(sha) \\
\tiny $\hat{R}_2$ & \^{w}rote$_{R_2}$(sha, kle) $\qgets{(1,1)}$ \\
\tiny $R_{2.1}$ & wrote(X, Y) $\gets$ pay$_\tp$, X$\sim$sha, Y$\sim$kle, \^{w}rote$_{R_2}$(sha, kle) \\
\tiny $R_{2.2}$ & authored(X, Y) $\gets$ pay$_{(0.9,0)}$, X$\sim$sha, Y$\sim$kle, \^{w}rote$_{R_2}$(sha, kle) \\
\tiny $\hat{R}_3$ & \^{w}rote$_{R_3}$(sha, hamlet) $\qgets{(1,1)}$ \\
\tiny $R_{3.1}$ & wrote(X, Y) $\gets$ pay$_\tp$, X$\sim$sha, Y$\sim$hamlet, \^{w}rote$_{R_3}$(sha, hamlet) \\
\tiny $R_{3.2}$ & authored(X, Y) $\gets$ pay$_{(0.9,0)}$, X$\sim$sha, Y$\sim$hamlet, \^{w}rote$_{R_3}$(sha, hamlet) \\
\tiny $\hat{R}_4$ & \^{g}ood\_work$_{R_4}$(G) $\qgets{(0.75,3)}$ famous(A)\#(0.5,100), authored(A, G) \\
\tiny $R_{4.1}$ & good\_work(X) $\gets$ pay$_\tp$, X$\sim$G, \^{g}ood\_work$_{R_4}$(G) \\[4mm]
\end{tabular} \\
\begin{tabular}{l@{\hspace{1.2cm}}l}
\% Program clauses for $\sim$: & \% Program clauses for pay: \\
X\,$\sim$Y $\gets$ X==Y & pay$_\tp$ $\gets$ \\
kle\,$\sim$\,kli $\gets$ pay$_{(0.8,2)}$ & pay$_{(0.9,0)}$ $\qgets{(0.9,0)}$ \\
$[\ldots]$ & pay$_{(0.8,2)}$ $\qgets{(0.8,2)}$ \\[2mm]
\end{tabular}
\end{center}

Finally, the goal $\elimS{G_r}$ for $\elimS{\Prog_r}$ is as follows:
\begin{center}
\it good\_work(X)\#W $\sep$ \!W $\dgeq^?$ \!\!(0.5,10) \mathproofbox
\end{center}
\end{exmp}

The next theorem proves the semantic correctness of the program transformation.

% Theorem: Program transformation correctness
\begin{thm}
\label{thm:SQCLP2QCLP:programs}
Consider a $\sqclp{\simrel}{\qdom}{\cdom}$-program $\Prog$\!, an atom $A$ relevant for $\Prog$\!, a qualification value $d \in \aqdom$ and a satisfiable finite set of $\cdom$-constraints $\Pi$.
Then, the following two statements are equivalent:
\begin{enumerate}
  \item $\Prog \sqchlrdc \cqat{A}{d}{\Pi}$
  \item $\elimS{\Prog} \qchldc \cqat{A_\sim}{d}{\Pi}$
\end{enumerate}
where $A_\sim$ is understood as in Definition \ref{def:sqclptransform}(1).
\end{thm}

\begin{proof*}
We separately prove each implication.

% 1. => 2.
\smallskip\noindent
[1. $\Rightarrow$ 2.] {\em (the transformation is complete).}
Assume that $T$ is a $\SQCHL(\simrel,\qdom,\cdom)$ proof tree witnessing $\Prog \sqchlrdc \cqat{A}{d}{\Pi}$.
We want to show the existence of a $\QCHL(\qdom,\cdom)$ proof tree $T'$
witnessing $\elimS{\Prog} \qchldc \cqat{A_\sim}{d}{\Pi}$. We reason by complete induction on $\Vert T \Vert$.
There are three possible cases according to the syntactic form of the atom $A$.
In each case we argue how to build the desired proof tree $T'$.

% 1. => 2. :: Primitive
\noindent --- $A$ is a primitive atom $\kappa$.
In this case $A_\sim$ is also $\kappa$ and $T$ contains only one {\bf SQPA} inference node. Because of the inference rules {\bf SQPA} and {\bf QPA}, both $\Prog \sqchlrdc \cqat{\kappa}{d}{\Pi}$ and $\elimS{\Prog} \qchldc \cqat{\kappa}{d}{\Pi}$ are equivalent to $\Pi \model{\cdom} \kappa$, therefore $T'$ trivially contains just one {\bf QPA} inference node.

% 1. => 2. :: Equation
\noindent --- $A$ is an equation $t == s$.
In this case $A_\sim$ is $t \sim s$ and $T$ contains just one {\bf SQEA} inference node. We know $\Prog \sqchlrdc \cqat{(t == s)}{d}{\Pi}$ is equivalent to $t \approx_{d,\Pi} s$ because of the inference rule {\bf SQEA}. From this equivalence follows $EQ_\simrel \qchldc \cqat{(t \sim s)}{d}{\Pi}$ due to Lemma \ref{lema:equiv} and hence $\elimS{\Prog} \qchldc \cqat{(t \sim s)}{d}{\Pi}$ by construction of $\elimS{\Prog}$. In this case, $T'$ will be a proof tree rooted by a {\bf QDA} inference step.

% 1. => 2. :: Defined
\noindent --- $A$ is a defined atom $p'(\ntup{t'}{n})$ with $p' \in DP^n$\!.
In this case $A_\sim$ is $p'(\ntup{t'}{n})$ and the root inference of $T$ must be a {\bf SQDA} inference step of the form:
$$
\displaystyle\frac
{~
  (~ \cqat{(t'_i == t_i\theta)}{d_i}{\Pi} ~)_{i = 1 \ldots n}
  \quad
  (~ \cqat{B_j\theta}{e_j}{\Pi} ~)_{j = 1 \ldots m}
~}
{~ \cqat{p'(\ntup{t'}{n})}{d}{\Pi} ~}
~ (\clubsuit)
$$
with $C : (p(\ntup{t}{n}) \qgets{\alpha} \qat{B_1}{w_1}, \ldots, \qat{B_m}{w_m}) \in \Prog$, $\theta$ substitution, $\simrel(p',p) = d_0 \neq \bt$, $e_j \dgeq^? w_j ~ (1 \leq j \leq m)$, $d \dleq d_i ~ (0 \leq i \leq n)$ and $d \dleq \alpha \circ e_j ~ (1 \leq j \leq m)$---which means $d \dleq \alpha$ in the case  $m = 0$.
We can assume that the first $n$ premises at ($\clubsuit$) are proved in $\sqclp{\simrel}{\qdom}{\cdom}$ w.r.t. $\Prog$ by proof trees $T_{1i} ~ (1 \leq i \leq n)$ satisfying $\Vert T_{1i} \Vert < \Vert T \Vert ~ (1 \leq i \leq n)$, and the last $m$ premises at ($\clubsuit$) are proved in $\sqclp{\simrel}{\qdom}{\cdom}$ w.r.t. $\Prog$ by proof trees $T_{2j} ~ (1 \leq j \leq m)$ satisfying $\Vert T_{2j} \Vert < \Vert T \Vert ~ (1 \leq j \leq m)$.

By Definition \ref{def:sqclptransform}, we know that the transformed program $\elimS{\Prog}$ contains two clauses of the following form:
$$
\begin{array}{c@{\hspace{1mm}}cl}
\hat{C}&:&
\hat{p}_C(\ntup{t}{n}) \qgets{\alpha} \qat{B_\sim^1}{w_1},~ \ldots,~ \qat{B_\sim^m}{w_m} \\
\hat{C}_{p'}&:&
p'(\ntup{X}{n}) \qgets{\tp} \qat{\mbox{pay}_{d_0}}{?},~ (~ \qat{(X_i \sim t_i)}{?} ~)_{i = 1 \ldots n},~ \qat{\hat{p}_C(\ntup{t}{n})}{?} \\
\end{array}
$$
where $X_i ~ (1 \leq i \leq n)$ are fresh variables not occurring in $C$ and $B_\sim^j ~ (1 \leq j \leq m)$ is the result of replacing `$\sim$' for `==' if $B_j$ is equation; and $B_j$ itself otherwise.
Given that the $n$ variables $X_i$ do not occur in $C$, we can assume that $\sigma \eqdef \theta' \uplus \theta$ with $\theta' \eqdef \{ X_1 \mapsto t'_1,~ \ldots,~ X_n \mapsto t'_n\}$ is a well-defined substitution.
We claim that $\elimS{\Prog} \qchldc \cqat{A_\sim}{d}{\Pi}$ can be proved with
a proof tree $T'$ rooted by the {\bf QDA} inference step ($\spadesuit$.1), which uses the clause $\hat{C}_{p'}$ instantiated by $\sigma$ and having $d_{n+1} = d$.
$$
\displaystyle\frac
{~
  \begin{array}{l}
  (~ \cqat{(t'_i == X_i\sigma)}{\tp}{\Pi} ~)_{i = 1 \ldots n} \\
  \cqat{\mbox{pay}_{d_0}\sigma}{d_0}{\Pi} \\
  (~ \cqat{(X_i \sim t_i)\sigma}{d_i}{\Pi} ~)_{i = 1 \ldots n} \\
  \cqat{\hat{p}_C(\ntup{t}{n})\sigma}{d_{n+1}}{\Pi} \\
  \end{array}
~}
{~ \cqat{p'(\ntup{t'}{n})}{d}{\Pi} ~}
~ (\spadesuit.1)
\quad
\displaystyle\frac
{~
  \begin{array}{l}
  (~ \cqat{(t'_i == X_i\theta')}{\tp}{\Pi} ~)_{i = 1 \ldots n} \\
  \cqat{\mbox{pay}_{d_0}}{d_0}{\Pi} \\
  (~ \cqat{(X_i\theta' \sim t_i\theta)}{d_i}{\Pi} ~)_{i = 1 \ldots n} \\
  \cqat{\hat{p}_C(\ntup{t}{n}\theta)}{d_{n+1}}{\Pi} \\
  \end{array}
~}
{~ \cqat{p'(\ntup{t'}{n})}{d}{\Pi} ~}
~ (\spadesuit.2)
$$
By construction of $\sigma$, ($\spadesuit$.1) can be rewritten as ($\spadesuit$.2), and in order to build the rest of $T'$\!, we show that each premise of ($\spadesuit$.2) admits a proof in $\QCHL(\qdom,\cdom)$ w.r.t. the transformed program $\elimS{\Prog}$:
\begin{itemize}
\item
$\elimS{\Prog} \qchldc \cqat{(t'_i == X_i\theta')}{\tp}{\Pi}$ for $i = 1 \ldots n$.
Straightforward using a single {\bf QEA} inference step since $X_i\theta' = t'_i$ and $t'_i \approx_\Pi t'_i$ is trivially true.
\item $\elimS{\Prog} \qchldc \cqat{\mbox{pay}_{d_0}}{d_0}{\Pi}$.
Immediate using the clause $(\mbox{pay}_{d_0} \qgets{d_0}) \in \elimS{\Prog}$ with a single {\bf QDA} inference step.
\item $\elimS{\Prog} \qchldc \cqat{(X_i\theta' \sim t_i\theta)}{d_i}{\Pi}$ for $i = 1 \ldots n$.
From the first $n$ premises of ($\clubsuit$) we know $\Prog \sqchlrdc \cqat{(t'_i == t_i\theta)}{d_i}{\Pi}$ with a proof tree $T_{1i}$ satisfying $\Vert T_{1i} \Vert < \Vert T \Vert$ for $i = 1 \ldots n$. Therefore, for $i = 1 \ldots n$, $\elimS{\Prog} \qchldc \cqat{(t'_i \sim t_i\theta)}{d_i}{\Pi}$ with some QCHL($\qdom$,$\cdom$) proof tree $T'_{1i}$ by inductive hypothesis. Since $(X_i\theta' \sim t_i\theta) = (t'_i \sim t_i\theta)$ for $i = 1 \ldots n$, we are done.
\item $\elimS{\Prog} \qchldc \cqat{\hat{p}_C(\ntup{t}{n}\theta)}{d}{\Pi}$.
This is proved by a $\QCHL(\qdom,\cdom)$ proof tree with a {\bf QDA} inference step node at its root of the following form:
$$
\displaystyle\frac
{~
  (~ \cqat{(t_i\theta == t_i\theta)}{d_i}{\Pi} ~)_{i = 1 \ldots n}
  \quad
  (~ \cqat{B_\sim^j\theta}{e_j}{\Pi} ~)_{j = 1 \ldots m}
~}
{~ \cqat{\hat{p}_C(\ntup{t}{n}\theta)}{d}{\Pi} ~}
~ (\heartsuit)
$$
which uses the program clause $\hat{C}$ instantiated by the substitution $\theta$.
Once more, we have to check that the premises can be derived in $\QCHL(\qdom,\cdom)$ from the transformed program $\elimS{\Prog}$ and that the side conditions of ($\heartsuit$) are satisfied:
\begin{itemize}
  \item The first $n$ premises can be trivially proved using {\bf QEA} inference steps.
  \item The last $m$ premises can be proved w.r.t. $\elimS{\Prog}$ with some $\QCHL(\qdom,\cdom)$ proof trees $T'_{2j} ~ (1 \leq j \leq m)$ by the inductive hypothesis, since we have premises $(~ \cqat{B_j\theta}{e_j}{\Pi} ~)_{j = 1 \ldots m}$ at ($\clubsuit$) that can be proved in $\sqclp{\simrel}{\qdom}{\cdom}$ w.r.t. $\Prog$ with proof trees $T_{2j}$ of size $\Vert T_{2j} \Vert < \Vert T \Vert ~ (1 \leq j \leq m)$.
  \item The side conditions---namely: $e_j \dgeq^? w_j ~ (1 \leq j \leq m)$, $d \dleq d_i ~ (1 \leq i \leq n)$
   and $d \dleq \alpha \circ e_j ~ (1 \leq j \leq m)$---trivially hold because they are also satisfied by ($\clubsuit$).
\end{itemize}
\end{itemize}

Finally, we complete the construction of $T'$ by checking that ($\spadesuit$.2) satisfies
the side conditions of the inference rule {\bf QDA}:
\begin{itemize}
  \item All threshold values at the body of $\hat{C}_{p'}$ are `?'\!, therefore the first group of side conditions becomes $d_i \dgeq^?\,\,? ~ (0 \leq i \leq n+1)$, which are trivially true.
  \item The second side condition reduces to $d \dleq \tp$, which is also trivially true.
  \item The third, and last, side condition is $d \dleq \tp \circ d_i ~ (0 \leq i \leq n+1)$, or equivalently $d \dleq d_i ~ (0 \leq i \leq n+1)$. In fact, $d \dleq d_i ~ (0 \leq i \leq n)$ holds due to the side conditions in ($\clubsuit$), and $d \dleq d_{n+1}$ holds because $d_{n+1} = d$ by construction of ($\spadesuit$.1) and ($\spadesuit$.2).
\end{itemize}

% 2. => 1.
\smallskip\noindent
[2. $\Rightarrow$ 1.] {\em (the transformation is sound).}
Assume that $T'$ is a $\QCHL(\qdom,\cdom)$ proof tree witnessing $\elimS{\Prog} \qchldc \cqat{A_\sim}{d}{\Pi}$.
We want to show the existence of a $\SQCHL(\simrel,\qdom,\cdom)$ proof tree $T$ witnessing $\Prog \sqchlrdc \cqat{A}{d}{\Pi}$. We reason by complete induction of $\Vert T' \Vert$.
There are three possible cases according to the syntactic form of the atom $A_\sim$.
In each case we argue how to build the desired proof tree $T$.

% 2. => 1. :: Primitive
\noindent --- $A_\sim$ is a primitive atom $\kappa$.
In this case $A$ is also $\kappa$ and $T'$ contains only one {\bf QPA} inference node. Both $\elimS{\Prog} \qchldc \cqat{\kappa}{d}{\Pi}$ and $\Prog \sqchlrdc \cqat{\kappa}{d}{\Pi}$ are equivalent to $\Pi \model{\cdom} \kappa$ because of the inference rules {\bf QPA} and {\bf SQPA}, therefore $T$ trivially contains just one {\bf SQPA} inference node.

% 2. => 1. :: Equation
\noindent --- $A_\sim$ is of the form $t \sim s$.
% \footnote{Note that $A_\sim : t == s$ is not possible due to construction of $A_\sim$.}
In this case $A$ is $t == s$ and $T'$ is rooted by a {\bf QDA} inference step. From $\elimS{\Prog} \qchldc \cqat{(t \sim s)}{d}{\Pi}$ and by construction of $\elimS{\Prog}$ we have $EQ_\simrel \qchldc \cqat{(t \sim s)}{d}{\Pi}$. By Lemma \ref{lema:equiv} we get $t \approx_{d,\Pi} s$ and, by the definition of the {\bf SQEA} inference step, we can build $T$ as a proof tree with only one {\bf SQEA} inference node proving $\Prog \sqchlrdc  \cqat{(t == s)}{d}{\Pi}$.

% 2. => 1. :: Defined
\noindent --- $A_\sim$ is a defined atom $p'(\ntup{t}{n})$ with $p'\in DP^n$ and $p' \neq\,\,\sim$.
In this case $A = A_\sim$ and the step at the root of $T'$ must be a  {\bf QDA} inference step using a clause $C' \in \elimS{\Prog}$ with head predicate $p'$ and a substitution $\theta$. Because of Definition \ref{def:sqclptransform} and the fact that $p'$ is relevant for $\Prog$, there must be some clause $C : (p(\ntup{t}{n}) \qgets{\alpha} \tup{B}) \in \Prog$ such that $\simrel(p,p') = d_0 \neq \bt$, and $C'$ must be of the form:
$$
C' :
p'(\ntup{X}{n}) \qgets{\tp} \qat{\mbox{pay}_{d_0}}{?},~ (\qat{(X_i \sim t_i)}{?})_{i = 1 \ldots n},~ \qat{\hat{p}_C(\ntup{t}{n})}{?}
$$
where the variables $\ntup{X}{n}$ do not occur in $C$.
Thus the {\bf QDA} inference step at the root of $T'$ must be of the form:
$$
\displaystyle\frac
{~
  \begin{array}{l}
  (~ \cqat{(t'_i == X_i\theta)}{d_{1i}}{\Pi} ~)_{i = 1 \ldots n} \\
  \cqat{\mbox{pay}_{d_0}\theta}{e_{10}}{\Pi} \\
  (~ \cqat{(X_i \sim t_i)\theta}{e_{1i}}{\Pi} ~)_{i = 1 \ldots n} \\
  \cqat{\hat{p}_C(\ntup{t}{n})\theta}{e_{1(n+1)}}{\Pi} \\
  \end{array}
~}
{~ \cqat{p'(\ntup{t'}{n})}{d}{\Pi} ~}
~ (\spadesuit)
$$
and the proof of the last premise must use the only clause for $\hat{p}_C$ introduced in $\elimS{\Prog}$ according to Definition  \ref{def:sqclptransform}, i.e.:
$$
\hat{C}:
\hat{p}_C(\ntup{t}{n}) \qgets{\alpha} \qat{B_\sim^1}{w_1},~ \ldots,~ \qat{B_\sim^m}{w_m} \enspace .
$$
Therefore, the proof of this premise must be of the form:
$$
\displaystyle\frac
{~
  (~ \cqat{(t_i\theta == t_i\theta')}{d_{2i}}{\Pi} ~)_{i = 1 \ldots n}
  \quad
  (~ \cqat{B_\sim^j\theta'}{e_{2j}}{\Pi} ~)_{j = 1 \ldots m}
~}
{~ \cqat{\hat{p}_C(\ntup{t}{n})\theta}{e_{1(n+1)}}{\Pi} ~}
~ (\heartsuit)
$$
for some substitution $\theta'$ not affecting $\ntup{X}{n}$.
We can assume that the last $m$ premises in ($\heartsuit$) are proved in $\QCHL(\qdom,\cdom)$ w.r.t. $\elimS{\Prog}$ by proof trees $T'_{j}$ satisfying $\Vert T'_{j} \Vert < \Vert T' \Vert ~ (1 \leq j \leq m)$.
Then we use the substitution $\theta'$ and clause $C$ to build a $\SQCHL(\simrel,\qdom,\cdom)$ proof tree $T$
with a {\bf SQDA} inference step at the root of the form:
$$
\displaystyle\frac
{~
  (~ \cqat{(t'_i == t_i\theta')}{e_{1i}}{\Pi} ~)_{i = 1 \ldots n}
  \quad
  (~ \cqat{B_j\theta'}{e_{2j}}{\Pi} ~)_{j = 1 \ldots m}
~}
{~ \cqat{p'(\ntup{t'}{n})}{d}{\Pi} ~}
~ (\clubsuit)
$$
Next we check that the premises of this inference step admit proofs in $\SQCHL(\simrel,\qdom,$ $\cdom)$ and that
$(\clubsuit)$ satisfies the side conditions of a valid {\bf SQDA} inference step.
\begin{itemize}
\item
$\Prog \sqchlrdc \cqat{(t'_i == t_i\theta')}{e_{1i}}{\Pi}$ for $i = 1 \ldots n$.
\begin{itemize}
\item From the premises $(\cqat{(X_i \sim t_i)\theta}{e_{1i}}{\Pi})_{i = 1 \ldots n}$ of $(\spadesuit)$ and by construction of $\elimS{\Prog}$ we know $EQ_\simrel \qchldc \cqat{(X_i \sim t_i)\theta}{e_{1i}}{\Pi} ~ (1 \leq i \leq n)$.
Therefore by Lemma \ref{lema:equiv} we have $X_i\theta \approx_{e_{1i},\Pi} t_i\theta$ for $i=1 \dots n$.
\item Consider now the premises $(\cqat{(t'_i == X_i\theta)}{d_{1i}}{\Pi})_{i = 1 \ldots n}$ of $(\spadesuit)$.
Their proofs  must rely on {\bf QEA} inference steps, and therefore $t'_i \approx_{\Pi} X_i\theta$ holds for $i=1 \dots n$.
\item Analogously, from the proofs of the premises $(\cqat{(t_i\theta == t_i\theta')}{d_{2i}}{\Pi})_{i = 1 \ldots n}$ we have $t_i\theta \approx_{\Pi} t_i\theta'$ (or equivalently $t_i\theta' \approx_{\Pi} t_i\theta$) for $i = 1 \ldots n$.
\end{itemize}
From the previous points we have $X_i\theta \approx_{e_{1i},\Pi} t_i\theta$, $t'_i \approx_{\Pi} X_i\theta$ and $t_i\theta' \approx_{\Pi} t_i\theta$, which by Lemma 2.7(1) of \cite{RR10TR} imply $ t'_i \approx_{e_{1i},\Pi} t_i\theta' ~ (1 \leq i \leq n)$.
Therefore the premises $(\cqat{(t'_i == t_i\theta')}{e_{1i}}{\Pi})_{i = 1 \ldots n}$ can be proven in $\SQCHL(\simrel,\qdom,\cdom)$ using a {\bf SQEA} inference step.
\item
$\Prog \sqchlrdc \cqat{B_j\theta'}{e_{2j}}{\Pi}$ for $j = 1 \ldots m$.
We know $\elimS{\Prog} \qchldc \cqat{B_\sim^j\theta'}{e_{2j}}{\Pi}$ with a proof tree $T'_j$ satisfying $\Vert T'_j \Vert < \Vert T' \Vert ~ (1 \leq j \leq m)$ because of ($\heartsuit$).
Therefore we have, by inductive hypothesis, $\Prog \sqchlrdc \cqat{B_j\theta'}{e_{2j}}{\Pi}$ for some $\SQCHL(\simrel,\qdom,\cdom)$ proof tree $T_j ~ (1 \leq j \leq m)$.
\item
$\simrel(p,p') = d_0 \neq \bt$.
As seen above.
\item
$e_{2j} \dgeq^? w_j$ for $j = 1 \ldots m$.
This is a side condition of the {\bf QDA} step in $(\heartsuit)$.
\item
$d \dleq e_{1i}$ for $i = 1 \ldots n$.
Straightforward from the side conditions of $(\spadesuit)$, which include $d \dleq \tp \circ e_{1i}$ for $(0 \le i \le n+1)$.
\item
$d \dleq \alpha \circ e_{2j}$ for $j = 1 \ldots m$.
This follows from the side conditions of $(\spadesuit)$ and $(\heartsuit)$, since we have $d \dleq \tp \circ e_{1i}$ for $i = 0 \ldots n+1$ (in particular $d \dleq e_{1(n+1)}$) and $e_{1(n+1)} \dleq \alpha \circ e_{2j}$ for $j = 1 \ldots m$. \mathproofbox
%This can be checked as follows: from the {\bf QDA} conditions of $(\spadesuit)$, we have $d \dleq \tp \circ e_{1i} ~ (0 \le i \le n+1)$, and in particular $d \dleq e_{1(n+1)}$. Finally, we have $e_{1(n+1)} \dleq \alpha \circ e_{2j}$, for $j = 1 \ldots m$ from the {\bf QDA} condition in $(\heartsuit)$. \mathproofbox
\end{itemize}
\end{proof*}

Finally, the next theorem extends the previous result to goals.

% Theorem: Goal transformation correctness
\begin{thm}
\label{thm:SQCLP2QCLP:goals}
Let $G$ be a goal for a $\sqclp{\simrel}{\qdom}{\cdom}$-program $\Prog$ whose atoms are all relevant for $\Prog$.
Assume $\Prog' = \elimS{\Prog}$ and $G' = \elimS{G}$.
Then, $\Sol{\Prog}{G} = \Sol{\Prog'}{G'}$.
\end{thm}
\begin{proof*}
According to the definition of goals in Section \ref{sec:sqclp}, and Definition \ref{def:sqclptransform}, $G$ and $G'$ must be of the form $(\qat{A_i}{W_i}, W_i \,{\dgeq}^? \beta_i)_{i = 1 \ldots m}$ and $(\qat{A_\sim^i}{W_i}, W_i \,{\dgeq}^? \beta_i)_{i = 1 \ldots m}$, respectively.
By Definitions \ref{dfn:goalsol} and \ref{dfn:qclp-goalsol}, both $\Sol{\Prog}{G}$ and $\Sol{\Prog'}{G'}$ are sets of triples $\langle \sigma, \mu, \Pi \rangle$ where $\sigma$ is a $\cdom$-substitution, $\mu : \warset{G} \to \aqdomd{\qdom}$ (note that $\warset{G} = \warset{G'}$) and $\Pi$ is a satisfiable finite set of $\cdom$-constraints.
Moreover:
\begin{enumerate}
\item
$\langle \sigma, \mu, \Pi \rangle \in \Sol{\Prog}{G}$ iff $W_i\mu = d_i \dgeq^? \!\beta_i$ and $\Prog \sqchlrdc \cqat{A_i\sigma}{W_i\mu}{\Pi}$ $(1 \leq i \leq m)$.
\item
$\langle \sigma, \mu, \Pi \rangle \in \Sol{\Prog'}{G'}$ iff $W_i\mu = d_i \dgeq^? \!\beta_i$ and $\Prog' \qchldc \cqat{A_\sim^i\sigma}{W_i\mu}{\Pi}$ $(1 \leq i \leq m)$.
\end{enumerate}
Because of Theorem \ref{thm:SQCLP2QCLP:programs}, conditions (1) and (2) are equivalent. \mathproofbox
\end{proof*}
 %\ref{sec:implemen:SQCLP2QCLP}
% ----
% Subsection 4.2: Transforming QCLP into CLP
% ----
\subsection{Transforming QCLP into CLP}
\label{sec:implemen:QCLP2CLP}

The results presented in this subsection are dependant on the assumption that the qualification domain $\qdom$ is existentially expressible in the constraint domain $\cdom$ via an injective mapping $\imath : \aqdomd{\qdom} \to C_\cdom$ and two existential $\cdom$-constraints of the following form:
\begin{itemize}
\item[] $\qval{X} = \exists U_1\ldots\exists U_k(B_1 \land \ldots \land B_m)$
\item[] $\qbound{X,Y,Z} = \exists V_1\ldots\exists V_l(C_1 \land \ldots \land C_q)$
\end{itemize}

The intuition behind $\qval{X}$ and $\qbound{X,Y,Z}$ has been explained in Definition \ref{dfn:expressible}. Roughly, they are intended to represent qualification values from $\qdom$ and the behaviour of $\qdom$'s attenuation operator $\circ$ by means of $\cdom$-constraints. Moreover, the assumption that $\qval{X}$ and $\qbound{X,Y,Z}$ have the existential form displayed above allows to build CLP clauses for two predicate symbols ${\it qVal} \in DP^1$ and ${\it qBound} \in DP^3$ which will capture the behaviour of the two corresponding constraints in the sense of Lemma \ref{lema:dec}. More precisely, we consider  the $\clp{\cdom}$-program $E_\qdom$ consisting of the following two clauses:
\begin{itemize}
\item[] ${\it qVal}(X) \gets B_1,\ \ldots,\ B_m$
\item[] ${\it qBound}(X,Y,Z) \gets C_1,\ \ldots,\ C_q$
\end{itemize}

The next example shows the CLP clauses in $E_\qdom$ for $\cdom = \rdom$ and three different choices of a qualification domain $\qdom$ that is existentially expressible in $\rdom$, namely: $\U$, $\W$ and $\U{\otimes}\W$.
In each case, the CLP clauses in $E_\qdom$ are obtained straightforwardly from the $\rdom$ constraints $\qval{X}$ and $\qbound{X,Y,Z}$ shown in Example \ref{exmp:qdom-representation}.

\begin{exmp}
\label{exmp:eqdom-clauses}
\begin{enumerate}
\item[1.] $E_\U$ consists of the following two clauses:
\end{enumerate}
\begin{center}
\footnotesize\it
\renewcommand{\arraystretch}{1.5}
\begin{tabular}{@{\hspace{1cm}}p{11cm}}
${\it qVal}(X) \gets {\it cp}_<(0,X),\ {\it cp}_\leq(X,1)$\\
${\it qBound}(X,Y,Z) \gets {\it op}_\times(Y,Z,X'),\ {\it cp}_\leq(X,X')$\\
\end{tabular}
\end{center}
%\begin{itemize}
%\item[] $qV\!al(X) ~\gets~ cp_<(0,X),  cp_\leq(X,1)$
%\item[] $qBound(X,Y,Z) ~\gets~ op_\times(Y,Z,X'), cp_\leq(X,X')$
%\end{itemize}
\begin{enumerate}
\item[2.] $E_\W$ consists of the following two clauses:
\end{enumerate}
\begin{center}
\footnotesize\it
\renewcommand{\arraystretch}{1.5}
\begin{tabular}{@{\hspace{1cm}}p{11cm}}
${\it qVal}(X) \gets {\it cp}_\ge(X,0)$\\
${\it qBound}(X,Y,Z) \gets {\it op}_+(Y,Z,X'),\ {\it cp}_\ge(X,X')$\\
\end{tabular}
\end{center}
%\begin{itemize}
%\item[] $qV\!al(X) ~\gets~ cp_>(X,0)$
%\item[] $qBound(X,Y,Z) ~\gets~ op_+(Y,Z,X'), cp_\geq(X,X')$
%\end{itemize}
\begin{enumerate}
\item[3.] $E_{\U{\otimes}\W}$ consists of the following two clauses:
\end{enumerate}
\begin{center}
\footnotesize\it
\renewcommand{\arraystretch}{1.5}
\begin{tabular}{@{\hspace{1cm}}p{11cm}}
${\it qVal}(X) \gets X=={\sf pair}(X_1,X_2),\ {\it cp}_<(0,X_1),\ {\it cp}_\le(X_1,1),\ {\it cp}_\ge(X_2,0)$\\
${\it qBound}(X,Y,Z) \gets X=={\sf pair}(X_1,X_2),\ Y=={\sf pair}(Y_1,Y_2),\ Z=={\sf pair}(Z_1,Z_2),$\\
$\qquad
{\it op}_\times(Y_1,Z_1,X'_1),\ {\it cp}_\le(X_1,X'_1),\
{\it op}_+(Y_2,Z_2,X'_2),\ {\it cp}_\ge(X_2,X'_2) \mathproofbox
$\\
\end{tabular}
\end{center}
%\begin{itemize}
%\item[] $qV\!al(X) ~\gets~X ==$ {\sf pair}$(X_1,X_2), cp_<(0,X_1), cp_\leq(X_1,1), cp_>(X_2,0)$
%\item[] $qBound(X,Y,Z) ~\gets~X ==$ {\sf pair}$(X_1,X_2),  Y ==$ {\sf pair}$(Y_1,Y_2), Z ==$ {\sf pair}$(Z_1,Z_2),
%op_\times(Y_1,Z_1,X'_1), cp_\leq(X_1,X'_1), op_+(Y_2,Z_2,X'_2), cp_\geq(X_2,X'_2)$ \mathproofbox
%\end{itemize}
\end{exmp}

In general, the CLP clauses in $E_\qdom$ along with other techniques explained in the rest of this subsection will be used to present semantically correct transformations from $\qclp{\qdom}{\cdom}$ into $\clp{\cdom}$, working both for programs and goals. All our results will work under the assumption that ${\it qVal} \in DP^1$ and ${\it qBound} \in DP^3$ are chosen as fresh predicate symbols not occurring in the $\qclp{\qdom}{\cdom}$ programs and goals to be transformed. The next technical lemma ensures that the predicates {\it qVal} and {\it qBound} correctly represent the behaviour of the constraints $\qval{X}$ and $\qbound{X,Y,Z}$.

%In order to compute with the encodings of $\qdom$ values in $\cdom$,
%we will use  the $\clp{\cdom}$-program $E_\qdom$ consisting of the following two clauses:
%\begin{itemize}
%\item[] $qV\!al(X) ~\gets~ B_1,\ \ldots,\ B_m$
%\item[] $qBound(X,Y,Z) ~\gets~ C_1,\ \ldots,\ C_q$
%\end{itemize}
%where $qV\!al \in DP^1$ and $qBound \in DP^3$ do not occur in the $\qclp{\qdom}{\cdom}$ programs and goals to be transformed.

%The lemma stated below is an immediate consequence of Lemma \ref{lema:dec} and Definition \ref{dfn:expressible}.

% Lemma to deal with qVal and qBound via defined predicates.
\begin{lem}
\label{lema:expr}
For any satisfiable finite set $\Pi$ of $\cdom$-constraints one has:
\begin{enumerate}
\item
For any ground term $t \in C_\cdom$:
$$t \in \mbox{ran}(\imath) \iff \qval{t} \mbox{ true in } \cdom \iff E_\qdom \chlc \cat{{\it qVal}(t)}{\Pi}$$
\item
For any ground terms $r = \imath(x)$, $s = \imath(y)$, $t = \imath(z)$ with $x,y,z \in \aqdomd{\qdom}$:
$$x \dleq y \circ z \iff \qbound{r,s,t} \mbox{ true in } \cdom \iff E_\qdom \chlc \cat{{\it qBound}(r,s,t)}{\Pi}$$
\end{enumerate}
The two items above are also valid if $E_\qdom$ is replaced by any $\clp{\cdom}$-program including the two clauses in $E_\qdom$ and having no additional occurrences of {\it qVal} and {\it qBound} at the head of clauses.
\end{lem}
\begin{proof}
Immediate consequence of Lemma \ref{lema:dec} and Definition \ref{dfn:expressible}.
\end{proof}

% Figure: Transformation rules
\begin{figure}[ht]
\figrule
\centering
\begin{tabular}{@{\hspace{4mm}}l@{\hspace{1mm}}l}
\multicolumn{2}{l}{\textbf{Transforming Atoms}} \\&\\
\bf TEA & $\transform{({t == s})} \!= (t == s, ~\imath(\tp))$. \\
&\\
\bf TPA & $\transform{({\kappa})} \!= (\kappa, ~\imath(\tp))$ with $\kappa$ primitive atom. \\
&\\
\bf TDA & $\transform{({p(\ntup{t}{n})})} \!= (p'(\ntup{t}{n},W), ~W)$ with $p \in DP^n$ and $W$\! a fresh CLP variable. \\
&\\
\multicolumn{2}{l}{\textbf{Transforming qc-Atoms}\vspace{2mm}} \\
\bf TQCA &  $\displaystyle\frac
  {\transform{A} = (A',w) }
  {\quad \transform{(\qat{A}{d} \Leftarrow \Pi)} = (A' \Leftarrow \Pi, ~\{\qval{w},\ \qbound{\imath(d),\imath(\tp),w}\}) \quad}$ \\
&\\
\multicolumn{2}{l}{\textbf{Transforming Program Clauses}\vspace{2mm}} \\
\bf TPC & $\displaystyle\frac
  { (~ \transform{B_j} = (B_j',  w'_j) ~)_{j=1 \dots m}}
  {\quad
        \transform{C} = p'(\ntup{t}{n},W) ~\gets~ qV\!al(W),\ \left(
          \begin{array}{l}
            qV\!al(w_j'),\ \encode{w'_j \dgeq^? \!\imath(w_j)}, \\
            qBound(W, \imath(\alpha), w'_j),\ B'_j
          \end{array}
        \right)_{j = 1 \ldots m}
  \quad}$ \\
&\\
\multicolumn{2}{l}{$\qquad$ where $C : p(\ntup{t}{n}) \qgets{\alpha} \qat{B_1}{w_1}, \ldots, \qat{B_m}{w_m}$, $W$\! is a fresh CLP variable and} \\
\multicolumn{2}{l}{$\qquad$ $\encode{w'_j \dgeq^? \imath(w_j)}$ is omitted if $w_j =\ ?$, otherwise abbreviates $qBound(\imath(w_j),\imath(\tp),w'_j)$.}\\
&\\
\multicolumn{2}{l}{\textbf{Transforming Goals}\vspace{2mm}} \\
\bf TG & $\displaystyle\frac
  {(~ \transform{B_j} = (B_j',  w'_j) ~)_{j=1 \dots m}}
  {\quad
          \elimD{G} = \left(
            \begin{array}{l}
              qV\!al(W_j),\ \encode{W_j\dgeq^? \imath(\beta_j)}, \\
              qV\!al(w'_j),\ qBound(W_j, \imath(\tp), w'_j),\  B_j'
            \end{array}
          \right)_{j = 1 \ldots m}
  \quad}$ \\
&\\
\multicolumn{2}{l}{$\qquad$ where $G : (\qat{B_j}{W_j}, W_j \dgeq^? \beta_j)_{j=1 \dots m}$ and $\encode{W_j\dgeq^? \imath(\beta_j)}$ as in \textbf{TPC} above.} \\
\end{tabular}
\caption{Transformation rules}
\label{fig:transformation}
\figrule
\end{figure}

Now we are ready to define the  transformations from $\qclp{\qdom}{\cdom}$ into $\clp{\cdom}$.
\begin{defn}\label{def:qclptransform}
Assume that $\qdom$ is existentially expressible in $\cdom$, and let $\qval{X}$, $\qbound{X,Y,Z}$ and  $E_\qdom$ be as explained above.
Assume also a $\qclp{\qdom}{\cdom}$-program $\Prog$ and a $\qclp{\qdom}{\cdom}$-goal $G$ for $\Prog$ without occurrences of the defined predicate symbols $qV\!al$ and $qBound$.
Then:
\begin{enumerate}
\item
$\Prog$ is transformed into the $\clp{\cdom}$-program  $\elimD{\Prog}$ consisting of the two clauses in $E_\qdom$ and the
transformed $\transform{C}\!$ of each clause $C \in \Prog$, built as specified in Figure \ref{fig:transformation}. The transformation rules of this figure translate each $n$-ary predicate symbol $p \in DP^n$ into a different $(n+1)$-ary predicate symbol  $p' \in DP^{n+1}$\!.
\item
$G$ is transformed into the $\clp{\cdom}$-goal $\elimD{G}$ built as specified in Figure \ref{fig:transformation}.
Note that the qualification variables $\ntup{W}{\!n}$ occurring in $G$ become normal CLP variables in the transformed goal.
\mathproofbox
\end{enumerate}
\end{defn}

\label{fig:transformation:explanation}
The first three rules in Figure \ref{fig:transformation} are used for transforming atoms. For convenience, the transformation of an atom produces a pair where the first value is the transformed atom and the second one is either a new variable or the representation of $\tp$.
In the first two cases, namely {\bf TEA} and {\bf TPA}, the transformation behaves as the identity and no new variables are introduced.
The third case, namely {\bf TDA}, corresponds to the transformation of a defined atom. In this case, a new CLP variable $W$---intended to represent the qualification value associated to the atom---is added as its last argument.
The rule {\bf TQCA} transforms qc-atoms of the form $\qat{A}{d} \Leftarrow \Pi$ by means of the transformation of $A$ using one of the three aforementioned transformation rules. This transformation returns a pair $(A',w)$ in which, as shown above, $w$ can be either a new variable or the representation of $\tp$. Since $w$ can be a new variable $W$,  the constraint $\qval{w}$ is introduced to ensure that it represents a qualification value.
Finally, the constraint $\qbound{\imath(d),\imath(\tp),w}$ encodes ``$d \dleq \tp \circ w$,'' or equivalently ``$d \dgeq w.$''
The rule {\bf TPC} is employed for transforming program clauses $C : p(\ntup{t}{n}) \qgets{\alpha} \qat{B_1}{w_1}, \ldots, \qat{B_m}{w_m}$ where each $w_i$ is either a qualification value or $?$ indicating that proving the atom with any qualification value different from $\bt$ is acceptable. The rule introduces a new variable $W$ together with a constraint ${\it qVal}(W)$. The variable represents the qualification value associated to the computation of user defined atoms involving $p$ (renamed as $p'$ in the transformed program). The premises $(\transform{B_j} = (B_j',  w'_j))_{j=1 \dots m}$ transform the atoms in the body of the clause using in each case either {\bf TEA}, {\bf TPA} or {\bf TDA}. Therefore, each $w'_j$ obtained in this way represents a qualification value encoded as a constraint value.
Moreover, the qualification value encoded by $w'_j$ must be  greater or equal than the corresponding qualification value
$w_j$ that occurs in the program clause.
These two requirements are represented as ${\it qVal}(w_j'),\ \encode{w'_j \dgeq^? \!\imath(w_j)}$ in the transformed clause. The predicate call ${\it qBound}(W, \imath(\alpha), w'_j)$ ensures that the value in $W$ must be less than or equal to ``$\alpha \circ w'_j$'' for every $j$. For each $j=1 \dots m$ all the atoms associated to the transformation of $B_j$ precede the transformed atom $B'_j$. In a {\tt Prolog}-based implementation, this helps to prune the search space as soon as possible during the computations.
The ideas behind rule {\bf TG} are similar. A goal $G : (\qat{B_j}{W_j}, W_j \dgeq^? \beta_j)_{j=1 \dots m}$ is transformed by introducing
atoms in charge of checking that: each $W_j$ is a valid qualification value; each $W_j$ is indeed less than or equal to the representation of $\beta_j$ in CLP; each value $w_j$---obtained during the transformation of the atoms $B_j$---corresponds to an actual qualification value; and finally, that each $W_j$ is satisfactory---i.e. less or equal to---w.r.t. its corresponding $w_j$ before effectively introducing the transformed atoms $B'_j$. The following example illustrates the transformation elim$_\qdom$.
\begin{exmp}[Running example: $\clp{\rdom}$-program $\elimD{\elimS{\Prog_r}}$]
\label{exmp:elimd-elims-pr}

Consider the $\qclp{\U{\otimes}\W}{\rdom}$-program $\elimS{\Prog_r}$ and the goal $\elimS{G_r}$ for the same program as presented in Example \ref{exmp:elims-pr}. The transformed $\clp{\rdom}$-program $\elimD{\elimS{\Prog_r}}$ is as follows:
\begin{center}
\footnotesize\it
\renewcommand{\arraystretch}{1.5}
\begin{tabular}{rl}
\tiny $\hat{R}_1$ & \^{f}amous$_{R_1}$(sha, W) $\gets$ qVal(W), qBound(W, $\tp$, (0.9,1)) \\
\tiny $R_{1.1}$ & famous(X, W) $\gets$ qVal(W), qVal(W$_1$), qBound(W, $\tp$, \!W$_1$), pay$_\tp$(W$_1$), \\
& $\quad$ qVal(W$_2$), qBound(W, $\tp$, \!W$_2$), $\sim$(X, sha, W$_2$), \\
& $\quad$ qVal(W$_3$), qBound(W, $\tp$, \!W$_3$), \^{f}amous$_{R_1}$(sha, W$_3$) \\
\end{tabular}
\end{center}
\begin{center}
\footnotesize\it
\renewcommand{\arraystretch}{1.5}
\begin{tabular}{rl}
\tiny $\hat{R}_2$ & \^{w}rote$_{R_2}$(sha, kle, W) $\gets$ qVal(W), qBound(W, $\tp$, (1,1)) \\
\tiny $R_{2.1}$ & wrote(X, Y, W) $\gets$ qVal(W), qVal(W$_1$), qBound(W, $\tp$, \!W$_1$), pay$_\tp$(W$_1$), \\
& $\quad$ qVal(W$_2$), qBound(W, $\tp$, \!W$_2$), $\sim$(X, sha, W$_2$), \\
& $\quad$ qVal(W$_3$), qBound(W, $\tp$, \!W$_3$), $\sim$(Y, kle, W$_3$), \\
& $\quad$ qVal(W$_4$), qBound(W, $\tp$, \!W$_4$), \^{w}rote$_{R_2}$(sha, kle, W$_4$) \\
\tiny $R_{2.2}$ & authored(X, Y, W) $\gets$ qVal(W), qVal(W$_1$), qBound(W, $\tp$, \!W$_1$), pay$_{(0.9,0)}$(W$_1$), \\
& $\quad$ qVal(W$_2$), qBound(W, $\tp$, \!W$_2$), $\sim$(X, sha, W$_2$), \\
& $\quad$ qVal(W$_3$), qBound(W, $\tp$, \!W$_3$), $\sim$(Y, kle, W$_3$), \\
& $\quad$ qVal(W$_4$), qBound(W, $\tp$, \!W$_4$), \^{w}rote$_{R_2}$(sha, kle, W$_4$) \\
\tiny $\hat{R}_3$ & \^{w}rote$_{R_3}$(sha, hamlet, W) $\gets$ qVal(W), qBound(W, $\tp$, (1,1)) \\
\tiny $R_{3.1}$ & wrote(X, Y, W) $\gets$ qVal(W), qVal(W$_1$), qBound(W, $\tp$, \!W$_1$), pay$_\tp$(W$_1$), \\
& $\quad$ qVal(W$_2$), qBound(W, $\tp$, \!W$_2$), $\sim$(X, sha, W$_2$), \\
& $\quad$ qVal(W$_3$), qBound(W, $\tp$, \!W$_3$), $\sim$(Y, hamlet, W$_3$), \\
& $\quad$ qVal(W$_4$), qBound(W, $\tp$, \!W$_4$), \^{w}rote$_{R_3}$(sha, hamlet, W$_4$) \\
\tiny $R_{3.2}$ & authored(X, Y, W) $\gets$ qVal(W), qVal(W$_1$), qBound(W, $\tp$, \!W$_1$), pay$_{(0.9,0)}$(W$_1$), \\
& $\quad$ qVal(W$_2$), qBound(W, $\tp$, \!W$_2$), $\sim$(X, sha, W$_2$), \\
& $\quad$ qVal(W$_3$), qBound(W, $\tp$, \!W$_3$), $\sim$(Y, hamlet, W$_3$), \\
& $\quad$ qVal(W$_4$), qBound(W, $\tp$, \!W$_4$), \^{w}rote$_{R_3}$(sha, hamlet, W$_4$) \\
\tiny $\hat{R}_4$ & \^{g}ood\_work$_{R_4}$(G, W) $\gets$ qVal(W), \\
& $\quad$ qVal(W$_1$), qBound((0.5,100), $\tp$, W$_1$), qBound(W, (0.75,3), W$_1$), famous(A, W$_1$), \\
& $\quad$ qVal(W$_2$), qBound(W, (0.75,3), W$_2$), authored(A, G, W$_2$) \\
\tiny $R_{4.1}$ & good\_work(X, W) $\gets$ qVal(W), qVal(W$_1$), qBound(W, $\tp$, \!W$_1$), pay$_\tp$(W$_1$), \\
& $\quad$ qVal(W$_2$), qBound(W, $\tp$, \!W$_2$), $\sim$(X, G, W$_2$), \\
& $\quad$ qVal(W$_3$), qBound(W, $\tp$, \!W$_3$), \^{g}ood\_work$_{R_4}$(G, W$_3$) \\[3mm]
& \% Program clauses for $\sim$: \\
& $\sim$(X, Y, W) $\gets$ qVal(W), qVal($\tp$), qBound(W, $\tp$, $\tp$), X==Y \\
& $\sim$(kle, kli, W) $\gets$ qVal(W), qVal(W$_1$), qBound(W, $\tp$, \!W$_1$), pay$_{(0.8,2)}$(W$_1$) \\
& $[\ldots]$ \\[3mm]
& \% Program clauses for pay: \\
& pay$_\tp$(W) $\gets$ qVal(W), qBound(W, $\tp$, $\tp$) \\
& pay$_{(0.9,0)}$(W) $\gets$ qVal(W), qBound(W, $\tp$, (0.9,0)) \\
& pay$_{(0.8,2)}$(W) $\gets$ qVal(W), qBound(W, $\tp$, (0.8,2)) \\[3mm]
& \% Program clauses for qVal \& qBound: \\
& qVal((X$_1$,X$_2$)) $\gets$ X$_1$ $>$ 0, X$_1$ $\le$ 1, X$_2$ $\ge$ 0 \\
& qBound((W$_1$,W$_2$), (Y$_1$,Y$_2$), (Z$_1$,Z$_2$)) $\gets$ W$_1$ $\le$ Y$_1$ $\times$ Z$_1$, W$_2$ $\ge$ Y$_2$ $+$ Z$_2$ \\[2mm]
\end{tabular}
\end{center}

Finally, the goal $\elimD{\elimS{G_r}}$ for $\elimD{\elimS{\Prog_r}}$ is as follows:
\begin{center}
\footnotesize\it qVal(W), qBound((0.5,10), $\tp$, \!W), qVal(W'), qBound(W, $\tp$, \!W'), good\_work(X, W')
\end{center}

Note that, in order to improve the clarity of the program clauses of this example, the qualification value $(1,\!0)$---top value in $\U{\otimes}\W$---has been replaced by $\tp$. \mathproofbox
\end{exmp}

The next theorem proves the semantic correctness of the program transformation.

% Theorem: Program transformation correctness
\begin{thm}
\label{thm:QCLP2CLP:programs}
Let $A$ be an atom such that $qV\!al$ and $qBound$ do not occur in $A$. Assume $d \in \aqdom$ such that $\transform{(\cqat{A}{d}{\Pi})} = (\cat{A'}{\Pi}, \Omega)$.
Then, the two following statements are equivalent:
\begin{enumerate}
  \item  $\Prog \qchldc \cqat{A}{d}{\Pi}$
  \item $\elimD{\Prog} \chlc \cat{A'\!\rho}{\Pi}$ for some $\rho \in \Sol{\cdom}{\Omega}$ such that $\domset{\rho} = \varset{\Omega}$.
\end{enumerate}
\end{thm}
\begin{proof*}
We separately prove each implication.

% [1 => 2].
\smallskip\noindent
[1. $\Rightarrow$ 2.] {\em (the transformation is complete).}
We assume that $T$ is a $\QCHL(\qdom,\cdom)$ proof tree witnessing $\Prog \qchldc \cqat{A}{d}{\Pi}$.
We want to show the existence of a $\clp{\cdom}$ proof tree  $T'$
witnessing $\elimD{\Prog} \chlc \cat{A'\!\rho}{\Pi}$
for some $\rho \in \Sol{\cdom}{\Omega}$ such that $\domset{\rho} = \varset{\Omega}$.
We reason  by complete induction on $\Vert T \Vert$.
There are three possible cases,  according to the the syntactic form of the atom $A$.
In each case we argue how to build the desired proof tree $T'$\!.

% [1 => 2]. Primitive.
\noindent --- $A$ is a primitive atom $\kappa$.
In this case {\bf TQCA} and {\bf TPA} compute $A' = \kappa$ and $\Omega = \{\qval{\imath(\tp)},\ \qbound{\imath(d),\imath(\tp),\imath(\tp)}\}$. Now, from $\Prog \qchldc \cqat{\kappa}{d}{\Pi}$ follows $\Pi \model{\cdom} \kappa$ due to the {\bf QPA} inference, and therefore taking $\rho = \varepsilon$ we can prove $\elimD{\Prog} \chlc \cat{\kappa\varepsilon}{\Pi}$ with a proof tree $T'$\! containing only one {\bf PA} node. Moreover, $\varepsilon \in \Solc{\Omega}$ is trivially true because the two constraints belonging to $\Omega$ are obviously true in $\cdom$.

% [1 => 2]. Equation.
\noindent --- $A$ is an equation $t == s$.
In this case {\bf TQCA} and {\bf TEA} compute $A' = (t == s)$ and $\Omega = \{\qval{\imath(\tp)},\ \qbound{\imath(d),\imath(\tp),\imath(\tp)}\}$. Now, from $\Prog \qchldc \cqat{(t==s)}{d}{\Pi}$ follows $t \approx_\Pi s$ due to the {\bf QEA} inference, and therefore taking $\rho = \varepsilon$ we can prove $\elimD{\Prog} \chlc \cat{(t == s)\varepsilon}{\Pi}$ with a proof tree $T'$\! containing only one {\bf EA} node. Moreover, $\varepsilon \in \Solc{\Omega}$ is trivially true because the two constraints belonging to $\Omega$ are obviously true in $\cdom$.

% [1 => 2]. Defined.
\noindent --- $A$ is a defined atom $p(\ntup{t'}{n})$ with $p \in DP^n$.
In this case {\bf TQCA} and {\bf TDA} compute $A' = p'(\ntup{t'}{n},W)$ and $\Omega = \{\qval{W},\ \qbound{\imath(d),\imath(\tp),W}\}$ where $W$ is a fresh CLP variable.
On the other hand, $T$ must be rooted by a {\bf QDA} step of the form:
$$
\displaystyle\frac
    {~ (~ \cqat{(t'_i == t_i\theta)}{d_i}{\Pi} ~)_{i = 1 \ldots n} \quad (~ \cqat{B_j\theta}{e_j}{\Pi} ~)_{j = 1 \ldots m} ~}
    {\cqat{p(\ntup{t'}{n})}{d}{\Pi}} \quad (\clubsuit)
$$
using a clause $C : (p(\ntup{t}{n}) \qgets{\alpha} \qat{B_1}{w_1}, \ldots, \qat{B_m}{w_m}) \in \Prog$ instantiated by a substitution $\theta$ and such that the side conditions $e_j \dgeq^? w_j ~ (1 \le j \le m)$, $d \dleq d_i ~ (1 \le i \le n)$ and $d \dleq \alpha \circ e_j ~ (1 \le j \le m)$ are fulfilled.

For $j = 1 \ldots m$ we can assume $\transform{B_j} = (B'_j, w'_j)$ and thus $\transform{(\cqat{B_j\theta}{e_j}{\Pi})} = (\cat{B'_j\theta}{\Pi}, \Omega_j)$ where $\Omega_j = \{\qval{w'_j},\ \qbound{\imath(e_j),\imath(\tp),w'_j}\}$. The proof trees $T_j$ of the last $m$ premises of $(\clubsuit)$ will have less than $\Vert T \Vert$ nodes,
and hence the induction hypothesis can be applied to each $(\cqat{B_j\theta}{e_j}{\Pi})$ with $1 \leq j \leq m$, obtaining CHL($\cdom$) proof trees $T'_j$ proving $\elimD{\Prog} \chlc \cat{B'_j\theta\rho_j}{\Pi}$ for some $\rho_j \in \Solc{\Omega_j}$ with $\domset{\rho_j} = \varset{\Omega_j}$.

Consider $\rho = \{W \mapsto \imath(d)\}$ and $\transform{C} \in \elimD{\Prog}$ of the form:
%$$
%\begin{array}{rcl}
%\transform{C} :~ p'(\ntup{t}{n},W') & \gets & qV\!al(W'), \\
%&&
%  \left( \begin{array}{l}
%    qV\!al(w_j'), ~\encode{w'_j \dgeq^? \imath(w_j)}, \\
%    qBound(W', \imath(\alpha), w'_j), ~B'_j \\
%  \end{array}\right)_{j = 1 \ldots m.} \\
%\end{array}
%$$
$$
\transform{C} :~ p'(\ntup{t}{n},W') ~\gets~ qV\!al(W'),\
  \left( \begin{array}{l}
    qV\!al(w_j'), ~\encode{w'_j \dgeq^? \imath(w_j)}, \\
    qBound(W', \imath(\alpha), w'_j), ~B'_j \\
  \end{array}\right)_{j = 1 \ldots m.} \\
$$
Obviously, $\rho \in \Solc{\Omega}$ and $\domset{\rho} = \varset{\Omega}$. To finish the proof we must prove $\elimD{\Prog} \chlc \cat{A'\!\rho}{\Pi}$. We claim that this can be done with a CHL($\cdom$) proof tree $T'$ whose root inference is a {\bf DA} step of the form:
$$
\displaystyle\frac
{~
  \begin{array}{l}
     ~~~~(~ \cat{(t'_i\rho == t_i\theta')}{\Pi} ~)_{i = 1 \ldots n} \\
     ~~~~\cat{(W\rho == W'\theta')}{\Pi} \\
     ~~~~\cat{qV\!al(W')\theta'}{\Pi}  \\
     \left( \begin{array}{l}
      \cat{qV\!al(w'_j)\theta'}{\Pi} \\
      \cat{\encode{w'_j \dgeq^? \imath(w_j)}\theta'}{\Pi} \\
      \cat{qBound(W',\imath(\alpha),w'_j)\theta'}{\Pi} \\
      \cat{B'_j\theta'}{\Pi}
    \end{array}\right)_{j = 1 \ldots m} \\
  \end{array}
~}
{\cat{p'(\ntup{t'}{n},W)\rho}{\Pi}} ~ (\spadesuit)
$$
using $\transform{C}$ instantiated by the substitution $\theta' = \theta \uplus \rho_1 \uplus \dots \uplus \rho_m \uplus \{W' \mapsto \imath(d)\}$.
We check that the premises of ($\spadesuit$) can be derived from $\elimD{\Prog}$ in CHL($\cdom$):
\begin{itemize}
\item
$\elimD{\Prog} \chlc \cat{(t'_i\rho == t_i\theta')}{\Pi}$ for $i = 1 \ldots n$.
By construction of $\rho$ and $\theta'$\!, these are equivalent to prove $\elimD{\Prog} \chlc \cat{(t'_i == t_i\theta)}{\Pi}$ for $i = 1 \ldots n$ and these hold with CHL($\cdom$) proof trees of only one {\bf EA} node because of $t'_i \approx_\Pi t_i\theta$, which is a consequence of the first $n$ premises of ($\clubsuit$).
\item
$\elimD{\Prog} \chlc \cat{(W\rho == W'\theta')}{\Pi}$.
By construction of $\rho$ and $\theta'$\!, this is equivalent to prove $\elimD{\Prog} \chlc \cat{(\imath(d) == \imath(d))}{\Pi}$ which results trivial.
\item
$\elimD{\Prog} \chlc \cat{qV\!al(W')\theta'}{\Pi}$.
By construction of $\theta'$, this is equivalent to prove $\elimD{\Prog} \chlc \cat{qV\!al(\imath(d))}{\Pi}$. We trivially have that $\imath(d) \in \mbox{ran}(\imath)$. Then, by Lemma \ref{lema:expr}, this premise holds.
\item
$\elimD{\Prog} \chlc \cat{qV\!al(w'_j)\theta'}{\Pi}$ for $j = 1 \ldots m$.
By construction of $\theta'$ and Lemma \ref{lema:expr} we must prove, for any fixed $j$, that $\qval{w'_j\rho_j}$ is true in $\cdom$. As $\rho_j \in \Solc{\Omega_j}$ we know $\rho_j \in \Solc{\qval{w'_j}}$, therefore $\qval{w'_j\rho_j}$ is trivially true in $\cdom$.
\item
$\elimD{\Prog} \chlc \cat{\encode{w'_j \dgeq^? \imath(w_j)}\theta'}{\Pi}$ for $j = 1 \ldots m$.
We reason for any fixed $j$.
If $w_j =\ ?$ this results trivial.
Otherwise, it amounts to $\qbound{\imath(w_j),\imath(\tp),w'_j\rho_j}$ being true in $\cdom$, by construction of $\theta'$ and Lemma \ref{lema:expr}.
As seen before, $\qval{w'_j\rho_j}$ is true in $\cdom$, therefore $w'_j\rho_j = \imath(e'_j)$ for some $e'_j \in \aqdom$. From the side conditions of ($\clubsuit$) we have $w_j \dleq e_j$.
On the other hand, $\rho_j \in \Solc{\Omega_j}$ and, in particular, $\rho_j \in \Solc{\qbound{\imath(e_j),\imath(\tp),w'_j}}$. This, together with $w'_j\rho_j = \imath(e'_j)$, means $e_j \dleq e'_j$, which with $w_j \dleq e_j$ implies $w_j \dleq e'_j$, i.e. $\qbound{\imath(w_j),\imath(\tp),w'_j\rho_j}$ is true in $\cdom$.
\item
$\elimD{\Prog} \chlc \cat{qBound(W',\imath(\alpha),w'_j)\theta'}{\Pi}$ for $j = 1 \ldots m$.
We reason for any fixed $j$. By construction of $\theta'$ and Lemma \ref{lema:expr}, we must prove that $\qbound{\imath(d),\imath(\alpha),w'_j\rho_j}$ is true in $\cdom$. As seen before, $\qval{w'_j\rho_j}$ is true in $\cdom$, therefore $w'_j\rho_j = \imath(e'_j)$ for some $e'_j \in \aqdom$. From the side conditions of ($\clubsuit$) we have $d \dleq \alpha \circ e_j$. On the other hand, $\rho_j \in \Solc{\Omega_j}$ and, in particular, $\rho_j \in \Solc{\qbound{\imath(e_j), \imath(\tp), w'_j}}$. This, together with $w'_j\rho_j = \imath(e'_j)$, means $e_j \dleq e'_j$. Now, $d \dleq \alpha \circ e_j$ and $e_j \dleq e'_j$ implies $d \dleq \alpha \circ e'_j$, i.e. $\qbound{\imath(d),\imath(\alpha),w'_j\rho_j}$ is true in $\cdom$.
\item
$\elimD{\Prog} \chlc \cat{B'_j\theta'}{\Pi}$ for $j = 1 \ldots m$.
In this case, it is easy to see that $B'_j\theta' = B'_j\theta\rho_j$ by construction of $\theta'$ and because of the program transformation rules.
On the other hand, proof trees $T'_j$ proving $\elimD{\Prog} \chlc \cat{B'_j\theta\rho_j}{\Pi}$ can be obtained by inductive hypothesis as seen before.
\end{itemize}

% [2 => 1].
\smallskip\noindent
[2. $\Rightarrow$ 1.] {\em (the transformation is sound).}
We assume that $T'$ is a a CHL($\cdom$) proof tree
witnessing $\elimD{\Prog} \chlc \cat{A'\rho}{\Pi}$
for some $\rho \in \Sol{\cdom}{\Omega}$ such that $\domset{\rho} = \varset{\Omega}$.
We want to to show the existence of a $\QCHL(\qdom,\cdom)$ proof tree $T$
witnessing $\Prog \qchldc \cqat{A}{d}{\Pi}$.
We reason  by complete induction on $\Vert T' \Vert$.
There are three possible cases according to the the syntactic form of the atom $A'$\!.
In each case we argue how to build the desired proof tree $T$.

% [2 => 1] Primitive.
\noindent --- $A'$ is a primitive atom $\kappa$.
In this case due to {\bf TQCA} and {\bf TPA} we can assume $A = \kappa$ and $\Omega = \{\qval{\imath(\tp)},\ \qbound{\imath(d),\imath(\tp),\imath(\tp)}\}$.
Note that $\domset{\rho} = \varset{\Omega} = \emptyset$ implies $\rho = \varepsilon$.
Now, from $\elimD{\Prog} \chlc \cat{\kappa\varepsilon}{\Pi}$ follows $\Pi \model{\cdom} \kappa$ due to the {\bf PA} inference, and therefore we can prove $\Prog \qchldc \cqat{\kappa}{d}{\Pi}$ with a proof tree $T$ containing only one {\bf QPA} node.

% [2 => 1] Equation.
\noindent --- $A'$ is an equation $t == s$.
In this case due to {\bf TQCA} and {\bf TEA} we can assume $A = (t == s)$ and $\Omega = \{\qval{\imath(\tp)},\ \qbound{\imath(d),\imath(\tp),\imath(\tp)}\}$.
Note that $\domset{\rho} = \varset{\Omega} = \emptyset$ implies $\rho = \varepsilon$.
Now, from $\elimD{\Prog} \chlc \cat{(t == s)\varepsilon}{\Pi}$ follows $t \approx_\Pi s$ due to the {\bf EA} inference, and therefore we can prove $\Prog \qchldc \cqat{(t == s)}{d}{\Pi}$ with a proof tree $T$ containing only one {\bf QEA} node.

% [2 => 1] Defined.
\noindent --- $A'$ is a defined atom $p'(\ntup{t'}{n},W)$ with $p' \in DP^{n+1}$.
In this case due to {\bf TQCA} and {\bf TDA} we can assume $A = p(\ntup{t'}{n})$ and $\Omega = \{\qval{W},\ \qbound{\imath(d),\imath(\tp),W}\}$.
On the other hand, $T'$ must be rooted by a {\bf DA} step ($\spadesuit$) using a clause $\transform{C} \in \elimD{\Prog}$ instantiated by a substitution $\theta'$. We can assume that ($\spadesuit$), $\transform{C}$ and the corresponding clause $C \in \Prog$ have the form already displayed in [1. $\Rightarrow$ 2.].

By construction of $\transform{C}$\!, we can assume $\transform{B_j} = (B'_j,\ w'_j)$.
Let $\theta = \theta'{\upharpoonright}\varset{C}$ and $\rho_j = \theta'{\upharpoonright}\varset{w'_j} ~ (1 \ge j \ge m)$.
Then, due to the premises $\cat{qV\!al(w'_j)\theta'}{\Pi}$ of ($\spadesuit$) and Lemma \ref{lema:expr} we can assume $e'_j \in \aqdom ~ (1 \leq j \leq m)$ such that $w'_j\rho_j = \imath(e'_j)$.

To finish the proof, we must prove $\Prog \qchldc \cqat{A}{d}{\Pi}$.
We claim that this can be done with a $\QCHL(\qdom,\cdom)$ proof tree $T$ whose root inference is a {\bf QDA} step of the form of ($\clubsuit$), as displayed in [1. $\Rightarrow$ 2.], using clause $C$ instantiated by $\theta$.
In the premises of this inference we choose $d_i = \tp ~ (1 \leq i \leq n)$ and $e_j = e'_j ~ (1 \leq j \leq m)$.
Next we check that these premises can be derived from $\Prog$ in $\QCHL(\qdom,\cdom)$ and that the side conditions are fulfilled:
\begin{itemize}
\item
$\Prog \qchldc \cqat{(t'_i == t_i\theta)}{d_i}{\Pi}$ for $i = 1 \ldots n$.
This amounts to $t'_i \approx_\Pi t_i\theta$ which follows from the first $n$ premises of ($\spadesuit$) given that $t'_i\rho = t'_i$ and $t_i\theta' = t_i\theta$.
\item
$\Prog \qchldc \cqat{B_j\theta}{e_j}{\Pi}$ for $j = 1 \ldots m$.
From $\transform{B_j} = (B_j',  w'_j)$ and due to rule {\bf TQCA}, we have $\transform{(\cqat{(B_j\theta)}{e_j}{\Pi})} = (\cat{B_j\theta}{\Pi}, \Omega_j)$ where $\Omega_j = \{\qval{w'_j},\ \qbound{\imath(e_j), \imath(\tp),$ $w'_j}\}$. From the premises of ($\spadesuit$) and the fact that $B'_j\theta' = B'_j\theta\rho_j$ we know that $\elimD{\Prog} \chlc \cat{B'_j\theta\rho_j}{\Pi}$ with a CHL($\cdom$) proof tree $T'_j$ such that $\Vert T'_j \Vert < \Vert T' \Vert$. Therefore $\Prog \qchldc \cqat{B_j\theta}{e_j}{\Pi}$ follows by inductive hypothesis provided that $\rho_j \in \Solc{\Omega_j}$. In fact, due to the form of $\Omega_j$, $\rho_j \in \Solc{\Omega_j}$ holds iff $w'_j\rho_j = \imath(e'_j)$ for some $e'_j$ such that $e_j \dleq e'_j$, which is the case because of the choice of $e_j$.
\item
$e_j \dgeq^? w_j$ for $j = 1 \ldots m$.
Trivial in the case that $w_j =\ ?$.
Otherwise they are equivalent to $w_j \dleq e'_j$ which follow from premises $\cat{\encode{w'_j \dgeq^? \imath(w_j)}\theta'}{\Pi}$ (i.e. $\cat{\encode{w'_j\rho_j \dgeq^? \imath(w_j)}}{\Pi}$) of ($\spadesuit$) and Lemma \ref{lema:expr}.
\item
$d \dleq d_i$ for $i = 1 \ldots n$.
Trivially hold due to the choice of $d_i = \tp$.
\item
$d \dleq \alpha \circ e_j$ for $j = 1 \ldots m$.
Note that $\rho \in \Solc{\Omega}$ implies the existence of $d' \in \aqdom$ such that $\imath(d') = W\rho$ and $d \dleq d'$.
On the other hand, $e_j = e'_j$ by choice.
It suffices to prove $d' \dleq \alpha \circ e'_j$ for $j = 1 \ldots m$. Premises of ($\spadesuit$) and Lemma \ref{lema:expr} imply that $\qbound{W'\theta',\imath(\alpha),w'_j\theta'}$ is true in $\cdom$. Moreover, $W'\theta' = W\rho = \imath(d')$ because of another premise of ($\spadesuit$) and $w'_j\theta' = \imath(e'_j)$ as explained above.
Therefore $\qbound{W'\theta',\imath(\alpha),w'_j\theta'}$ amounts to $\qbound{\imath(d'),\imath(\alpha),\imath(e'_j)}$ which guarantees $d' \dleq \alpha \circ e'_j ~ (1 \leq j \leq m)$. \mathproofbox
\end{itemize}
\end{proof*}

The goal transformation correctness is established by the next theorem, which relies on the previous result.

% Theorem: Goal transformation correctness
\begin{thm}
\label{thm:QCLP2CLP:goals}
Let $G$ be a goal for a $\qclp{\qdom}{\cdom}$-program $\Prog$ such that $qV\!al$ and $qBound$ do not occur in $G$.
Let $\Prog' = \elimD{\Prog}$ and $G' = \elimD{G}$.
Assume a $\cdom$-substitution $\sigma$, a mapping $\mu : \warset{G} \to \aqdomd{\qdom}$ and a satisfiable finite set of $\cdom$-constraints $\Pi$.
Then, the following two statements are equivalent:
\begin{enumerate}
\item
$\langle \sigma, \mu, \Pi \rangle \in \Sol{\Prog}{G}$.
\item
$\langle \theta, \Pi \rangle \in \Sol{\Prog'}{G'}$ for some $\theta$ that verifies the following requirements:
  \begin{enumerate}
  \item
  $\theta =_{\varset{G}} \sigma$,
  \item
  $\theta =_{\warset{G}} \mu\imath$ and
  \item
  $W\theta \in \mbox{ran}(\imath)$ for each $W \in \varset{G'} \setminus (\varset{G} \cup \warset{G})$.
  \end{enumerate}
\end{enumerate}
\end{thm}

\begin{proof*}
As explained in Subsection \ref{sec:cases:qclp} the syntax of goals in $\qclp{\qdom}{\cdom}$-programs is the same as that of goals for $\sqclp{\simrel}{\qdom}{\cdom}$-programs, which is described in Section \ref{sec:sqclp}.
Therefore $G$, and $G'$ due to rule {\bf TG}, must have the following form:
$$
\begin{array}{c@{\hspace{1mm}}c@{\hspace{1mm}}l}
G &:& (~\qat{B_j}{W_j},\ W_j \dgeq^? \!\beta_j~)_{j = 1 \ldots m} \\
G' &:& (~qV\!al(W_j),~  \encode{W_j\dgeq^? \imath(\beta_j)},~  qV\!al(w'_j),~ qBound(W_j,$ $\imath(\tp),w'_j),~  B_j'~)_{j=1 \dots m} \\
\end{array}
$$
with $\transform{B_j} = (B'_j,  w'_j) ~ (1 \leq j \leq m)$.
Note that, because of rule {\bf TQCA}, we have $\transform{(\cqat{B_j\sigma}{W_j\mu}{\Pi})} = (\cat{B_j'\sigma}{\Pi}, \Omega_j)$ with $\Omega_j = \{\qval{w'_j},\ \qbound{\imath(W_j\mu),\imath(\tp),$ $w'_j}\}$ for $j = 1 \ldots m$.
We now prove each implication.

% 1. => 2.
\smallskip\noindent
[1. $\Rightarrow$ 2.]
Let  $\langle \sigma, \mu, \Pi \rangle \in \Sol{\Prog}{G}$.
This means, by Definition \ref{dfn:qclp-goalsol}, $W_j\mu \dgeq^? \!\beta_j$ and $\Prog \qchldc \cqat{B_j\sigma}{W_j\mu}{\Pi}$ for $j = 1 \ldots m$.
In these conditions, Theorem \ref{thm:QCLP2CLP:programs} guarantees $\Prog' \chlc \cat{B'_j\sigma\rho_j}{\Pi} ~ (1 \leq j \leq m)$ for some $\rho_j \in \Solc{\Omega_j}$ such that $\domset{\rho_j} = \varset{\Omega_j}$.
It is easy to see that $\varset{G'} \setminus (\varset{G} \cup \warset{G}) = \varset{\Omega_1} \uplus \cdots \uplus \varset{\Omega_m}$.
Therefore it is possible to define a substitution $\theta$ verifying $\theta =_{\varset{G}} \sigma$, $\theta =_{\warset{G}} \mu\imath$ and $\theta =_{\domset{\rho_j}} \rho_j ~ (1 \leq j \leq m)$. Trivially, $\theta$ satisfies conditions 2.(a) and 2.(b). It also satisfies condition 2.(c) because for any $j$ and any variable $X$ such that $X \in \varset{\Omega_j}$, we have a constraint $\qval{X} \in \Omega_j$ implying, due to Lemma \ref{lema:expr}, $X\rho_j \in \mbox{ran}(\imath)$ (because $\rho_j \in \Solc{\Omega_j}$).

In order to prove $\langle \theta, \Pi \rangle \in \Sol{\Prog'}{G'}$ in the sense of Definition \ref{dfn:clp-goalsol} we check the following items:
\begin{itemize}
\item
By construction, $\theta$ is a $\cdom$-substitution.
\item
By the theorem's assumptions, $\Pi$ is a satisfiable and finite set of $\cdom$-constraints.
\item
$\Prog' \chlc \cat{A\theta}{\Pi}$ for every atom $A$ in $G'$.
Because of the form of $G'$ we have to prove the following for any fixed $j$:
\begin{itemize}
\item
$\Prog' \chlc \cat{qV\!al(W_j)\theta}{\Pi}$.
By construction of $\theta$ and Lemma \ref{lema:expr}, this amounts to $\qval{\imath(W_j\mu)}$ being true in $\cdom$, which is trivial consequence of $W_j\mu \in \aqdom$.
\item
$\Prog' \chlc \cat{\encode{W_j \dgeq^? \!\imath(\beta_j)}\theta}{\Pi}$.
If $\beta_j =\ ?$ this becomes trivial. Otherwise, $W_j\theta = \imath(W_j\mu)$ by construction of $\theta$, and by Lemma \ref{lema:expr} it suffices to prove $\qbound{\imath(\beta_j),\imath(\tp),\imath(W_j\mu)}$ is true in $\cdom$. This follows from $W_j\mu \dgeq^? \!\beta_j$, that is ensured by $\langle \sigma, \mu, \Pi \rangle \in \Sol{\Prog}{G}$.
\item
$\Prog' \chlc \cat{qV\!al(w'_j)\theta}{\Pi}$.
By construction of $\theta$ and Lemma \ref{lema:expr}, this amounts to $\qval{w'_j\rho_j}$ being true in $\cdom$, that is guaranteed by $\rho_j \in \Solc{\Omega_j}$.
\item
$\Prog' \chlc \cat{qBound(W_j,\imath(\tp),w'_j)\theta}{\Pi}$.
By construction of $\theta$ and Lemma \ref{lema:expr}, this amounts to $\qbound{\imath(W_j\mu),\imath(\tp),w'_j\rho_j}$ being true in $\cdom$, that is also guaranteed by $\rho_j \in \Solc{\Omega_j}$.
\item
$\Prog' \chlc \cat{B'_j\theta}{\Pi}$.
Note that, by construction of $\theta$, $B'_j\theta = B'_j\sigma\rho_j$. On the other hand, $\rho_j$ has been chosen above to verify $\Prog' \chlc \cat{B'_j\sigma\rho_j}{\Pi}$.
\end{itemize}
\end{itemize}

% 2. => 1.
\smallskip\noindent
[2. $\Rightarrow$ 1.]
Let $\langle \theta, \Pi \rangle \in \Sol{\Prog'}{G'}$ and assume that $\theta$ verifies 2.(a), 2.(b) and 2.(c).
In order to prove $\langle \sigma, \mu, \Pi \rangle \in \Sol{\Prog}{G}$ in the sense of Definition \ref{dfn:qclp-goalsol} we must prove the following items:
\begin{itemize}
\item
By the theorem's assumptions, $\sigma$ is a $\cdom$-substitution, $\mu : \warset{G} \to \aqdomd{\qdom}$ and $\Pi$ is a satisfiable finite set of $\cdom$-constraints.
\item
$W_j\mu \dgeq^? \!\beta_j$.
We reason for any fixed $j$.
If $\beta_j =\ ?$ this results trivial.
Otherwise, we have $\Prog' \chlc \cat{\encode{W_j \dgeq^? \imath(\beta_j)}\theta}{\Pi}$ which, by condition 2.(b) and Lemma \ref{lema:expr} amounts to $\qbound{\imath(\beta_j),\imath(\tp),\imath(W_j\mu)}$ is true $\cdom$, i.e. $W_j\mu \dgeq \beta_j$.
\item
$\Prog \qchldc \cqat{B_j\sigma}{W_j\mu}{\Pi}$ for $j = 1 \ldots m$.
We reason for any fixed $j$. Let $\rho_j$ be the restriction of $\theta$ to $\varset{\Omega_j}$.
Then, $\Prog' \chlc \cat{B'_j\sigma\rho_j}{\Pi}$ follows from $\langle \theta,\Pi \rangle \in \Sol{\Prog'}{G'}$ and $B'_j\theta = B'_j\sigma\rho_j$.
Therefore, $\Prog \qchldc \cqat{B_j\sigma}{W_j\mu}{\Pi}$ follows from Theorem 5.3 provided that $\rho_j \in \Solc{\Omega_j}$.
By Lemma \ref{lema:expr} and the form of $\Omega_j$, $\rho_j \in \Solc{\Omega_j}$ holds iff $\Prog' \chlc \cat{qV\!al(w'_j\rho_j)}{\Pi}$ and $\Prog' \chlc \cat{qBound(\imath(W_j\mu),\imath(\tp),w'_j\rho_j)}{\Pi}$, which is true because $\langle \theta,\Pi \rangle \in \Sol{\Prog'}{G'}$ and construction of $\rho_j$. \mathproofbox
\end{itemize}
\end{proof*}
 %\ref{sec:implemen:QCLP2CLP}
% ----
% Subsection 5.3: Solving SQCLP Goals
% ----
\subsection{Solving SQCLP Goals}
\label{sec:implemen:solving}

In this subsection we show that the transformations from the two previous subsections can be used to specify abstract goal solving systems for SQCLP and arguing about their correctness.
In the sequel we consider a given $\sqclp{\simrel}{\qdom}{\cdom}$-program $\Prog$\! and a goal $G$ for $\Prog$ whose atoms are all relevant for $\Prog$.
We also consider $\Prog' \!= \elimS{\Prog}$, $G' = \elimS{G}$, $\Prog'' \!= \elimD{\Prog'}$ and $G'' = \elimD{G'}$.
Due to the definition of both elim$_\simrel$ and elim$_\qdom$, we can assume:
$$
\begin{array}{c@{\hspace{1mm}}c@{\hspace{1mm}}l}
G &:& (~ \qat{A_i}{W_i},~ W_i \dgeq^? \!\beta_i ~)_{i = 1 \ldots m} \\
G' &:& (~ \qat{A^i_\sim}{W_i},~ W_i \dgeq^? \!\beta_i ~)_{i = 1 \ldots m} \\
G'' &:& (~ qV\!al(W_i),~ \encode{W_i \dgeq^? \imath(\beta_i)},~ qV\!al(w'_i),~ qBound(W_i, \imath(\tp), w'_i),~ A_i' ~)_{i=1 \dots m} \\
&& \mbox{ where } \transform{A_i} = (A'_i, w'_i).
\end{array}
$$
In the particular case that the $G$ is a unification problem,
all atoms $A_i,\, i = 1 \ldots m,$ are equations $t_i == s_i$
and  $G''$ is such that $w'_i$ is a fresh CLP variable $W'_i$ and $A'_i$ has the form $\sim'(t_i, s_i, W'_i)$,
for all $i = 1 \ldots m$. Unification problems will be important for some examples when discussing our practical implementation in Section \ref{sec:practical}.

Next, we present an auxiliary result.

\begin{lem}
\label{lem:soltoground}
Assume $\Prog$\!, $G$, $\Prog'$\!, $G'\!$, $\Prog''$ and $G''$\! as above.
Let $\langle \sigma', \Pi \rangle \in \Sol{\Prog''}{G''}$, $\nu \in \Solc{\Pi}$ and $\theta = \sigma'\nu$.
Then $\langle \theta, \Pi \rangle \in \Sol{\Prog''}{G''}$. Moreover, $W \theta \in \mbox{ran}(\imath)$ for every $W \in \varset{G''} \setminus \varset{G}$.\footnote{Note that $\warset{G} \subseteq \varset{G''} \setminus \varset{G}$.}
\end{lem}
\begin{proof}
Consider an arbitrary atom $A''$\! occurring in $G''$\!.
Because of $\langle \sigma', \Pi \rangle \in \Sol{\Prog''}{G''}$ we have $\Prog \chlc \cat{A''\sigma'}{\Pi}$.
On the other hand, because of $\nu \in \Solc{\Pi}$ we have $\emptyset \model{\cdom} \Pi\nu$ and therefore also $\Pi \model{\cdom} \Pi\nu$.
This and Definition 3.1(4) of \cite{RR10TR} ensure $\cat{A''\sigma'}{\Pi} \entail{\cdom} \cat{A''\sigma'\nu}{\Pi}$, i.e. $\cat{A''\sigma'}{\Pi} \entail{\cdom} \cat{A''\theta}{\Pi}$.
This fact, $\Prog'' \chlc \cat{A''\sigma'}{\Pi}$ and the Entailment Property for Programs in $\clp{\cdom}$ imply $\Prog'' \chlc \cat{A''\theta}{\Pi}$.
Therefore, $\langle \theta,\Pi \rangle \in \Sol{\Prog''}{G''}$.

Consider now any $W \in \varset{G''} \setminus \varset{G}$.
By construction of $G''$\!, one of the atoms occurring in $G''$ is $qV\!al(W)$.
Then, due to $\langle \sigma'\Pi \rangle \in \Sol{\Prog''}{G''}$ we have $\Prog'' \chlc \cat{qV\!al(W\sigma')}{\Pi}$.
Because of Lemma \ref{lema:dec}(1) this implies $\Pi \model{\cdom} \qval{W\sigma'}$, i.e. $\Solc{\Pi} \subseteq \Solc{\qval{W\sigma'}}$.
Since $\nu \in \Solc{\Pi}$ we get $\nu \in \Solc{\qval{W\sigma'}}$, i.e. $W\sigma'\nu \in \mbox{ran}(\imath)$.
Since $W\sigma'\nu = W\theta$, we are done.
\end{proof}

Now, we  can explain how to define an abstract goal solving system for SQCLP from a given abstract goal solving system for CLP.

\begin{defn}
\label{dfn:SQCLPGSS}
Let $\mathcal{CA''}$ be an abstract goal solving system for $\clp{\cdom}$ (in the sense of Definition \ref{dfn:clp-goalsolsys}).
Then we define $\mathcal{CA}$ as an abstract goal solving system for $\sqclp{\simrel}{\qdom}{\cdom}$
(in the sense of Definition \ref{dfn:goalsolsys}) that works as follows:
\begin{enumerate}
\item
Given a goal $G$ for the $\sqclp{\simrel}{\qdom}{\cdom}$-program $\Prog$, consider $\Prog'$\!, $G'$\!, $\Prog''$\! and $G''$ as explained at the beginning of the subsection.
\item
For each  $\langle \sigma',\Pi \rangle \in \mathcal{CA''}_{\Prog''}(G'')$ and for any $\nu \in \Solc{\Pi}$,
let $\langle \sigma, \mu, \Pi \rangle \in \mathcal{CA}_{\Prog}(G)$,
where $\theta = \sigma'\nu$, $\sigma = \theta{\upharpoonright}\varset{G}$ and
$\mu = \theta\imath^{-1}{\upharpoonright}\warset{G}$.
Note that $\mu$ is well-defined thanks to Lemma \ref{lem:soltoground}.
\item
All the computed answers belonging to  $\mathcal{CA}_{\Prog}(G)$ are obtained as described in the previous item. \mathproofbox
\end{enumerate}
\end{defn}

The next theorem ensures that $\mathcal{CA}$ is correct provided that $\mathcal{CA''}$ is also correct.
The proof relies on the semantic results of the two previous subsections.

\begin{thm}[Correct Abstract Goal Solving Systems for SQCLP]
\label{thm:SQCLP-AGSS:correctness}
Let $\mathcal{CA}$ be obtained from $\mathcal{CA''}$ as in the previous definition.
Assume that $\mathcal{CA''}$ is correct as specified in Definition \ref{dfn:clp-goalsolsys}(3).
Then $\mathcal{CA}$ is correct as specified in  Definition \ref{dfn:goalsolsys}(4).
\end{thm}
\begin{proof*}
We separately prove that $\mathcal{CA}$ is {\em sound} and {\em weakly complete}.

% $\mathcal{CA}$ is sound.
\smallskip\noindent
--- $\mathcal{CA}$ {\em is sound.}
Assume $\langle \sigma, \mu, \Pi \rangle \in \mathcal{CA}_{\Prog}(G)$.
We must prove that $\langle \sigma, \mu, \Pi \rangle \in \Sol{\Prog}{G}$.
Because of Definition \ref{dfn:SQCLPGSS} there exist $\langle \sigma', \Pi \rangle \in \mathcal{CA''}_{\Prog''}(G'')$ and $\nu \in \Solc{\Pi}$
such that $\sigma = \theta{\upharpoonright}\varset{G}$ and $\mu = \theta\imath^{-1}{\upharpoonright}\warset{G}$ with $\theta = \sigma'\nu$.
By the soundness of $\mathcal{CA''}$ we get  $\langle \sigma', \Pi \rangle \in \Sol{\Prog''}{G''}$.
Moreover, because of Lemma \ref{lem:soltoground} we have $\langle \theta, \Pi \rangle \in \Sol{\Prog''}{G''}$ and $W\theta \in \mbox{ran}(\imath)$ for every $W \in \varset{G''} \setminus \varset{G}$.
Note that:
\begin{itemize}
\item
$\theta =_{\varset{G'}} \sigma$.
This follows from $\varset{G'} = \varset{G}$ and the construction of $\sigma$.
\item
$\theta =_{\warset{G'}} \mu\imath$.
This follows from $\warset{G'} = \warset{G}$ and $\theta =_{\warset{G}} \mu\imath$, that is obvious from the construction of $\mu$.
\item
$W\theta \in \mbox{ran}(\imath)$ for each $W \in \varset{G''} \setminus (\varset{G'} \cup \warset{G'})$.
This is a consequence of Lemma \ref{lem:soltoground} since $\varset{G''} \setminus (\varset{G'} \cup \warset{G'}) \subseteq \varset{G''} \setminus \varset{G'}$ and $\varset{G'} = \varset{G}$.
\end{itemize}
From the previous items and Theorem \ref{thm:QCLP2CLP:goals} we get $\langle \sigma, \mu, \Pi \rangle \in \Sol{\Prog'}{G'}$, which trivially implies $\langle \sigma, \mu, \Pi \rangle \in \Sol{\Prog}{G}$ because of Theorem \ref{thm:SQCLP2QCLP:goals}.

% $\mathcal{CA}$ is weakly complete.
\smallskip\noindent
--- $\mathcal{CA}$ {\em is weakly complete.}
Let $\langle \eta, \rho, \emptyset \rangle \in \GSol{\Prog}{G}$ be  a ground solution for $G$ w.r.t. $\Prog$.
We must prove that it is subsumed by some computed answer $\langle \sigma, \mu, \Pi \rangle \in \mathcal{CA}_{\Prog}(G)$.
By Theorem \ref{thm:SQCLP2QCLP:goals} we have that $\langle \eta, \rho, \emptyset \rangle$ is also a ground solution for $G'$ w.r.t. $\Prog'$.
Then by Theorem \ref{thm:QCLP2CLP:goals} we get $\langle \eta', \emptyset \rangle \in \Sol{\Prog''}{G''}$ for some $\eta'$ such that
\begin{itemize}
  \item (1) $\eta' =_{\varset{G'}} \eta$,
  \item (2) $\eta' =_{\warset{G'}} \rho\imath$ and hence $\eta'(\imath^{-1}) =_{\warset{G'}} \rho$, and
  \item $W\eta' \in \mbox{ran}(\imath)$ for each $W \in \varset{G''} \setminus (\varset{G'} \cup \warset{G'})$ (i.e. $w'_i\eta' \in \mbox{ran}(\imath)$ for each $i = 1 \ldots m$ such that $w'_i$ is a variable).
\end{itemize}
By construction of $\eta'$, it is clear that $\langle \eta', \emptyset \rangle$ is ground.
Now, by the weak completeness of $\mathcal{CA''}$, there is some computed answer $\langle \sigma',\Pi \rangle \in \mathcal{CA''}_{\Prog''}(G'')$
subsuming $\langle \eta',\emptyset \rangle$ in the sense of Definition \ref{dfn:clp-goalsol}(5), therefore satisfying:
\begin{itemize}
  \item (3) there is some $\nu \in \Solc{\Pi}$, such that
  \item (4) $\eta' =_{\varset{G''}} \sigma'\nu$.
\end{itemize}

Because of Definition \ref{dfn:SQCLPGSS} one can build a computed answer
$\langle \sigma, \mu, \Pi \rangle \in \mathcal{CA}_{\Prog}(G)$ as follows:
\begin{itemize}
  \item (5) $\sigma = \sigma'\nu{\upharpoonright}\varset{G}$
  \item (6) $\mu = \sigma'\nu\imath^{-1}{\upharpoonright}\warset{G}$
\end{itemize}

We now check that $\langle \sigma, \mu, \Pi \rangle$ subsumes $\langle \eta, \rho, \emptyset \rangle$ in the sense of Definition \ref{dfn:goalsol}(4):
\begin{itemize}
  \item $W_i\rho \dleq W_i\mu$ and even $W_i\rho = W_i\mu$ because:
  $$W_i\rho =_{(2)} W_i\eta'(\imath^{-1}) =_{(4)} W_i\sigma'\nu(\imath^{-1}) =_{(6)} W_i\mu \enspace .$$
  \item $\nu \in \Solc{\Pi}$ by (3) and, moreover, for any $X \in \varset{G}$:
  $$ X\eta =_{(1)} X\eta' =_{(4)} X\sigma'\nu =_{(\dagger)} X\sigma'\nu\nu =_{(5)} X\sigma\nu$$
  therefore $\eta =_{\varset{G}} \sigma\nu$.

  The step ($\dagger$) is justified because $\nu \in \mbox{Val}_\cdom$ implies $\nu = \nu\nu$. \mathproofbox
\end{itemize}
\end{proof*}

As an immediate consequence of Theorem \ref{thm:SQCLP-AGSS:correctness} and Lemma \ref{lema:flexrestr}, we obtain:

\begin{cor}[Flexibly Correct Abstract Goal Solving Systems for SQCLP]
\label{thm:SQCLP-AGSS:flex-correctness}
Let $\mathcal{CA}$ be obtained from $\mathcal{CA''}$ as in the Definition \ref{dfn:SQCLPGSS}.
Assume that $\mathcal{CA''}$ is correct as specified  in Definition \ref{dfn:clp-goalsolsys}(3).
Then  any flexible restriction $\mathcal{FCA}$ of  $\mathcal{CA}$ is correct in the flexible sense
as specified  in Definition \ref{dfn:goalsolsys}(5). \mathproofbox
\end{cor}
 %\ref{sec:implemen:solving}
 %\ref{sec:implemen}

% Section 5: A Practical Implementation in Prolog
% ----
% Section 5: A Practical Implementation
% ----
\section{A Practical Implementation}
\label{sec:practical}

This section is devoted to the more practical aspects of the SQCLP programming scheme.
We present a {\tt Prolog}-based prototype system that relies on the transformation techniques from Section \ref{sec:implemen} 
and supports several useful SQCLP instances.
The presentation is developed in three subsections. 
Subsection \ref{sec:practical:SQCLP} discusses in some detail how to bridge the gap between the abstract goal solving systems for SQCLP discussed in Subsection  \ref{sec:implemen:solving} and a practical {\tt Prolog}-based  implementation.
Subsection \ref{sec:practical:prototype} gives a user-oriented presentation of our prototype implementation, 
explaining how to write programs and how to solve goals.
Finally, in Subsection \ref{sec:practical:efficiency}  we study  the unavoidable overload caused by the implementation of  qualification and proximity relations in our system. The overload is shown in experimental results on the execution of some SQCLP programs which make only a trivial use of qualification and proximity. 

% Subsections
% ----
% Subsection 5.1: SQCLP over a CLP Prolog System
% ----
\subsection{SQCLP over a CLP Prolog System}
\label{sec:practical:SQCLP}

Our aim is to implement a goal solving system for SQCLP on top of an available  CLP {\tt Prolog} system,
taking the definitions and results from Subsection \ref{sec:implemen:solving} as a theoretical guideline.
Therefore, given a $\sqclp{\simrel}{\qdom}{\cdom}$-program $\Prog$ and a goal $G$ for $\Prog$,
the following steps should be carried out:
\begin{enumerate}
\item[(i)]
Apply the transformation elim$_\simrel$ specified in Definition \ref{def:sqclptransform}, obtaining the $\qclp{\qdom}{\cdom}$ program $\Prog' \!= \elimS{\Prog}$ and the $\qclp{\qdom}{\cdom}$ goal $G' \!= \elimS{G}$, where $G$ and $G'$ are as displayed at the beginning of Section \ref{sec:implemen:solving}, $\Prog'$ is of the form $EQ_\simrel \cup \hat{\Prog}_\simrel$, $EQ_\simrel$ is obtained following Definition \ref{def:EQ} and $\hat{\Prog}_\simrel$ is obtained following Definition \ref{def:sqclptransform}(3,2).
\item[(ii)]
Apply the transformation elim$_\qdom$ specified in Definition \ref{def:qclptransform}, obtaining the $\clp{\cdom}$-program $\Prog'' \!= \elimD{\Prog'}$ and the $\clp{\cdom}$-goal $G'' \!= \elimD{G'}$, where $G'$ and $G''$ (obtained from $G'$ by the goal transformation rules shown in Figure \ref{fig:transformation}) are as displayed at the beginning of Section \ref{sec:implemen:solving} and $\Prog''$ is built  according to Definition \ref{def:qclptransform}, by adding  the two clauses of the program $E_\qdom$ to the result of applying the program transformation rules shown in Figure \ref{fig:transformation} to the program $\Prog'$. In particular, $\Prog''$ includes as a subset the set $EQ_\simrel^{'}$ of $\clp{\cdom}$-clauses obtained by applying  the transformation rules from Figure \ref{fig:transformation}
to the set of $\qclp{\qdom}{\cdom}$-clauses $EQ_\simrel$.
\item[(iii)]
Use the available CLP {\tt Prolog} system to compute answers for the CLP goal $G''$ by executing the CLP program $\Prog''$.
\end{enumerate}
Following these steps literally would lead to a set of computed answers representing the behaviour of the abstract goal solving system $\mathcal{CA}$ from Definition \ref{dfn:SQCLPGSS}%
\footnote{Each answer $\langle \sigma',\Pi \rangle$ produced by the CLP system and shown to the user in step (iii)  serves as a compact representation of all answers of the form $\langle \sigma, \mu, \Pi \rangle \in \mathcal{CA}_{\Prog}(G)$, where $\theta = \sigma'\nu$, $\sigma = \theta{\upharpoonright}\varset{G}$, $\mu = \theta\imath^{-1}{\upharpoonright}\warset{G}$, and $\nu \in \Solc{\Pi}$ ranges over the solutions of $\Pi$.},
whose correctness has been proved in Theorem \ref{thm:SQCLP-AGSS:correctness}.
Therefore, the resulting implementation would be correct---i.e. both sound and weakly complete---in the sense of Definition \ref{dfn:goalsolsys}, except for the unavoidable failures in completeness due to {\tt Prolog}'s computation strategy and the incompleteness of the constraint solvers provided by practical CLP {\tt Prolog} systems.

However, our {\tt Prolog}-based implementation---presented in Subsection  \ref{sec:practical:prototype}---differs from the literal application of  step (ii)  in some aspects  concerning an optimized implementation of the CLP clauses in the sets $E_\qdom$ and $EQ_\simrel^{'}$.
In the rest of this subsection we  explain the optimizations and we discuss their influence on the correctness (i.e. soundness and weak completeness) of goal solving.
Subsections \ref{sec:practical:SQCLP:optimize1} and \ref{sec:practical:SQCLP:optimize2} below present some straightforward optimizations of the CLP clauses in $E_\qdom$ and $EQ_\simrel^{'}$, respectively, while Subsection \ref{sec:practical:SQCLP:eqsimrel} discusses three possible {\tt Prolog} implementations of the optimized set $EQ_\simrel^{'}$ obtained in Subsection \ref{sec:practical:SQCLP:optimize2}: a na\"ive one---called {\bf (A)}---that causes very inefficient computations and is not supported by our system;
and two optimized ones---called {\bf (B)} and {\bf (C)}---with a better computational behaviour, which are supported by our system.

% Simple Optimizations: $E_\qdom$
\subsubsection{Optimization of the $E_\qdom$ clauses}
\label{sec:practical:SQCLP:optimize1}

Here we present a straightforward optimization of $E_\qdom$ that does not modify the set of computed answers, thus preserving correctness of goal solving.
As explained at the beginning of Section \ref{sec:implemen:QCLP2CLP}, the set $E_\qdom$ contains CLP clauses for
two predicates {\it qVal} (unary) and {\it qBound} (ternary) which allow to represent qualification values from $\qdom$ and the behaviour of $\qdom$'s attenuation operator $\circ$ by means of $\cdom$-constraints.
Recall Example \ref{exmp:eqdom-clauses}, showing the clauses in  $E_\qdom$ for three significative choices of $\qdom$,
namely $\U$, $\W$ and $\U{\otimes}\W$.

Our  prototype system for SQCLP programming  supports SQCLP  instances of the form
$\sqclp{\simrel}{\qdom}{\rdom}$, where $\rdom$ is the real constraint domain and $\qdom$ is any qualification domain that can be built from $\B$, $\U$ and $\W$ by means of the strict cartesian product operation $\otimes$.
Instead of using a different set $E_\qdom$ for each choice of $\qdom$ supported by the system,
our  implementation uses a single set of {\tt Prolog} clauses for  two predicates {\tt qVal} (binary) and {\tt qBound} (quaternary), whose additional argument w.r.t. {\it qVal} and {\it qBound} is used to encode a representation of $\qdom$
in the following way: $\B$, $\U$ and $\W$ are encoded as {\tt b}, {\tt u} and {\tt w}, respectively;
while $\qdom_1 \otimes \qdom_2$ is encoded as an ordered pair built from the encodings of  $\qdom_1$ and $\qdom_2$.
The set of  {\tt Prolog} clauses for {\tt qVal} and {\tt qBound} used in our implementation is as follows%
\footnote{The semantic correctness of these clauses is obvious from the definition of $\B$, $\U$, $\W$ and $\otimes$;
see \cite{RR10TR}  for details.}:
\begin{center}
\footnotesize\it
\renewcommand{\arraystretch}{1.5}
\begin{tabular}{rp{11cm}}
\tiny$E_1$&\tt qVal(b,1).\\
\tiny$E_2$&\tt qVal(u,X) :- \{ X > 0, X =< 1 \}.\\
\end{tabular}
\end{center}
\begin{center}
\footnotesize\it
\renewcommand{\arraystretch}{1.5}
\begin{tabular}{rp{11cm}}
\tiny$E_3$&\tt qVal(w,X) :- \{ X > 0 \}.\\
\tiny$E_4$&\tt qVal((D$_1$,D$_2$),(X$_1$,X$_2$)) :- qVal(D$_1$,X$_1$), qVal(D$_2$,X$_2$).\\[2mm]
\tiny$E_5$&\tt qBound(b,1,1,1).\\
\tiny$E_6$&\tt qBound(u,X,Y,Z) :- \{ X =< Y * Z \}.\\
\tiny$E_7$&\tt qBound(w,X,Y,Z) :- \{ X >= Y + Z \}.\\
\tiny$E_8$&\tt qBound((D$_1$,D$_2$),(X$_1$,X$_2$),(Y$_1$,Y$_2$),(Z$_1$,Z$_2$)) :- qBound(D$_1$,X$_1$,Y$_1$,Z$_1$),\\
&\tt$\qquad$ qBound(D$_2$,X$_2$,Y$_2$,Z$_2$).\\
\end{tabular}
\end{center}
Therefore, calls such as ${\it qVal}(X)$ and ${\it qBound}(X,Y,Z)$ to the $E_\qdom$ predicates are implemented
as ${\it qVal}(b,X)$ and ${\it qBound}(b,X,Y,Z)$, if $\qdom = \B$;
as ${\it qVal}(u,X)$ and ${\it qBound}(u,X,Y,Z)$, if $\qdom = \U$;
as ${\it qVal}(w,X)$ and  ${\it qBound}(w,X,Y,Z)$, if $\qdom = \W$;
as ${\it qVal}((u,w),X)$ and ${\it qBound}((u,w),X,Y, Z)$, if $\qdom = \U \otimes \W$; etc.

In order to simplify the presentation, in the rest of Subsection \ref{sec:practical:SQCLP} we will omit the optimization just discussed,  considering $E_\qdom$ as a set of CLP clauses for a unary predicate {\it qVal} and a ternary predicate {\it qBound}
corresponding to some fixed choice of $\qdom$.

% Simple Optimizations: $EQ_\simrel^{'}$
\subsubsection{Optimization of the $EQ_\simrel^{'}$ clauses}
\label{sec:practical:SQCLP:optimize2}

Now we present a simple optimization of the CLP clauses in $EQ_\simrel^{'}$. Recall that $EQ_\simrel^{'}$ is the set of  $\clp{\cdom}$-clauses obtained by applying the transformation rules in Figure \ref{fig:transformation} to the set $EQ_\simrel$ of $\qclp{\qdom}{\cdom}$-clauses built according to Definition \ref{def:EQ}. Therefore, $EQ_\simrel^{'}$ consists of CLP clauses of the following forms:
\begin{center}
\footnotesize\it
\renewcommand{\arraystretch}{1.5}
\begin{tabular}{rp{11cm}}
\tiny EQ$_1$ & $\sim'$(X, Y, W) $\gets$ qVal(W), X==Y\\
\tiny EQ$_2$ & $\sim'$(u, u\/$'$\!, W) $\gets$ qVal(W), qVal(W\/$'$\!), qBound(W, $\tp$, \!W\/$'$\!), pay\/$'_\lambda$(W\/$'$\!) \\
\tiny EQ$_3$ & $\sim'$(c(\/$\ntup{X}{n}$), c\/$'$\!(\/$\ntup{Y}{\!n}$), W) $\gets$ qVal(W),\\
& $\quad$ qVal(W\/$'$\!), qBound(W, $\tp$, \!W\/$'$\!), pay\/$'_\lambda$(W\/$'$\!),\\
& $\quad$ qVal(W$_1$), qBound(W, $\tp$, \!W$_1$), $\sim'$(X$_1$, Y$_1$, \!W$_1$),\\
& $\quad$ \ldots\\
& $\quad$ qVal(W$_n$), qBound(W, $\tp$, \!W$_n$), $\sim'$(X$_n$, Y$_n$, \!W$_n$)\\
\tiny EQ$_4$ & pay\/$'_\lambda$(W) $\gets$ qVal(W), qBound(W, $\lambda$, $\tp$) \\
\end{tabular}
\end{center}
where clauses of the form $EQ_2$ are one for each $u,u' \in B_\cdom$ such that $\simrel(u,u') = \lambda \neq \bt$; $EQ_3$ are one for each $c,c' \in DC^n$ such that $\simrel(c,c') = \lambda \neq \tp$ (including the case $c=c'$\!, $\simrel(c,c') = \tp \neq \bt$); and $EQ_4$ are one for each $pay_\lambda$ such that there exist $x,y\in S$ satisfying $\simrel(x,y)=\lambda\neq\bt$.

By unfolding the calls to predicates {\it pay$_\lambda$} occurring in the bodies of clauses $EQ_2$ and $EQ_3$
with respect to the clauses $EQ_4$ defining {\it pay$_\lambda$}, all the occurrences of  {\it pay$_\lambda$}---including clauses $EQ_4$ themselves---can be removed. Moreover, the calls to the predicates
{\it qVal} and {\it qBound} occurring in the results of unfolding clauses $EQ_2$ and $EQ_3$ can be further simplified.
Let us illustrate this process with a clause of the form $EQ_2$. The original clause is:
\begin{flushleft}\small\it
$\quad$ $\sim'$(u, u\/$'$\!, \!W) $\gets$ qVal(W), qVal(W\/$'$\!), qBound(W,\,$\tp$,W\/$'$\!), pay\/$'_\lambda$(W\/$'$\!)
\end{flushleft}
which can be transformed into the equivalent clause:
\begin{flushleft}\small\it
$\quad$ $\sim'$(u, u\/$'$\!, \!W) $\gets$ qVal(W), qVal(W\/$'$\!), qBound(W,\,$\tp$,W\/$'$\!), qVal(W\/$'$\!), qBound(W\/$'$\!,\,$\lambda$,\,$\tp$)
\end{flushleft}
by unfolding the predicate call {\it pay\/$'_\lambda$(W\/$'$\!)\/} occurring in its body. Next, removing one of the two repeated predicate calls {\it qVal(W\/$'$\!)\/} in the new body yields the equivalent clause:
\begin{flushleft}\small\it
$\quad$ $\sim'$(u, u\/$'$\!, \!W) $\gets$ qVal(W), qVal(W\/$'$\!), qBound(W,\,$\tp$,W\/$'$\!), qBound(W\/$'$\!,\,$\lambda$,\,$\tp$)
\end{flushleft}
Observing the last clause we note:
\begin{itemize}
\item
The body is logically equivalent to the following formulation:
\begin{flushleft}\small\it
qVal(W) $\land$ $\exists$W\/$'$( qVal(W\/$'$\!) $\land$ qBound(W,\,$\tp$,W\/$'$\!) $\land$ qBound(W\/$'$\!,\,$\lambda$,\,$\tp$) )
\end{flushleft}
\item
The second conjunt above encodes the statement
\begin{flushleft}\small\it
$\exists$W\/$'$( $W' \in \aqdom$ $\land$ $W \dleq \tp \circ W'$ $\land$ $W' \dleq \lambda \circ \tp$ )
\end{flushleft}
Due to the transitivity of $\dleq$, this is equivalent to  $W \dleq \lambda$
and can be encoded as ${\it qBound}(W,\tp,\lambda)$.
\end{itemize}
Therefore, the last clause is equivalent to the following optimized form:
\begin{flushleft}\small\it
$\quad$ $\sim'$(u, u\/$'$\!, \!W) $\gets$ qVal(W), qBound(W,\,$\tp$,\,$\lambda$).
\end{flushleft}
Performing a similar transformation for clauses $EQ_3$ and removing clauses $EQ_4$ leads to an optimized version of the set $EQ_\simrel^{'}$ consisting of clauses of the following forms:
\begin{center}
\footnotesize\it
\renewcommand{\arraystretch}{1.5}
\begin{tabular}{rp{11cm}}
\tiny EQ$_1$ & $\sim'$(X, Y, W) $\gets$ qVal(W), X==Y\\
\tiny EQ$_2$ & $\sim'$(u, u\/$'$\!, W) $\gets$ qVal(W), qBound(W, $\tp$, $\lambda$) \\
\tiny EQ$_3$ & $\sim'$(c(\/$\ntup{X}{n}$), c\/$'$(\/$\ntup{Y}{\!n}$), W) $\gets$ qVal(W), qBound(W, $\tp$, $\lambda$),\\
& $\quad$ qVal(W$_1$), qBound(W, $\tp$, \!W$_1$), $\sim'$(X$_1$, Y$_1$, \!W$_1$),\\
& $\quad$ \ldots\\
& $\quad$ qVal(W$_n$), qBound(W, $\tp$, \!W$_n$), $\sim'$(X$_n$, Y$_n$, \!W$_n$)\\
\end{tabular}
\end{center}
Note that a similar optimization---unfolding of calls to predicates {\it pay$_\lambda$} followed by simplification of calls to predicates {\it qVal} and {\it qBound}---can be done for all those clauses in $\Prog''$ which include calls to predicates {\it pay}$_\lambda$ in their bodies.
The same is true for goals. All $\clp{\cdom}$-goals $G''$ occurring in subsequent examples will be displayed in the optimized form.

Clearly, the optimizations described in this subsection do not modify the set of computed answers.
Therefore, correctness of goal solving is preserved.

% {\tt Prolog} Implementation of the $EQ_\simrel^{'}$ clauses
\subsubsection{{\tt Prolog} Implementation of the optimized $EQ_\simrel^{'}$ clauses}
\label{sec:practical:SQCLP:eqsimrel}

The optimized version of $EQ_\simrel^{'}$ displayed near the end of the previous subsection just consists of clauses for the predicate $\sim'$.
In the sequel, the notation $EQ_\simrel^{'}$ will refer to this optimized version.
We will consider in turn  three possible {\tt Prolog} implementations of the $EQ_\simrel^{'}$ clauses, called {\bf (A)},  {\bf (B)} and {\bf (C)}.
We will give reasons for discarding implementation {\bf (A)}---not supported by our prototype system---and we will discuss the properties of implementations {\bf (B)} and {\bf (C)}---both supported by our system---concerning correectness of goal solving.
At some points, our  discussion will refer to Example \ref{exmp:eqs-implementation}.

The {\tt Prolog} code displayed below is a na\"ive implementation of $EQ_\simrel^{'}$.
Its structure does not directly resemble the  clauses in the set $EQ'_S$, but it serves
as  a first step towards the more practical implementations {\bf (B)} and {\bf (C)} discussed below.
\begin{center}
\footnotesize\it
\renewcommand{\arraystretch}{1.4}
\begin{tabular}{rp{11cm}}
\multicolumn{2}{@{\hspace{0mm}}l}{\normalfont\normalsize {\bf (A)} Na\"ive implementation of $\sim'$.}\\
\tiny $S_1$&\tt $\sim'$(X,Y,W) :- var(X), var(Y), $\sim'_v$(X,Y,W).\\
\tiny $S_2$&\tt $\sim'$(X,Y,W) :- var(X), nonvar(Y), $\sim'_v$(X,Y,W).\\
\tiny $S_3$&\tt $\sim'$(X,Y,W) :- nonvar(X), var(Y), $\sim'_v$(X,Y,W).\\
\tiny $S_4$&\tt $\sim'$(X,Y,W) :- nonvar(X), nonvar(Y), $\sim'_c$(X,Y,W).\\
\tiny $V_1$&\tt $\sim'_v$(X,Y,W) :- qVal(W), X = Y.\\
\tiny $V_2$&\tt $\sim'_v$(X,Y,W) :- $\sim'_c$(X,Y,W).\\
\tiny $C_1$&\tt $\sim'_c$(u,u',W) :- qVal(W), qBound(W,$\tp$,$\lambda$).\\
\tiny $C_2$&\tt $\sim'_c$(c(X$_1$,..,X$_n$),c'(Y$_1$,..,Y$_n$),W) :- qVal(W), qBound(W,$\tp$,$\lambda$),\\
&\tt $\qquad$ qVal(W$_1$), qBound(W,$\tp$,W$_1$), $\sim'$(X$_1$,Y$_1$,W$_1$),\\
&\tt $\qquad$ \ldots\\
&\tt $\qquad$ qVal(W$_n$), qBound(W,$\tp$,W$_n$), $\sim'$(X$_n$,Y$_n$,W$_n$).\\
\end{tabular}
\end{center}
where clauses of the form $C_1$ are one for each $u,u'\in B_\cdom$ such that $\simrel(u,u') = \lambda\neq\bt$,
and clauses of the form $C_2$ are one for each $c,c' \in DC^n$ such that $\simrel(c,c') = \lambda\neq\bt$
(including the case $c=c'$\!, $\simrel(c,c')=\tp\neq\bt$).

We claim that both {\bf (A)} and  $EQ'_S$ compute the same  solutions. In order to understand
that, consider the behaviour of {\bf (A)} when an atom of the form {\tt $\sim'$(X,Y,W)} is to be solved. The {\tt Prolog} metapredicates {\tt var} and {\tt nonvar} are first used to distinguish four possible cases concerning {\tt X} and {\tt Y}.
If either {\tt X} or {\tt Y}, or both, is a variable---more precisely, it is bound to a variable at execution time---then a first answer is computed by clause $V_1$ by performing the normal {\tt Prolog} unification of {\tt X} and {\tt Y}, and clause $V_2$ can invoke clauses $C_1$ and $C_2$ in order to compute additional answers corresponding to non-syntactical unifiers of (the terms bound to) {\tt X} and {\tt Y} modulo the proximity relation $\simrel$.
If neither {\tt X} and {\tt Y} is (bound to) a variable,
then clauses $C_1$ and $C_2$ will compute answers corresponding to the unifiers of (the terms bound to) {\tt X} and {\tt Y} modulo $\simrel$.
Each computed answer also includes the appropriate constraints for the variable {\tt W}, thus representing a qualification level.

As far as permitted by {\tt Prolog}'s computation strategy---which solves goal atoms from left to right and tries to apply program clauses in their textual order---, the answers computed by {\bf (A)} are the same as those which would be computed by  $EQ_\simrel^{'}$.
Therefore, the na\"ive implementation guarantees soundness and weak completeness of goal solving---recall Definition \ref{dfn:goalsolsys}---except for failures in completeness due to {\tt Prolog}'s computation strategy.

As an illustration, let us show the behaviour of implementation {\bf (A)} when solving the unfication problem of
Example \ref{exmp:eqs-implementation}:

\begin{exmp}
\label{exmp:eqs-implementation-1}
Let $\langle \simrel, \U, \rdom \rangle$, $\Prog$ and $G$ be as in Example \ref{exmp:eqs-implementation}.
Then, $G''$ is the following $\clp{\rdom}$-goal:
\[
\begin{array}{l}
{\it qVal}(W_1),\ {\it qBound}(0.8,1,W_1),\ {\sim'}(Y,X,W_1), \\
{\it qVal}(W_2),\ {\it qBound}(0.8,1,W_2),\ {\sim'}(X,b,W_2), \\
{\it qVal}(W_3),\ {\it qBound}(0.8,1,W_3),\ {\sim'}(Y,c,W_3)
\end{array}
\]
In this simple example, the {\tt Prolog}'s computation strategy causes no loss of completeness, and the na\"ive {\tt Prolog} implementation of $\sim'$ allows to compute sol$_i$ ($i = 1,2,3$) as answers for $G''$.
\mathproofbox
\end{exmp}

However, {\tt Prolog}'s computation strategy leads in general to a very poor computational behaviour when executing the {\tt Prolog} code {\bf (A)} for predicate $\sim'$.
As justification for this claim, we argue as follows:

\begin{enumerate}
\item % (1) Infinite computations.
Solving a given $\sqclp{\simrel}{\qdom}{\cdom}$-goal $G$ yields to solving the translated $\clp{\cdom}$-goal $G''$.
As seen in Example \ref{exmp:eqs-implementation-1}, $G''$ may include subgoals such as \[(\star) \quad {\it qVal}(W),\ {\it qBound}(d,\tp,W),\ {\sim'}(X,Y,W)\] with $d \in \aqdom$.
Solving such a subgoal in a {\tt Prolog} system that relies on the na\"ive code {\bf (A)} for the predicate $\sim'$ may lead to compute infinitely many answers.
For instance, assuming a proximity relation $\simrel$ such that $\simrel(c,d) = \simrel(d,c) = \lambda$ with $c,d \in DC^1$, the {\tt Prolog} code {\bf (A)} will include, among others, the following clauses
\begin{center}
\footnotesize\it
\renewcommand{\arraystretch}{1.4}
\begin{tabular}{rp{11cm}}
\tiny $S_1$&\tt $\sim'$(X,Y,W) :- var(X), var(Y), $\sim'_v$(X,Y,W).\\
\tiny $V_1$&\tt $\sim'_v$(X,Y,W) :- qVal(W), X = Y.\\
\tiny $V_2$&\tt $\sim'_v$(X,Y,W) :- $\sim'_c$(X,Y,W).\\
\tiny $C_2$&\tt $\sim'_c$(c(X$_1$),c(Y$_1$),W) :- qVal(W), qBound(W,$\tp$,$\tp$), \\
&\tt$\quad$ qVal(W$_1$), qBound(W,$\tp$,W$_1$), $\sim'$(X$_1$,Y$_1$,W$_1$).\\
\end{tabular}
\end{center}
whose application, in the given textual order, yields to the computation of the following answers:
\begin{itemize}
\item $\langle \{Y \mapsto X\}, \{ W \mapsto \tp\}, \emptyset \rangle$
\item $\langle \{X\mapsto c(A),\ Y \mapsto c(A)\}, \{ W \mapsto \tp \}, \emptyset \rangle$
\item $\langle \{X\mapsto c(c(A)),\ Y \mapsto c(c(A))\}, \{ W \mapsto \tp \}, \emptyset \rangle$
\item \ldots
\end{itemize}
\item % (2) Missing answers
Due to the infinite sequence of {\tt Prolog} computed answers for the goal $(\star)$ shown in the previous item,
{\tt Prolog} never comes to computing other valid solutions for $(\star)$ involving data constructors other than $c$.
More concretely, due to $\simrel(c,d) = \simrel(d,c) = \lambda$, the {\tt Prolog} code {\bf (A)} must include clauses of the following form:
\begin{center}
\footnotesize\it
\renewcommand{\arraystretch}{1.4}
\begin{tabular}{rp{11cm}}
\tiny $C_{2.1}$&\tt $\sim'_c$(c(X$_1$),d(Y$_1$),W) :- qVal(W), qBound(W,$\tp$,$\lambda$), [...].\\
\tiny $C_{2.2}$&\tt $\sim'_c$(d(X$_1$),c(Y$_1$),W) :- qVal(W), qBound(W,$\tp$,$\lambda$), [...].\\
\tiny $C_{2.3}$&\tt $\sim'_c$(d(X$_1$),d(Y$_1$),W) :- qVal(W), qBound(W,$\tp$,$\tp$), [...].\\
\end{tabular}
\end{center}
If all these clauses happen to occur after the clause $C_2$ of item (1) in the textual order, {\tt Prolog}'s computation strategy will never come to the point of trying to apply them to compute answers for $(\star)$.
\end{enumerate}

% (3) Conclusion and example.

Items (1) and (2) above show that the na\"ive implementation of $\sim'$ is inclined to go into infinite computations which may produce infinitely many computed answers of a certain shape,  while failing to compute some other answers needed for completeness. In situations  a bit more complex than the one considered in items (1) and (2) above, this unfortunate behaviour can lead to failure (i.e., compute no answer at all) for goals which do have solutions, as illustrated by the following example:

\begin{exmp}[Failure of the na\"ive implementation of $\sim'$]
\label{exmp:naive-loops}
Consider the admissible triple $\langle \simrel, \U, \rdom \rangle$ where $\simrel$ is a proximity relation such that: $\simrel(f,g) = \simrel(g,f) = 0.8$ and $\simrel(g,h) = \simrel(h,g) = 0.8$ where $f,g,h \in DC^1$.
Assume also a constant $a \in DC^0$.
Let $\Prog$ be the empty program and let $G$ be the following unification problem:
 \[\qat{(X==f(Y))}{W_1},\ \qat{(X==h(Z))}{W_2}  \sep W_1 \ge 0.5,\ W_2 \ge 0.5\]
 Then, using the na\"ive implementation {\bf (A)} of $\sim'$ leads to the following {\tt Prolog} code for the $\clp{\rdom}$-program $\Prog''$:
\begin{center}
\footnotesize\it
\renewcommand{\arraystretch}{1.4}
\begin{tabular}{rp{11cm}}
\tiny 1&\tt qVal(X) :- \{X > 0, X =< 1\}.\\
\tiny 2&\tt qBound(X,Y,Z) :- \{X =< Y * Z\}.\\
\tiny 3&\tt $\sim'$(X,Y,W) :- var(X), var(Y), $\sim'_v$(X,Y,Z).\\
\tiny 4&\tt $\sim'$(X,Y,W) :- var(X), nonvar(Y), $\sim'_v$(X,Y,Z).\\
\tiny 5&\tt $\sim'$(X,Y,W) :- nonvar(X), var(Y), $\sim'_v$(X,Y,Z).\\
\tiny 6&\tt $\sim'$(X,Y,W) :- nonvar(X), nonvar(Y), $\sim'_c$(X,Y,Z).\\
\tiny 7&\tt $\sim'_v$(X,Y,W) :- qVal(W), X = Y.\\
\tiny 8&\tt $\sim'_v$(X,Y,W) :- $\sim'_c$(X,Y,W).\\
\tiny 9&\tt $\sim'_c$(a,a,W) :- qVal(W), qBound(W,1,1).\\
\tiny 10&\tt $\sim'_c$(f(X),f(Y),W) :- qVal(W), qBound(W,1,1), [..]. \\
\tiny11&\tt $\sim'_c$(g(X),g(Y),W) :- qVal(W), qBound(W,1,1), [..].\\
\tiny12&\tt $\sim'_c$(h(X),h(Y),W) :- qVal(W), qBound(W,1,1), [..].\\
\tiny13&\tt $\sim'_c$(f(X),g(Y),W) :- qVal(W), qBound(W,1,0.8), [..].\\
\tiny14&\tt $\sim'_c$(g(X),f(Y),W) :- qVal(W), qBound(W,1,0.8), [..].\\
\tiny15&\tt $\sim'_c$(g(X),h(Y),W) :- qVal(W), qBound(W,1,0.8), [..].\\
\tiny16&\tt $\sim'_c$(h(X),g(Y),W) :- qVal(W), qBound(W,1,0.8), [..].\\
\end{tabular}
\end{center}
where the ellipsis ``{\small\tt [..]}'' stands for ``{\small\tt qVal(W$_1$), qBound(W,1,W$_1$), $\sim'$(X,Y,W$_1$)}''. Note that the definitions for the program transformations do not require any specific order for the final clauses.
On the other hand, $G''$ becomes the $\clp{\rdom}$-goal:
\begin{center}\it
\begin{tabular}{l}
$\quad$ qVal(W$_1$), qBound(0.5,1,W$_1$), $\sim'$(X,f(Y),W$_1$), \\
$\quad$ qVal(W$_2$), qBound(0.5,1,W$_2$), $\sim'$(X,h(Z),W$_2$) \\
\end{tabular}
\end{center}
When trying to solve $G''$ using the na\"ive implementation of $\sim'$, {\tt Prolog} successively computes infinitely many answers for the subgoal consisting of the first three atoms, none of which can be continued to a successful answer of the whole goal.
Therefore, the overall global computation fails. Since $G$ has valid solutions such as
\[\langle \{X\mapsto g(Y),\ Z\mapsto Y\},\ \{W_1\mapsto 0.8,\ W_2\mapsto 0.8\},\ \emptyset\rangle\]
and also valid ground solutions such as
\[\langle \{X\mapsto g(a),\ Y\mapsto a,\ Z\mapsto a\},\ \{W_1\mapsto 0.8,\ W_2\mapsto 0.8\},\ \emptyset\rangle\]
the incompleteness of {\tt Prolog}'s computation strategy causes weak completeness of SQCLP goal solving to fail in this example.
\mathproofbox
\end{exmp}

The problems just explained have a big impact concerning not only completeness, but also efficiency.
Therefore,  our {\tt Prolog}-based system for SQCLP programming discards the na\"ive implementation of the $EQ_\simrel^{'}$ clauses. Instead, the following {\tt Prolog} code for predicate $\sim'$ is used by our system:
\begin{center}
\footnotesize\it
\renewcommand{\arraystretch}{1.5}
\begin{tabular}{rp{11cm}}
\multicolumn{2}{@{\hspace{0mm}}l}{\normalfont\normalsize
{\bf (B)} Practical implementation of $\sim'$ intended for arbitrary proximity relations.}\\
\tiny $S_1$&\tt $\sim'$(X,Y,W) :- var(X), var(Y), $\sim'_v$(X,Y,W).\\
\tiny $S_2$&\tt $\sim'$(X,Y,W) :- var(X), nonvar(Y), $\sim'_c$(X,Y,W).\\
\tiny $S_3$&\tt $\sim'$(X,Y,W) :- nonvar(X), var(Y), $\sim'_c$(X,Y,W).\\
\tiny $S_4$&\tt $\sim'$(X,Y,W) :- nonvar(X), nonvar(Y), $\sim'_c$(X,Y,W).\\[2mm]
\tiny $V_1$&\tt $\sim'_v$(X,Y,W) :- qVal(W), X = Y.\\[2mm]
\tiny $C_1$&\tt $\sim'_c$(u,u',W) :- qVal(W), qBound(W,$\tp$,$\lambda$).\\
\tiny $C_2$&\tt $\sim'_c$(c(X$_1$,..,X$_n$),c'(Y$_1$,..,Y$_n$),W) :- qVal(W), qBound(W,$\tp$,$\lambda$),\\
&\tt $\qquad$ qVal(W$_1$), qBound(W,$\tp$,W$_1$), $\sim'$(X$_1$,Y$_1$,W$_1$),\\
&\tt $\qquad$ \ldots\\
&\tt $\qquad$ qVal(W$_n$), qBound(W,$\tp$,W$_n$), $\sim'$(X$_n$,Y$_n$,W$_n$).\\
\end{tabular}
\end{center}
where, again, clauses of the form $C_1$ are one for each $u,u'\in B_\cdom$ such that $\simrel(u,u') = \lambda\neq\bt$; and $C_2$ are one for each $c,c' \in DC^n$ such that $\simrel(c,c') = \lambda\neq\bt$
(including the case $c=c'$\!, $\simrel(c,c')=\tp\neq\bt$).

The difference between the implementation {\bf (B)} and the implementation {\bf (A)} is the use of the predicate call {\tt $\sim'_c$(X,Y,W)}
instead of {\tt $\sim'_v$(X,Y,W)} at the bodies of clauses $S_2$ and $S_3$ and the removal of clause $V_2$. These two changes have the effect of avoiding the enumeration of solutions when an equality between two variables is being solved.
For example, for the goal $(\star)$ shown above, the {\tt Prolog} code {\bf (B)} just computes the answer  $\langle \{Y \mapsto X\}, \{W\mapsto\tp\}, \emptyset \rangle$,
while the {\tt Prolog} code {\bf (A)} infinitely enumerates many computed answers, as explained before.
In general, answers computed by the implementation {\bf (B)} of $\sim'$ correspond to a more limited enumeration of solutions,
depending on the data constructor symbols present in the goal.
% will only enumerate for data constructor symbols present in the goal, while keeping variables without instantiation.
The following example illustrates the behaviour of implementaion {\bf (B)} in a more interesting case:
\begin{exmp}[Avoiding infinite computations]
\label{exmp:naive-loops-avoid}
Consider the admissible triple $\langle \simrel,\U,\rdom \rangle$ of Example \ref{exmp:naive-loops}, and let $\Prog$ be the empty program. Recall the goal $G''$ from Example \ref{exmp:naive-loops}:
\begin{center}\it
\begin{tabular}{l}
$\quad$ qVal(W$_1$), qBound(0.5,1,W$_1$), $\sim'$(X,f(Y),W$_1$), \\
$\quad$ qVal(W$_2$), qBound(0.5,1,W$_2$), $\sim'$(X,h(Z),W$_2$)
\end{tabular}
\end{center}
Then, for the subgoal goal consisting of the first three atoms of $G''$
the answers computed by {\tt Prolog} when the predicate $\sim'$ is implemented as in {\bf (B)} are:
\begin{center}
$\langle \{X\mapsto f(Y)\},\ \{W_1 \mapsto 1\},\ \emptyset \rangle$ and
$\langle \{X\mapsto g(Y)\},\ \{W_1 \mapsto 0.8\},\ \emptyset \rangle$.
\end{center}
And for the whole goal $G''$, the only computed answer is:
\[\langle \{X\mapsto g(Y),\ Z\mapsto Y\},\ \{W_1\mapsto 0.8,\ W_2\mapsto  0.8\},\ \emptyset \rangle. \mathproofbox\]
\end{exmp}

Note, however, that the optimization achieved by the move from {\bf (A)} to {\bf (B)}  has a trade-off to pay.
Soundness---in the sense of Definition \ref{dfn:goalsolsys}(1)---is preserved, because the set of computed answers for the implementation {\bf (B)} is a subset of the computed answers for the implementation {\bf (A)}.
However, weak completeness---in the sense of Definition \ref{dfn:goalsolsys}(2)---is not preserved in general, as shown by the following example.
\begin{exmp}
\label{exmp:eqs-implementation-2}
Let  $\langle \simrel, \U, \rdom \rangle$, $\Prog$ and $G$ be as in Example \ref{exmp:eqs-implementation}.
Remember that $G''$ is as shown in Example \ref{exmp:eqs-implementation-1}.
Then, considering the implementation {\bf (B)} of $\sim'$ for generic proximity relations, {\tt Prolog} only computes the answer sol$_1 = \langle \sigma_1, \mu_1, \emptyset \rangle$ for $G''$.
No computed answer subsumes the ground solutions
sol$_2$, sol$_3$ of $G$ shown in Example \ref{exmp:eqs-implementation}.
{\tt Prolog}'s computation strategy is not responsible for the lack of completeness in this case. \mathproofbox
\end{exmp}

Nevertheless, we conjecture that the implementation {\bf (B)} behaves as a flexible restriction of the goal solving system given by the implementation {\bf (A)}  in the sense of Definition \ref{dfn:flexrestr}.
Then, due to Lemma \ref{lema:flexrestr}, we conjecture correctness in the flexible sense for {\bf (B)},
In other words, we claim that our  {\tt Prolog}-based system for SQCLP using implementation {\bf (B)} of $\sim'$
is sound and we conjecture that it is also weakly complete in the flexible sense,
except for the unavoidable failures caused by {\tt Prolog}'s computation strategy.
This conjecture is confirmed as far as the Example \ref{exmp:eqs-implementation} is concerned,
because the computed answer sol$_1$ subsumes the other ground solutions sol$_2$ and sol$_3$ of $G$
in the flexible sense, as shown in the same example.

%To check this we recall Definition \ref{dfn:goalsol}(5) and we
%see that sol$_2$ is subsumed by sol$_1$ in the flexible sense because:
%\begin{itemize}
%\item $\nu = \varepsilon \in \Solc{\emptyset}$ satisfies $\simrel(X\sigma_2,X\sigma_1\varepsilon) = \simrel(b, a) = 0.9 \dgeq 0.9$ and also $\simrel(Y\sigma_2, Y\sigma_1\varepsilon) = \simrel(a, a) = 1 \dgeq 0.9$.
%\item The qualification level of both sol$_2$ and sol$_1$ is $0.9$, thus trivially, $0.9 \dgeq 0.9$.
%\end{itemize}
%Moreover,  sol$_3$ is also subsumed by sol$_1$ in the flexible sense, because:
%\begin{itemize}
%\item $\nu = \varepsilon \in \Solc{\emptyset}$ satisfies $\simrel(X\sigma_3,X\sigma_1\varepsilon) = \simrel(a, a) = 1 \dgeq 0.9$ and also $\simrel(Y\sigma_3, Y\sigma_1\varepsilon) = \simrel(c, a) = 0.9 \dgeq 0.9$.
%\item The qualification level of both sol$_3$ and sol$_1$ is $0.9$, thus trivially, $0.9 \dgeq 0.9$.
%\end{itemize}
%In fact, it is easy to check that any of the three ground solutions sol$_1$, sol$_2$ and sol$_3$ subsumes the other two in the flexible sense. Therefore, any goal solving system that computes at least one of these three solutions would satisfy weak completeness in the flexible sense concerning goal $G$.

A further optimization of implementation {\bf (B)} is possible if the given proximity relation $\simrel$ is transitive --- i.e. a similarity.
In this case our prototype system implements $\sim'$ by means of the following {\tt Prolog} code:
\begin{center}
\footnotesize\it
\renewcommand{\arraystretch}{1.5}
\begin{tabular}{rp{11cm}}
\multicolumn{2}{@{\hspace{0mm}}l}{\normalfont\normalsize
{\bf (C)} Practical implementation of $\sim'$ intended for similarity relations.}\\
\tiny $S_1$&\tt $\sim'$(X,Y,W) :- var(X), var(Y), $\sim'_v$(X,Y,W).\\
\tiny $S_2$&\tt $\sim'$(X,Y,W) :- var(X), nonvar(Y), $\sim'_v$(X,Y,W).\\
\tiny $S_3$&\tt $\sim'$(X,Y,W) :- nonvar(X), var(Y), $\sim'_v$(X,Y,W).\\
\tiny $S_4$&\tt $\sim'$(X,Y,W) :- nonvar(X), nonvar(Y), $\sim'_c$(X,Y,W).\\[2mm]
\tiny $V_1$&\tt $\sim'_v$(X,Y,W) :- qVal(W), X = Y.\\[2mm]
\tiny $C_1$&\tt $\sim'_c$(u,u',W) :- qVal(W), qBound(W,$\tp$,$\lambda$).\\
\tiny $C_2$&\tt $\sim'_c$(c(X$_1$,..,X$_n$),c'(Y$_1$,..,Y$_n$),W) :- qVal(W), qBound(W,$\tp$,$\lambda$),\\
&\tt $\qquad$ qVal(W$_1$), qBound(W,$\tp$,W$_1$), $\sim'$(X$_1$,Y$_1$,W$_1$),\\
&\tt $\qquad$ \ldots\\
&\tt $\qquad$ qVal(W$_n$), qBound(W,$\tp$,W$_n$), $\sim'$(X$_n$,Y$_n$,W$_n$).\\
\end{tabular}
\end{center}
where the only difference w.r.t. implementation {\bf (B)} is that {\bf (C)} uses the predicate call {\tt $\sim'_v$(X,Y,W)}
instead of {\tt $\sim'_c$(X,Y,W)} at the bodies of clauses $S_2$ and $S_3$.

A useful way to understand the difference between {\bf (B)} and {\bf (C)}  is to think of both as different implementations of a unification algorithm modulo a given proximity relation $\simrel$.
In both cases, a predicate call {\tt $\sim'$(X,Y,W)} is intended to compute a unifier modulo $\simrel$ with qualification degree {\tt W} for {\tt X} and {\tt Y}---more precisely, for the terms  bound to {\tt X}  and {\tt Y} at run-time---and clauses $S_i$ $(i = 1,2,3,4)$ distinguish four possible cases in the same manner.
The two implementations differ only in the actions taken in each of these four cases.
The actions executed by implementation {\bf (B)} can be intuitively described as follows:
\begin{enumerate}
\item
{\em Case 1:} both {\tt X} and {\tt Y} are variables. \\
{\em Action:} just unify them (achieved by clause $V_1$).
\item
{\em Case 2:} 
%{\tt X} is a variable and {\tt Y} is bound to a non-variable term. \\
%{\em Actions:} compute a first solution by binding  {\tt X} to the term bound to {\tt %Y} (achieved by clause $V_1$). Compute alternative solutions by binding {\tt X} to %non-variable terms whose root symbol is
%$\simrel$-close to the root symbol of the term bound to {\tt Y} (achieved by clauses %$C_1$ and $C_2$).
%\noindent {\bf Case 2 bis}:  
{\tt X} is a variable and {\tt Y} is bound to a non-variable term. \\
{\em Actions:} Compute alternative solutions by binding {\tt X} to non-variable terms whose root symbol is
$\simrel$-close to the root symbol of the term bound to {\tt Y} (achieved by clauses $C_1$ and $C_2$). In particular one of these solutions will correspond to binding  {\tt X} to the term bound to {\tt Y}.

\item
{\em Case 3:} 
%{\tt Y} is a variable and {\tt X} is bound to a non-variable term. \\
%{\em Actions:} compute a first solution by binding  {\tt Y} to the term bound to {\tt %X} (achieved by clause $V_1$). Compute alternative solutions by binding {\tt Y} to %non-variable terms whose root symbol is
%$\simrel$-close to the root symbol of the term bound to {\tt X} (achieved by clauses %$C_1$ and $C_2$).
%
%
%\noindent {\bf Case 3 bis}:  
{\tt Y} is a variable and {\tt X} is bound to a non-variable term. \\
{\em Actions:} Compute alternative solutions by binding {\tt Y} to non-variable terms whose root symbol is
$\simrel$-close to the root symbol of the term bound to {\tt X} (achieved by clauses $C_1$ and $C_2$). In particular one of these solutions will correspond to binding  {\tt Y} to the term bound to {\tt X}.

\item
{\em Case 4:} both {\tt X} and {\tt Y} are bound to non-variable terms, both with root and $n$ children terms. \\
{\em Action:} first check that the root symbols of the terms bound to {\tt X} and {\tt Y} are $\simrel$-close;
then decompose these two terms and recursively proceed to unify the $i$-th child of the term bound to {\tt  X}
and the $i$-th child of the term bound to {\tt Y}, for $i = 1 \ldots n$ (achieved by clauses $C_1$ and $C_2$).
\end{enumerate}

On the other hand, an intuitive description of  implementation  {\bf (C)} is as follows:
\begin{enumerate}
\item
{\em Case 1:} both {\tt X} and {\tt Y} are variables. \\
{\em Action:} as in case (1) of implementation {\bf (B)}.
\item
{\em Case 2:} {\tt X} is a variable and {\tt Y} is bound to a non-variable term. \\
{\em Action:} just bind {\tt X} to the term bound to {\tt Y} (achieved by clause $V_1$).
\item
{\em Case 3:} {\tt Y} is a variable and {\tt X} is bound to a non-variable term. \\
{\em Action:} just bind {\tt Y} to the term bound to {\tt X} (achieved by clause $V_1$).
\item
{\em Case 4:} both {\tt X} and {\tt Y} are bound to non-variable terms, both with root and $n$ children terms. \\
{\em Action:} as in case (4) of implementation {\bf (B)}.
\end{enumerate}

Clearly, the difference between these two implementations  is limited to cases (2) and (3), where {\bf (B)}  enumerates a set of various alternative unifiers while {\bf (C)} behaves in a deterministic way, computing just one of these unifiers.
In fact, {\bf (C)} behaves as a {\tt Prolog} implementation of known unification algorithms modulo a given similarity relation $\simrel$,
as those presented in \cite{AF02,Ses02} (only for the qualification domain $\U$) and other related papers,
which are complete in the flexible sense for solving unification problems.
This is due to the fact that the substitution $\{X \mapsto t\}$ can be taken as the unique unifier computed for
a variable $X$ and a term $t$, that subsumes in the weak sense other possible unifiers thanks to the transitivity property of $\simrel$.

Concerning the behaviour of our {\tt Prolog}-based SQCLP system when  {\bf (C)} is used as the implementation of the $\sim'$ predicate,
we claim soundness for any choice of $\simrel$ (transitive or not), because all the computed answers can be also computed by {\bf (B)},
which is sound. In case that $\simrel$ is transitive, weak completeness in the flexible sense is the best behaviour that can be expected,
but more research is still needed to clarify this issue.
The example below shows that weak completeness in the flexible sense generally fails for unification problems
(and with more reason for general SQCLP goals), when $\simrel$ is not transitive.
In fact, the same example shows that transitivity of $\simrel$ is a necessary requirement for the completeness
(in the flexible sense) of unification algorithms modulo $\simrel$ of the kind presented in \cite{Ses02} and related papers.

%This implementation avoids enumeration of possible answers altogether, thus this is the most efficient implementation of all, but it is only correct in the flexible sense if the proximity relation is in fact a similarity relation.
%This is a consequence of the nature of similarity relations with respect to proximity relations. For a similarity relation it is not important what particular element of a given equivalence class is considered, because its similarity degree with respect to another specific element will never be lower to that of others in the same equivalence class with respect to the same specific element. This holds true due to the transitivity property of similarity relations, and this behaviour let us solve any equation ``$X \sim t$'' with the substitution $\{X\mapsto t\}$ as its only computed answer. Note that this is the behaviour for the implementation {\bf (C)} of clauses $S_2$ and $S_3$.

\begin{exmp}
\label{exmp:eqs-implementation-3}
Consider for the last time the admissible triple $\langle \simrel, \U, \rdom \rangle$ of Example \ref{exmp:eqs-implementation},
the empty program, the goal $G$ shown in Example \ref{exmp:eqs-implementation-1},
and the CLP goal $G''$ obtained as translation of $G$ and shown in Example \ref{exmp:eqs-implementation-1}, which is:
\[
\begin{array}{l}
{\it qVal}(W_1),\ {\it qBound}(0.8,1,W_1),\ {\sim'}(Y,X,W_1), \\
{\it qVal}(W_2),\ {\it qBound}(0.8,1,W_2),\ {\sim'}(X,b,W_2), \\
{\it qVal}(W_3),\ {\it qBound}(0.8,1,W_3),\ {\sim'}(Y,c,W_3)
\end{array}
\]
Note that $G$ is a  unification problem modulo $\simrel$ with the three ground solutions shown in Example \ref{exmp:eqs-implementation}.
The proximity relation $\simrel$ is not transitive, because $\simrel(b,c) = 0.4 \not \geq 0.9 = \simrel(b,a) \sqcap \simrel(a,c)$.
The resolution of $G''$ by using the {\tt Prolog} code {\bf (C)} for $\sim'$ eventually reduces to solving a new goal of the form
$${\it qVal}(W_3),\ {\it qBound}(0.8,1,W_3),\ {\sim'}(b,c,W_3)$$
which fails, since $\simrel(b,c) = 0.4 \not \geq 0.8$.
In this example, {\tt Prolog}'s computation strategy is not responsible for the lack of completeness. \mathproofbox
%
%one has that no computed answers are found using the implementation {\bf (C)} of the predicate $\sim'$. This is correct because $\simrel$ is not a valid similarity relation given that, in particular, $\simrel(b,c) \dgeq \simrel(b,a) \sqcap \simrel(a,c)$ does not hold. This leads to the failure of the goal when trying to prove ${\sim'}(b,c,W_3)$ which results from the substitution $\{Y\mapsto X,\ X\mapsto b\}$ computed after the two first atoms (not considering {\it qVal} and {\it qBound} constraints) in the goal. \mathproofbox
\end{exmp}

We have just discussed three possible {\tt Prolog} implementations of the CLP clauses in the set $EQ_\simrel^{'}$, called {\bf (A)}, {\bf (B)} and {\bf (C)}. The {\tt Prolog}-based prototype system for SQCLP programming  presented in the next subsection  only supports implementations {\bf (B)} and {\bf (C)}, using two different predicates ${\it prox}/4$ and ${\it sim}/4$, respectively, to implement the behaviour of $\sim'$ appropriate in each case. By default, the system assumes implementation {\bf (B)}, and a program directive {\tt \#optimized\_unif} must be used in the case that  implementation {\bf (C)} is desired.

\subsection{{\tt (S)QCLP}: A Prototype System for SQCLP Programming}
\label{sec:practical:prototype}

The prototype implementation object of this subsection is publicly available, and can be found at:
\begin{center}
{\tt http://gpd.sip.ucm.es/cromdia/qclp}
\end{center}

The system currently requires the user to have installed either {\em SICStus Prolog} or {\em SWI-Prolog}, and it has been tested to work under Windows, Linux and MacOSX platforms.
The latest version available at the time of writing this paper is {\tt 0.6}.
If a latter version is available some things might have changed but in any case the main aspects of the system should remain the same. Please consult the {\em changelog} provided within the system itself for specific changes between versions.

SQCLP is a very general programming scheme and, as such, it supports different proximity relations, different qualification domains and different constraint domains when building specific instances of the scheme for any specific purpose.
As it would result impossible to provide an implementation for every admissible triple (or instance of the scheme), it becomes mandatory to decide in advance what specific instances will be available for writing programs in {\tt (S)QCLP}.
In essence:
\begin{enumerate}
\item
In its current state, the only available constraint domain is $\rdom$.
Thus, under both {\em SICStus Prolog} and {\em SWI-Prolog} the library {\tt clpr} will provide all the available primitives in {\tt (S)QCLP} programs.
\item
The available qualification domains are: `{\tt b}' for the domain $\B$; `{\tt u}' for the domain $\U$; `{\tt w}' for the domain $\W$; and any strict cartesian product of those, as e.g. `{\tt (u,w)}' for the product domain $\U{\otimes}\W$.
\item With respect to proximity relations, the user will have to provide, in addition to the two symbols and their proximity value, their {\em kind} (either predicate or constructor) and their {\em arity}. Both kind and arity must be the same for each pair of symbols having a proximity value different of $\bt$.
\end{enumerate}
Note, however, that when no specific proximity relation $\simrel$ is provided for a given program, $\sid$ is then assumed.
Under this circumstances, an obvious technical optimization consists in transforming the original program only with elim$_\qdom$, thus reducing the overload introduced in this case by elim$_\simrel$.
The reason behind this optimization is that for any given $\sqclp{\sid}{\qdom}{\cdom}$-program $\Prog$, it is also true that $\Prog$ is a $\qclp{\qdom}{\cdom}$-program, therefore $\elimD{\elimS{\Prog}}$ must semantically be equivalent to $\elimD{\Prog}$.
Nevertheless, $\elimD{\Prog}$ behaves more efficiently than $\elimD{\elimS{\Prog}}$ due to the reduced number of resulting clauses.
Thus, in order to improve the efficiency, the system will avoid the use of elim$_\simrel$ when no proximity relation is provided by the user.

The final available instances in the {\tt (S)QCLP} system are: $\sqclp{\simrel}{\tt b}{\tt clpr}$, $\sqclp{\simrel}{\tt u}{\tt clpr}$, $\sqclp{\simrel}{\tt w}{\tt clpr}$, $\sqclp{\simrel}{\tt (u,w)}{\tt clpr}$, \ldots{} and their counterparts in the QCLP scheme when $\simrel$ = $\sid$.

% Programming in SQCLP
\subsubsection{Programming in {\tt (S)QCLP}}
\label{sec:practical:prototype:programming}

Programming in {\tt (S)QCLP} is straightforward if the user is accustomed to the Prolog programming style.
However, there are three syntactic differences with pure Prolog:
\begin{enumerate}
  \item Clauses implications are replaced by ``{\tt <-$d$-}'' where $d \in \aqdom$. If $d = \tp$, then the implication can become just ``{\tt <--}''. E.g. ``{\tt <-0.9-}'' is a valid implication in the domains $\U$ and $\W$; and ``{\tt <-(0.9,2)-}'' is a valid implication in the domain ${\U{\otimes}\W}$.
  \item Clauses in {\tt (S)QCLP} are not finished with a dot ({\tt .}). They are separated by layout, therefore all clauses in a {\tt (S)QCLP} program must start in the same column. Otherwise, the user will have to explicitly separate them by means of semicolons ({\tt ;}).
  \item After every body atom (even constraints) the user can provide a threshold condition using `{\tt \#}'. The notation `{\tt ?}' can also be used instead of some particular qualification value, but in this case the threshold condition `{\tt \#?}' can be omitted.
\end{enumerate}
Comments are as in Prolog:
\begin{verbatim}
% This is a line comment.
/* This is a multi-line comment, /* and they nest! */. */
\end{verbatim}
and the basic structure of a {\tt (S)QCLP} program is the following (line numbers are for reference):
\begin{center}
\small
\begin{tabular}{rp{11cm}}
\multicolumn{2}{@{\hspace{0mm}}l}{{\bf File:} {\em Peano.qclp}} \\[2mm]
\tiny 1 & \verb+% Directives...+ \\
\tiny 2 & \verb+# qdom w+ \\[2mm]
\tiny 3 & \verb+% Program clauses...+ \\
\tiny 4 & \verb+% num( ?Num )+ \\
\tiny 5 & \verb+num(z) <--+ \\
\tiny 6 & \verb+num(s(X)) <-1- num(X)+ \\
\end{tabular}
\end{center}
In the previous small program, lines {\tt 1}, {\tt 3} and {\tt 4} are line comments, line {\tt 2} is a program directive telling the compiler the specific qualification domain the program is written for, and lines {\tt 5} and {\tt 6} are program clauses defining the well-known Peano numbers.
As usual, comments can be written anywhere in the program as they will be completely ignored (remember that a line comment must necessarily end in a new line character, therefore the very last line of a file cannot contain a line comment),
and directives must be declared before any program clause.
There are three program directives in {\tt (S)QCLP}:
\begin{enumerate}
% qdom
\item
The first one is ``{\tt \#qdom} {\em qdom}'' where {\em qdom} is any system available qualification domain, i.e. {\tt b}, {\tt u}, {\tt w}, {\tt (u,w)}\ldots{}
See line {\tt 2} in the previous program sample as an example.
This directive is mandatory because the user must tell the compiler for which particular qualification domain the program is written.
% prox
\item
The second one is ``{\tt \#prox} {\em file}'' where {\em file} is the name of a file (with extension {\tt .prox}) containing a proximity relation.
If the name of the file starts with a capital letter, or it contains spaces or any special character, {\em file} will have to be quoted with single quotes.
For example, assume that with our program file we have another file called {\em Proximity.prox}.
Then, we would have to write ``{\tt \#prox `Proximity'}'' to link the program with such proximity relation.
This directive is optional, and if omitted, the system assumes that the program is of an instance of the QCLP scheme.
% optimize_unif
\item
The third one is ``{\tt \#optimized\_unif}''. This directive tells the compiler that the program is intended to be used with the implementation {\bf (C)} for the predicate $\sim'$, as explained in Subsection \ref{sec:practical:SQCLP:eqsimrel}. 
%However,  this could have the effect of losing valid answers, although we conjecture that if %the proximity relation is transitive and if the program clauses do not make use of attenuation %factors other that $\tp$, this will not happen.
\end{enumerate}

Proximity relations are defined in files of extension {\tt .prox} with the following form:
\begin{center}
\small
\begin{tabular}{rp{11cm}}
\multicolumn{2}{@{\hspace{0mm}}l}{{\bf File:} {\em Work.prox}} \\[2mm]
\tiny 1 & \verb+% Predicates: pprox( S1, S2, Arity, Value ).+ \\
\tiny 2 & \verb+pprox(wrote, authored, 2, (0.9,0)).+ \\[2mm]
\tiny 3 & \verb+% Constructors: cprox( S1, S2, Arity, Value ).+ \\
\tiny 4 & \verb+cprox(king_lear, king_liar, 0, (0.8,2)).+ \\
\end{tabular}
\end{center}
where the file can contain {\tt pprox/4} Prolog facts, for defining proximity between predicate symbols of any arity;
or {\tt cprox/4} Prolog facts, for defining proximity between constructor symbols of any arity.
The arguments of both {\tt pprox/4} and {\tt cprox/4} are: the two symbols, their arity and its proximity value.
Note that, although it is not made explicit the qualification domain this proximity relation is written for, all values in it must be of the same specific qualification domain, and this qualification domain must be the same declared in every program using the proximity relation.
Otherwise, the solving of equations may produce unexpected results or even fail.

Reflexive and symmetric closure is inferred by the system, therefore, there is no need for writing reflexive proximity facts, nor the symmetric variants of proximity facts already provided.
You can notice this in the previous sample file in which neither reflexive proximity facts, nor the symmetric proximity facts to those at lines {\tt 2} and {\tt 4} are provided.
In the case of being explicitly provided, additional (repeated) solutions might be computed for the same given goal, although soundness and weak completeness of the system should still be preserved.
Transitivity is neither checked nor inferred so the user will be responsible for ensuring it if desired.

As the reader would have already guessed, the file  {\tt Work.prox} implements the proximity relation $\simrel_r$ of Example \ref{exmp:pr} in {\tt (S)QCLP}. Finally, the program $\Prog_r$ of Example \ref{exmp:pr} can be represented in {\tt (S)QCLP} as follows:
\begin{center}
\small
\begin{tabular}{rp{11cm}}
\multicolumn{2}{@{\hspace{0mm}}l}{{\bf File:} {\em Work.qclp}} \\[2mm]
\tiny 1 & \verb+# qdom (u,w)+ \\
\tiny 2 & \verb+# prox 'Work'+ \\[2mm]
\tiny 3 & \verb+% famous( ?Author )+ \\
\tiny 4 & \verb+famous(shakespeare) <-(0.9,1)-+ \\[2mm]
\tiny 5 & \verb+% wrote( ?Author, ?Book )+ \\
\tiny 6 & \verb+wrote(shakespeare, king_lear) <-(1,1)-+ \\
\tiny 7 & \verb+wrote(shakespeare, hamlet) <-(1,1)-+ \\[2mm]
\tiny 8 & \verb+% good_work( ?Work )+ \\
\tiny 9 & \verb+good_work(X) <-(0.75,3)- famous(Y)#(0.5,100), authored(Y,X)+ \\
\end{tabular}
\end{center}
Note that, at line {\tt 1} the qualification domain $\U{\otimes}\W$ is declared, and at line {\tt 2} the proximity relation at {\tt Work.prox} is linked to the program.
In addition, observe that one threshold constraint is imposed for a body atom in the program clause at line {\tt 9}, effectively requiring to prove {\tt famous(Y)} for a qualification value of {\em at least} {\tt (0.5,100)} to be able to use this program clause.

Finally, we explain how constraints are written in {\tt (S)QCLP}.
As it has already been said, only $\rdom$ is available, thus both in {\em SICStus Prolog} and {\em SWI-Prolog} the library {\tt clpr} is the responsible for providing the available primitive predicates.
Given that constraints are primitive atoms of the form {\tt r($\ntup{{\tt t}}{n}$)} where {\tt r} $\in PP^n$ and {\tt t$_i$} are terms; primitive atoms share syntax with usual Prolog atoms.
At this point, and having that many of the primitive predicates are syntactically operators (hence not valid identifiers), the syntax for predicate symbols has been extended to include operators, therefore predicate symbols like $op_+ \in PP^3$, which codifies the operation {\tt +} in a 3-ary predicate, will let us to build constraints of the form {\tt +(A,B,C)}, that must be understood as in $A+B=C$ or $C=A+B$.
Similarly, predicate symbols like $cp_> \in PP^2$, which codifies the comparison operator {\tt >} in a binary predicate, will let us to build constraints of the form {\tt >(A,B)}, that must be understood as in $A > B$.
Any other primitive predicate such as {\em maximize} $\in PP^1$, will let us to build constraints like {\tt maximize(X)}.
Valid primitive predicate symbols include {\tt +}, {\tt -}, {\tt *}, {\tt /}, {\tt >}, {\tt >=}, {\tt =<}, {\tt <}, {\tt maximize}, {\tt minimize}, etc.

Threshold constraints can also be provided for primitive atoms in the body of clauses with the usual notation.
Note, however, that due the semantics of SQCLP, all primitive atoms can be trivially proved with $\tp$ if they ever succeeds---so threshold constraints become, in this case, of no use.

The syntax for constraints explained above follows the standard syntax for atoms.
Nonetheless, the system also allows to write these constraints in a more natural infix notation.
More precisely, {\tt +(A,B,C)} can be also written in the infix form {\tt A+B=C} or {\tt C=A+B}, and {\tt >(X,Y)} in the infix form {\tt X>Y}; and similarly for other $op$ and $cp$ constraints.
When using infix notation, threshold conditions can be set by (optionally) enclosing the primitive atom between parentheses, therefore becoming {\tt (A+B=C)\#$\tp$}, {\tt (C=A+B)\#$\tp$} or {\tt (X>Y)\#$\tp$} (or any other valid qualification value or `?').
Using parentheses is recommended to avoid understanding that the threshold condition is set only for the last term in the constraint, which would make no sense.
Note that even in infix notation, operators cannot be nested, that is, terms {\tt A}, {\tt B}, {\tt C}, {\tt X} and {\tt Y} cannot have operators as main symbols (neither in prefix nor in infix notation), so the infix notation is just a syntactic sugar of its corresponding prefix notation.

As a final example for constraints, one could write the predicate {\tt double/2} in {\tt (S)QCLP}, for computing the double of any given number, with just the clause \verb+double(N,D) <-- *(N,2,D)+, or \verb+double(N,D) <-- N*2=D+ for a clause with a more natural syntax.

%The interpreter for (S)QCLP
\subsubsection{The interpreter for {\tt (S)QCLP}}

The interpreter for {\tt (S)QCLP} has been implemented on top of both {\em SICStus Prolog} and {\em SWI-Prolog}.
To load it, one must first load her desired (and supported) Prolog system and then load the main file of the interpreter---i.e. {\tt qclp.pl}---, that will be located in the main {\tt (S)QCLP} folder among other folders.
Once loaded, one will see the welcome message and will be ready to compile and load programs, and to execute goals.
\begin{verbatim}
WELCOME TO (S)QCLP 0.6
(S)QCLP is free software and comes with absolutely no warranty.
Support & Updates: http://gpd.sip.ucm.es/cromdia/qclp.

Type ':help.' for help.
yes
| ?-
\end{verbatim}

From the interpreter for {\tt (S)QCLP} one can, in addition to making use of any standard Prolog goals, use the specific {\tt (S)QCLP} commands required for both interacting with the {\tt (S)QCLP} system, and  for compiling/loading SQCLP programs.
All these commands take the form:
\begin{center}
\tt :command.
\end{center}
if they do not require arguments, or:
\begin{center}
\tt :command({\em Arg}$_1$, \ldots, {\em Arg}$_n$).
\end{center}
if they do; where each argument {\tt\em Arg}$_i$ must be a Prolog atom unless stated otherwise.
The most useful commands are:
\begin{itemize}
\item {\tt :cd(}{\em Folder}{\tt ).} \\
Changes the working directory to {\em Folder}. {\em Folder} can be an absolute or relative path.
\item {\tt :compile(}{\em Program}{\tt ).} \\
Compiles the {\tt (S)QCLP} program `{\em Program}{\tt .qclp}' producing the equivalent  Prolog program in the file `{\em Program}{\tt .pl}'.
\item {\tt :load(}{\em Program}{\tt ).} \\
Loads the already compiled {\tt (S)QCLP} program `{\em Program}{\tt .qclp}' (note that the file `{\em Program}{\tt .pl}' must exist for the program to correctly load).
\item {\tt :run(}{\em Program}{\tt ).} \\
Compiles the {\tt (S)QCLP} program `{\em Program}{\tt .qclp}' and loads it afterwards.
This command is equivalent to executing: {\tt :compile(}{\em Program}{\tt ), :load(}{\em Program}{\tt ).}
\end{itemize}

For illustration purposes, we will assume that you have the files {\tt Work.prox} and {\tt Work.qclp} (both as seen before) in the folder {\tt $\sim$/examples}.
Under these circumstances, after loading your preferred Prolog system and the interpreter for {\tt (S)QCLP}, one would only have to change the working directory to that where the files are located:
\begin{Verbatim}[commandchars=\\\{\}]
| ?- :cd('\prox/examples').
\end{Verbatim}
and run the program:
\begin{verbatim}
| ?- :run('Work').
\end{verbatim}
If no errors are encountered, one should see the output:
%\begin{Verbatim}[commandchars=\\\{\}]
\begin{verbatim}
| ?- :run('Work').
<Work> Compiling...
<Work> QDom: 'u,w'.
<Work> Prox: 'Work'.
<Work> Translating to QCLP...
<Work> Translating to CLP...
<Work> Generating code...
<Work> Done.
<Work> Loaded.
yes
\end{verbatim}
and now everything is ready to execute goals for the program loaded.

% Compiling SQCLP-Goals
\subsubsection{Executing SQCLP-Goals}
\label{sec:practical:prototype:exec}

Recall that goals have the form
$\qat{A_1}{W_1},\ \ldots,\ \qat{A_m}{W_m} \sep W_1 \!\dgeq^? \beta_1,\ \ldots,\ W_m \dgeq^? \!\beta_m$
which in actual {\tt (S)QCLP} syntax becomes:
\begin{verbatim}
| ?- A1#W1, ..., Am#Wm :: W1 >= B1, ..., Wm >= Bm.
\end{verbatim}
Note the following:
\begin{enumerate}
\item Goals must end in a dot ({\tt .}).
\item The symbol `$\sep$' is replaced by `\verb+::+'.
\item The symbol `${\dgeq}^?$' is replaced by `\verb+>=+' (and this is independent of the qualification domain in use, so that it may mean $\le$ in $\W$).
\item Conditions of the form $W \dgeq^?\ ?$ {\em must be omitted}, therefore $\qat{A_1}{W_1}, \qat{A_2}{W_2} \sep W_1 \dgeq^?\ ?, W_2 \dgeq^? \!\beta_2$ becomes ``\verb+A1#W1, A2#W2 :: W2 >= B2.+'',
and $\qat{A}{W}\sep W \dgeq^?\ ?$ becomes just ``\verb+A#W.+''.
\end{enumerate}

Assuming now that we have loaded the program {\tt Work.qclp} as explained before, we can execute the goal $\qat{good\_work(king\_liar)}{W}\sep W \dgeq^? (0.5,100)$:
\begin{verbatim}
| ?- good_work(king_liar)#W::W>=(0.5,10).
W = (0.6,5.0) ? ;
W = (0.675,4.0) ? ;
no
\end{verbatim}
Note that the system computes two answers, with different qualification values. In this simple example, the second computed answer provides a better qualification value. In general, different computed answers for the same goal come with different qualification values and it is not always the case that one of the answers provides the optimal qualification value.

% Examples
\subsubsection{Examples}
\label{sec:practical:prototype:examples}

To finish this subsection, we are now showing some additional goal executions using the interpreter for {\tt (S)QCLP} and the programs displayed along the paper.

\paragraph{Peano.}
Consider the program {\tt Peano.qclp} as displayed at the beginning of Subsection \ref{sec:practical:prototype:programming}.
Qualifications in this program are intended as a cost measure for obtaining a given number in the Peano representation, assuming that each use of the clause at line {\tt 6} requires to pay {\em at least} {\tt 1}.
In essence, threshold conditions will impose an upper bound over the maximum number obtainable in goals containing the atom {\tt num(X)}.
Therefore if we ask for numbers {\em up to} a cost of {\tt 3} we get the following answers:
\begin{center}
\small
\begin{tabular}{rp{11cm}}
\tiny Goal & \verb+?- num(X)#W::W>=3.+ \\[2mm]
\tiny Sol$_1$ & \verb+W = 0.0, X = z ? ;+ \\
\tiny Sol$_2$ & \verb+W = 1.0, X = s(z) ? ;+ \\
\tiny Sol$_3$ & \verb+W = 2.0, X = s(s(z)) ? ;+ \\
\tiny Sol$_4$ & \verb+W = 3.0, X = s(s(s(z))) ? ;+ \\
\tiny & \verb+no+ \\
\end{tabular}
\end{center}

\paragraph{Work.}
Consider now the program {\tt Work.qclp} and the proximity relation {\tt Work.prox}, both as displayed in Subsection \ref{sec:practical:prototype:programming} above.
In this program, qualifications behave as the conjunction of the certainty degree of the user confidence about some particular atom, and a measure of the minimum cost to pay for proving such atom.
In these circumstances, we could ask---just for illustration purposes---for famous authors with a minimum certainty degree---for them being actually famous---of {\tt 0.5}, and with a proof cost of no more than {\tt 30} (think of an upper bound for possible searches in different databases).
Such a goal would have, in this very limited example, only the following solution:
\begin{center}
\small
\begin{tabular}{rp{11cm}}
\tiny Goal & \verb+?- famous(X)#W::W>=(0.5,30).+ \\[2mm]
\tiny Sol$_1$ & \verb+W = (0.9,1.0), X = shakespeare ? ;+ \\
\tiny & \verb+no+ \\
\end{tabular}
\end{center}
meaning that we can have a confidence of {\tt shakespeare} being famous of {\tt 0.9}, and that we can prove it with a cost of {\tt 1}.

Now, in a similar fashion we could try to obtain different works that can be considered as good works by using the last clause in the example.
Limiting the search to those works that can be considered good with a qualification value better or equal to {\tt (0.5,100)} produce the following result:
\begin{center}
\small
\begin{tabular}{rp{11cm}}
\tiny Goal & \verb+?- good_work(X)#W::W>=(0.5,100).+ \\[2mm]
\tiny Sol$_1$ & \verb+W = (0.675,4.0), X = king_lear ? ;+ \\
\tiny Sol$_2$ & \verb+W = (0.6,5.0), X = king_liar ? ;+ \\
\tiny & \verb+no+ \\
\end{tabular}
\end{center}
A valid ground answer for this goal is $gsol = \langle \eta, \rho, \emptyset \rangle$ where $\eta = \{X \mapsto {\it king\_liar}\}$ and $\rho = \{W \mapsto (0.675, 4)\}$ (which corresponds to the second computed answer for the ground goal displayed in Subsection \ref{sec:practical:prototype:exec}).
Note that the first computed answer shown above is $ans = \langle \sigma, \mu, \emptyset\rangle$ where $\sigma = \{X \mapsto {\it king\_lear}\}$ and $\mu = \{W \mapsto (0.675, 4)\}$ which subsumes $gsol$ in the flexible sense via $\nu = \varepsilon \in \Sol{\rdom}{\emptyset}$.

\paragraph{Library.}
Finally, consider the program $\Prog_s$ and the proximity relation $\simrel_s$, both as displayed in Figure \ref{fig:library} of Section \ref{sec:sqclp}.
As it has been said when this example was introduced, the predicate {\em guessRdrLvl} takes advantage of attenuation factors to encode heuristic rules to compute reader levels on the basis of vocabulary level and other book features.
As an illustration of use, consider the following goal:
\begin{center}
\small
\begin{tabular}{rp{11cm}}
\tiny Goal & \verb+?- guessRdrLvl(book(2, 'Dune', 'F. P. Herbert', english, sciFi,+\\
& \verb+   medium, 345), Level)#W.+ \\[2mm]
\tiny Sol$_1$ & \verb+W = 0.8, Level = intermediate ? ;+ \\
& $\cdots$ \\
\tiny Sol$_6$ & \verb+W = 0.7, Level = upper ?+ \\
\tiny & \verb+yes+ \\
\end{tabular}
\end{center}
Here we ask for possible ways of classifying the second book in the library according to reader levels.
We obtain as valid solutions, among others, {\tt intermediate} with a certainty factor of {\tt 0.8}; and {\tt upper} with a certainty factor of {\tt 0.7}.
These valid solutions show that the predicate {\em guessRdrLvl} tries with different levels for any certain book based on the heuristic implemented by the qualified clauses.

To conclude, consider now the goal proposed in Section \ref{sec:sqclp} for this program.
For such goal we obtain:
\begin{center}
\small
\begin{tabular}{rp{11cm}}
\tiny Goal & \verb+?- search(german, essay, intermediate, ID)#W::W>=0.65.+\\[2mm]
\tiny Sol$_1$ & \verb+W = 0.8, ID = 4 ?+ \\
\tiny & \verb+yes+ \\
\end{tabular}
\end{center}
What tells us that the forth book in the library is written in German, it can be considered to be an essay, and it is targeted for an intermediate reader level.
All this with a certainty degree of {\em at least} {\tt 0.8}.

\subsection{Efficiency}
\label{sec:practical:efficiency}
The minimum---and unavoidable---overload introduced by qualifications and proximity relations in the transformed programs manifests itself in the case of {\tt (S)QCLP} programs which use the identity proximity relation and have $\tp$ as the attenuation factor of all their clauses.
In order to measure this overload we have made some experiments using some program samples, taken from the {\em SICStus Prolog Benchmark} that can be found in:
\begin{center}
\tt http://www.sics.se/isl/sicstuswww/site/performance.html
\end{center}
and we have compared the time it took to repeatedly execute a significant number of times each program in both {\tt (S)QCLP} and {\em SICStus Prolog} making use of a {\em slightly} modified (to ensure a correct behaviour in both systems) version of the harness also provided in the same site.

From all the programs available in the aforementioned site, we selected the following four:
\begin{itemize}
\item {\em naivrev:} na\"{i}ve implementation of the predicate that reverses the contents of a list.
\item {\em deriv:} program for symbolic derivation.
\item {\em qsort:} implementation of the well-known sorting algorithm {\em Quicksort}.
\item {\em query:} obtaining the population density of different countries.
\end{itemize}
No other program could be used because they included impure features such as cuts which are not currently supported by our system. In order to adapt these Prolog programs to our setting
the following modifications were required:
\begin{enumerate}
\item All the program clause are assumed to have $\tp$ as attenuation factor.
      After including these attenuation factors, we obtain as results QCLP programs.
      More specifically we obtain two QCLP programs for each initial Prolog program,
      one using  the qualification domain $\B$ (because this domain uses trivial constraints), and another using  the qualification domain $\U$ (which uses $\rdom$-constraints).
\item We define an empty proximity relation, allowing us to obtain two additional
SQCLP-programs.
\item By means of the program directive ``{\tt \#optimized\_unif}'' defined in Subsection \ref{sec:practical:prototype:programming}, each SQCLP program can be also executed
    in this optimized mode. Therefore each original Prolog Program produces six {\tt (S)QCLP}
    programs, denoted as Q(b), Q(u), PQ(b), PQ(u), SQ(b) and SQ(u) in Table \ref{table:results}.
\end{enumerate}

Additionally some  minor modifications to the program samples have been introduced for compatibility reasons, i.e. additions using the predicate {\tt is/2} were replaced, both in the Prolog version of the benchmark and in the multiple {\tt (S)QCLP} versions, by {\tt clpr} constraints.
In any case, all the program samples used for this benchmarks in this subsection can be found in the folder {\tt benchmarks/} of the {\tt (S)QCLP} distribution.

Finally, we proceeded to solve the same goals for every version of the benchmark programs, both in {\em SICStus Prolog} and in {\tt (S)QCLP}.
The benchmark results can be found in Table \ref{table:results}. All the experiments were performed in a computer with a Intel(R) Core(TM)2 Duo CPU at 2.19GHz and with 3.5 GB RAM.

\begin{table}[ht]
\caption{Time overload factor with respect to Prolog}
\label{table:results}
\begin{minipage}{\textwidth}
\begin{tabular}{lrrrrrr}
\hline\hline
Program &
Q(b)\footnote{$\qclp{\B}{\rdom}$ version (i.e. the program does not have the {\tt \#prox} directive).} &
Q(u)\footnote{$\qclp{\U}{\rdom}$ version (i.e. the program does not have the {\tt \#prox} directive).} &
PQ(b)\footnote{$\sqclp{\sid}{\B}{\rdom}$ version.} &
PQ(u)\footnote{$\sqclp{\sid}{\U}{\rdom}$ version.} &
SQ(b)\footnote{$\sqclp{\sid}{\B}{\rdom}$ version with directive {\tt \#optimized\_unif}.} &
SQ(u)\footnote{$\sqclp{\sid}{\U}{\rdom}$ version with directive {\tt \#optimized\_unif}.} \\
\hline
naivrev & 1.80  & 10.71 & 4289.79 & 4415.11 & 56.22& 65.75\\
deriv   & 1.94  & 10.60 &  331.45 &  469.67 & 29.63 & 39.32 \\
qsort   & 1.05  & 1.11 & 135.59  & 136.98 & 2.51 & 2.83\\
query   & 1.02  & 1.12 & 7.17  & 7.13 & 3.80 & 3.88\\
\hline\hline
\end{tabular}
\vspace{-2\baselineskip}
\end{minipage}
\end{table}

The results in the table indicate the slowdown factor obtained for each version of each program.
For instance, the first column indicates that the time required for evaluating the goal corresponding to the sample program {\em naivrev} in $\qclp{\B}{\rdom}$ is about 1.80 times the required time for the evaluation of the same goal in Prolog. Next we discuss the results:
\begin{itemize}
\item {\em Influence of the qualification domain.}
      In general the difference between the slowdown factors obtained for the two considered
      qualification domains is not large. However, in the case of QCLP-programs {\em naivrev} and  {\em deriv}  the difference increases notably.
      This is due to the different ratios of the $\B$-constraints w.r.t. the program and
      $\U$-constraints w.r.t. the program. It must be noticed that the transformed programs
      are the same in both cases, but for the implementation of {\tt qval} and {\tt qbound}
      constraints, which is more complex for $\U$ as one can see in Subsection \ref{sec:practical:SQCLP}.
      In the case of {\em naivrev} and  {\em deriv} this makes a big difference because
      the number of computation steps directly required by the programs is much smaller
      than in the other cases. Thus the slowdown factor becomes noticeable for the qualification domain $\U$ in computations which require a large number of steps.

\item {\em Influence of the proximity relation.}
The introduction of a proximity relation---even the identity---is very significative,
since unification in the original {\tt Prolog} program is handled by calls to the predicate $\sim'$ in the SQCLP program.
This is particularly relevant  when the computation introduces large constructor terms, as in the case of {\em naivrev} which deals with
{\tt Prolog} lists. The efficient Prolog unification is replaced by an explicit term decomposition.
%The introduction of a proximity relation, even
%of empty, is very significative. This is due to the introduction of the predicate $\sim$,
%which replaces Prolog unification. The situation even worsens when the computation
%introduces large constructor terms, as in the case of {\em naivrev} which deals with
%Prolog lists. The efficient Prolog unification is replaced by an explicit term decomposition.

\item \label{review1:item30} {\em Influence of the optimized unification.} As seen in the table, the use of the program directive {\tt \#optimized\_unif} causes a clear increase in the efficiency of goal solving for these examples.
This is due to  the use of the implementation {\bf (C)} for the predicate $\sim'$ instead of the implementation {\bf (B)} (see Subsection \ref{sec:practical:SQCLP}).
The speed-up is especially noticeable when large data structures are involved in the unification as can be seen for the sample programs {\em naivev} and {\em deriv}. The reason is that the implementation {\bf (C)} avoids costly term decompositions required by the other implementation.
\end{itemize}
  %\ref{sec:practical}

% Section 6: Conclusions
% ----
% Section 6: Conclusions
% ----
\section{Conclusions}
\label{sec:conclusions}

% Summary of results

In our recent work \cite{RR10} we extended the classical CLP scheme to a new programming scheme SQCLP whose instances $\sqclp{\simrel}{\qdom}{\cdom}$ were parameterized by a proximity relation $\simrel$, a qualification domain $\qdom$ and a constraint domain $\cdom$.
This new scheme offered extra facilities for dealing with expert knowledge representation and flexible query answering.
In this paper we have set the basis for a practical use of SQCLP by providing
a prototype implementation on top CLP($\rdom$) systems like {\em SICStus Prolog} and {\em SWI-Prolog},
based on semantically correct program transformation techniques and supporting several  interesting instances of the scheme.

The transformation techniques presented in Section \ref{sec:implemen} work over programs and goals in two steps, formalized as the composition of two  transformations:  elim$_\simrel$ and elim$_\qdom$. Our mathematical results show that elim$_\simrel$ replaces the explicit use of a proximity relation by
using just qualification values and clause annotations, which are in turn replaced by purely CLP computations thanks to  elim$_\qdom$.
The composed effect of the two transformations ultimately enables to solve goals for SQCLP programs by applying any capable CLP goal solving system
to their CLP translations. 

%The two-step transformation technique presented in Section \ref{sec:implemen} has provided us with the needed theoretical results for effectively
% showing how proximity relations can be reduced to qualifications and clause annotations by means of the transformation elim$_\simrel$; and how 
% qualifications and clause annotations can be reduced to classical CLP programming by means of the transformation elim$_\qdom$.
% These two transformations altogether, ultimately enables the use of the classical mechanism of SLD resolution to obtain computed answers for SQCLP 
% goals w.r.t SQCLP programs, via their equivalent CLP programs and goals and the computed answers obtained from them by any capable CLP goal 
% solving procedure.

The prototype implementation presented in Section \ref{sec:practical} relies on the transformation techniques, improved with some optimiztions.
It has finally allowed us to execute all the examples shown in this paper---and in previous ones---, and a series of benchmarks for measuring the overload actually introduced by proximity relations---or by similarity relations---and by clause annotations and qualifications.
While we are aware that the prototype implementation presented in this paper has to be considered a research tool (and as such, we  admit that it cannot be used for industrial applications), we think that it can contribute to the field as a quite 
solid implementation of an extension of CLP($\rdom$) with proximity relations and qualifications.

Some related implementation techniques and systems have been presented in the Introduction.
However, as far as we know, no other implementation in this field has ever provided simultaneous support for proximity (and similarity) relations, qualifications via clause annotations and CLP($\rdom$) style programming.
Moreover, the development of our prototype has used both semantically correct methods and careful optimizations, 
aiming at a balance between theoretical foundations  and a sound but practical system.

%Moreover, our results in Section \ref{sec:implemen} on the semantic correctness of our implementation technique are in our opinion another 
% contribution of this paper which has no counterpart in related approaches.

% Future work

In the future, and taking advantage of the prototype system we have already developed, we plan to investigate possible applications which can profit from proximity relations and qualifications, such as in the area of flexible query answering.
In particular, we plan to investigate application related to flexible answering of queries to XML documents, in the line of \cite{CDG+09} and other related papers.

As support for practical applications, we also plan to increase the repertoire of constraint and qualification domains which can be used in the {\tt (S)QCLP} prototype, adding the constraint domain $\mathcal{FD}$ and the qualification domain $\W_d$ defined in Section 2.2.3 of \cite{RR10TR}.
On a more theoretical line, other possible lines of future work include:
a) investigation of unification modulo a given proximity relation $\simrel$, not assuming transitivity for $\simrel$ and proving soundness and completeness properties for the resulting unification algorithm;
b) building upon (a), extension of the SLD($\qdom$) resolution procedure presented in \cite{RR08} to a SQCLP goal solving procedure able to work with constraints and a proximity relation, including also soundness and completeness proofs;
and c) extension of the QCFLP ({\em qualified constraint functional logic programming}) scheme in \cite{CRR09} to work with a proximity relation and higher-order functions, as well as the implementation of the resulting scheme in the CFLP($\cdom$)-system {\sf Toy} \cite{toy}.
 %\ref{sec:conclusions}

% Acknowledgements
\section*{Acknowledgements}
The authors would like to thank to the anonymous reviewers, whose detailed and constructive comments and suggestions helped to revise and  improve the paper;
and to Jes\'us Almendros for pointing to bibliographic references in the area of flexible query answering.

\newpage

% References
\bibliographystyle{acmtrans}

\end{document}